\documentclass[11pt,letterpaper]{article}
\pdfoutput=1

\usepackage[numbers,sort&compress]{natbib}
\usepackage{mciteplus}
\usepackage{amsmath}
\usepackage{amsthm}
\usepackage{amsfonts}
\usepackage{amssymb}
\usepackage{braket}
\usepackage{tikz-cd}
\usepackage{tikz}
\usetikzlibrary{arrows.meta, calc}

\usepackage[letterpaper, margin=1.25in]{geometry}

\newcommand\numberthis{\addtocounter{equation}{1}\tag{\theequation}}
\let\oldket\ket
\renewcommand{\ket}[1]{{\oldket{#1}}}

\newtheorem{lemma}{Lemma}
\newtheorem{theorem}{Theorem}

\newtheorem{postulateletter}{Postulate}
\newenvironment{postulatecustom}[1]{
  \renewcommand{\thepostulateletter}{#1}
  \begin{postulateletter}
}{
  \end{postulateletter}
}

\usepackage{titlesec}
\titleformat{\part}[hang]
  {\normalfont\Large\bfseries}
  {\partname~\thepart:}
  {0.5em}
  {}

\usepackage{hyperref}
\hypersetup{
    colorlinks=true,
    linkcolor=blue,
    citecolor=blue,
    urlcolor=blue,
    pdftitle={On the relationship between the POVM framework and the no-signaling hypothesis},
    pdfauthor={Will Cavendish}
}

\title{On the relationship between the POVM framework and the no-signaling hypothesis}

\author{Will Cavendish\thanks{John Bell Institute for the Foundations of Physics. Email: \href{mailto:willcavendish@johnbellinstitute.org}{willcavendish@johnbellinstitute.org}. ORCID: 0009-0005-7578-756X.}}
\date{}

\allowdisplaybreaks[1]

\begin{document}
\maketitle

\begin{abstract}
The quantum foundations literature contains dozens of ``no-signaling'' arguments that rely on the Born rule or, more generally, the positive operator-valued measure (POVM) framework to show that the resources of non-relativistic quantum mechanics (NRQM) cannot be used to send superluminal signals. What these arguments typically establish, however, is not a prohibition on superluminal signaling per se, but the seemingly stronger claim that signaling at any speed between mutually isolated systems is forbidden. This paper gives a precise formulation of this claim, the \emph{no-signaling hypothesis}, and investigates the following converse question: under what additional assumptions does it imply that all experimental outcome probabilities are described by POVMs? Surprisingly few are needed: given the standard NRQM representation of composite experiments and two textbook claims about marginal probabilities, the no-signaling hypothesis is equivalent to the universality of the POVM framework. More broadly, by formalizing the relevant concepts and deriving the logical relationships between them, the paper shows that the no-signaling hypothesis can be used to derive most, though not all, of the probability-generating fragment of NRQM. The resulting map of implications clarifies which common postulates of NRQM are logically interrelated and which can be modified independently.
\end{abstract}

\section{Introduction}

Non-relativistic quantum mechanics (NRQM) and special relativity (SR) were developed to describe different empirical phenomena in different regimes. It is therefore unsurprising that their basic mathematical formulations are in tension. What is arguably more surprising is that the two theories can be reconciled at all. As its name suggests, NRQM makes use of ``non-relativistic'' resources. For example, NRQM describes multi-particle states using wave functions defined on configuration space, the points of which refer (or appear to refer) to the simultaneous positions of particles that can be arbitrarily far apart. SR, by contrast, is the paradigmatic ``relativistic'' theory, and the relativity of simultaneity is among its central tenets. Yet the prevailing view is that these tensions are only apparent, and that what Shimony called the ``peaceful coexistence'' \cite{shimony1984controllable} of NRQM and SR reflects deep physical principles. This paper investigates the quantum-mechanical basis of this compatibility. Which features of NRQM, if we gave them up, would irreparably disturb the peace?

It is widely believed that the answer to this question has something to do with ``signaling.'' On this view, what underwrites the compatibility of NRQM and SR is that both theories, for reasons of their own, ``forbid signaling'' in the relevant sense. But what is the relevant sense? In SR, it is \emph{superluminal} signaling that is assumed to be problematic. NRQM, however, has a different theoretical vocabulary in which the speed of light plays no explicit role. As Kennedy points out, many ``no-signaling'' results in the literature\footnote{The specific arguments Kennedy has in mind are those in \cite{bussey1982super, jordan1983quantum, ghirardi1980general, shimony1984controllable, redhead1987incompleteness, eberhard1989quantum, scherer1993problem}.} rely on assumptions that are surprisingly agnostic about the underlying physics:
\begin{quote}
The proofs do not involve Planck's constant, the de Broglie relations, nor the Schr\"odinger equation, i.e., none of the ruling features of quantum theory. Likewise, they do not involve the speed of light. Signalling between the correlated subsystems at any speed is prohibited. \cite{kennedy1995empirical}
\end{quote}
This should be puzzling, because we can and do send \emph{subluminal} signals. NRQM would not be empirically plausible if it ruled them out. The first order of business in addressing this puzzle is to say, in terms that NRQM has the resources to interpret, what sort of ``signaling'' ought to be avoided.

As I will argue, ruling out the relevant sort of ``signaling'' requires a highly structured theory of ``observability.'' Indeed, once enough features of the ``standard NRQM package'' are in place, the \emph{assumption} of no-signaling can be used to prove the generalized Born rule, i.e., that all experimental outcome probabilities for quantum-mechanical prepare-and-measure experiments are described by positive operator-valued measures (POVMs). I will refer to this claim as POVM universality (see Postulate \ref{POVM} below). Since the generalized Born rule in turn implies no-signaling, the two are, modulo auxiliary hypotheses, equivalent (see Theorem \ref{mainthm}). The goal of this paper is to give a rigorous, general account of this equivalence, and a full inventory of the assumptions it requires.

The main difficulty in providing this account is not supplying the proofs, but stating the claim and the relevant definitions clearly enough to make rigorous deductions possible. Once concepts like ``signaling'' or ``prepare-and-measure experiment'' have been formalized, the rest is just math, in the sense that the correctness of the proofs can be evaluated without reference to what the formalism is meant to represent. Whether such deductions are physically meaningful, however, depends on the extent to which the postulates correctly formalize the intended physical concepts. This is something that theorists disagree about, and such disagreements bear directly on the claims at issue (see, e.g., \cite{masanes2019measurement, kent2025measurement, masanes2025response, stacey2024contradictions}). As we will see, several common mathematical postulates (e.g., the use of density matrices to represent proper mixtures, see Section \ref{sec:ProperMixture}) encode much stronger assumptions than they appear to. So what I choose to assume (or, if I am less careful than I intend to be, what I implicitly assume by accident) is at least as important as any deductions built on these assumptions.

For this reason, the bulk of this paper (Part \ref{part:map}) is dedicated to laying out how I have formalized the relevant concepts. My choices are all ``standard,'' in the sense that they are adapted from formulations in the existing literature, or (as is the case with Postulate \ref{PL}) unconditionally equivalent to standard assumptions. I do not intend to argue that the postulates in question are unique or inevitable on physical grounds, only that their ``standardness'' makes them both worth taking seriously and worthy of scrutiny. Figure \ref{Implications} maps the logical relationships between them.

\begin{figure}[hpt]
\centering
\begin{tikzpicture}[
    state/.style={align=center, text width=2.5cm},
    arrow/.style={-{Implies[length=3mm, width=3mm]}, double distance=1.5pt, thick, black},
    bi_arrow/.style={{Implies[length=3mm, width=3mm]}-{Implies[length=3mm, width=3mm]}, double distance=1.5pt, thick, black},
    bg_post/.style={align=center}, 
    annot_arrow/.style={-stealth, thin, gray, bend right=15}
]

    \draw[dotted, thick] (-5.5,3.5) rectangle (9.5,-6);
    \draw[dotted, thick] (1.0,3.5) -- (1.0,-6);

    \node[font=\large\scshape] at (-2.25,3.1) {Composite systems};
    \node[font=\large\scshape] at (5.1,3.1) {Single systems};

    \node[state, fill=white, inner sep=1pt] (M) at (1.1, -0.3) {Existence \\of Qubit \\Composites\\ {\color{blue}(\ref{EQC})}};
    \node[state] (A) at (-4,2) {No Signaling \\{\color{blue}(\ref{NS})}};
    \node[state] (B) at (2.5,2) {Parallelogram \\ Law  {\color{blue}(\ref{PL})}};
    \node[state] (C) at (6.5, 2) {Weak PL \\ {{\color{blue}(\ref{WPL})}}};
    \node[state] (D) at (6.5,-2) {Quadraticity \\ {\color{blue}(\ref{QF})}};
    \node[state] (E) at (2.5,-2) {Weak SEP {\color{blue}(\ref{WSEP})}};
    \node[state] (F) at (-4,-2) {Partial Trace \\ sufficiency \\{\color{blue}(\ref{PT})}};
    \node[state] (G) at (6.5, -5) {POVM \\ universality \\ {\color{blue}(\ref{POVM})}};
    \node[state] (H) at (2.5, -5) {Standard\\ {$\mathrm{Ensemble}{\text -}\mathrm{level}$}\\Probabilities {\color{blue}(\ref{SEP})}};
    \node[state] (J) at (-2.0, 0.2) {{$\mathrm{Standard~Marginal}$} \\ {$\mathrm{decompositions~for}$}\\ {$\mathrm{Qubit~Composites}$}\\ {\color{blue}(\ref{SMQC})}};
    \node[state] (I) at (-0.7, -3.6) {$\mathrm{Independence~of}$\\ $\mathrm{Product~states}$\\
    {\color{blue}(\ref{IP})}};
    \node[state] (K) at (-1.0, 2) {Weak NS\\{\color{blue}(\ref{WNS})}};
    \node[state] (L) at (-3.5, -5.0) {Product-state \\Multiplicativity {\color{blue}(\ref{PM})}};

    \draw[arrow, shorten <=-2mm, shorten >=-6mm ] (A) to coordinate[midway] (AK_mid) (K);
    \draw[arrow, shorten <=-5mm, shorten >=-2mm ] (K) to coordinate[midway] (KB_mid) (B);
    \draw[arrow, shorten >= -6mm ] (B) to (C);
    \draw[arrow] (B) to coordinate[midway] (BC_mid) (D);
    \draw[arrow] (C) to coordinate[midway] (CD_mid) (D);
    \draw[arrow, shorten <=-2mm, shorten >=-4mm] (D) to coordinate[midway] (DE_mid) (E);
    \draw[arrow] (E) to coordinate[midway] (EB_mid)  (B);
    \draw[arrow, shorten <=-6mm, shorten >=-4mm] (E) to coordinate[midway] (EF_mid) (F);
    \draw[arrow] (F) to coordinate[midway] (FA_mid) (A);
    \draw[arrow, shorten <=-1mm, shorten >=-2.5mm] (F) to coordinate[midway] (FJ_mid) ($(J)+(-0.95,-0.85)$);
    \draw[arrow, shorten <=-2mm, shorten >=-3mm] (D) to coordinate[midway] (DH_mid) (H);
    \draw[arrow] (H) to (E);
    \draw[bi_arrow, shorten <=-0mm, shorten >=-0mm] (D) to coordinate[midway] (DG_mid)(G);
    \draw[arrow, shorten <=-1mm, shorten >=-1mm] (L) to [bend right=3] coordinate[midway] (LI_mid) ($(I)+(-0.6,-0.4)$);

    \node[draw, circle, fill=white, inner sep=1pt, minimum size=6mm, font=\scriptsize, xshift=-0.2cm] at (KB_mid) {{\color{blue}\ref{sec:WNStoPL}}};
    \node[draw, circle, fill=white, inner sep=1pt, minimum size=6mm, font=\scriptsize] at (BC_mid) {{\color{blue}\ref{sec:PLtoQF}}};
    \node[draw, circle, fill=white, inner sep=1pt, minimum size=6mm, font=\scriptsize] at (DE_mid) {{\color{blue}\ref{sec:QFtoWSEP}}};
    \node[draw, circle, fill=white, inner sep=1pt, minimum size=6mm, font=\scriptsize] at (DH_mid) {{\color{blue}\ref{sec:QFtoSEP}}};
    \node[draw, circle, fill=white, inner sep=1pt, minimum size=6mm, font=\scriptsize] at (DG_mid) {{\color{blue}\ref{sec:POVM}}};
    \node[draw, circle, fill=white, inner sep=1pt, minimum size=6mm, font=\scriptsize] at (EB_mid) {{\color{blue}\ref{sec:WSEPtoPL}}};
    \node[draw, circle, fill=white, inner sep=1pt, minimum size=6mm, font=\scriptsize] at (EF_mid) {{\color{blue}\ref{sec:WSEPtoPT}}};
    \node[draw, circle, fill=white, inner sep=1pt, minimum size=6mm, font=\scriptsize] at (FJ_mid) {{\color{blue}\ref{sec:PTtoSQ}}};
    \node[draw, circle, fill=white, inner sep=1pt, minimum size=6mm, font=\scriptsize] at (FA_mid) {{\color{blue}\ref{sec:PTtoNS}}};
    \node[draw, circle, fill=white, inner sep=1pt, minimum size=6mm, font=\scriptsize] at (CD_mid) {{\color{blue}\ref{sec:WPLtoQF}}};
    \node[draw, circle, fill=white, inner sep=1pt, minimum size=6mm, font=\scriptsize] at (LI_mid) {{\color{blue}\ref{sec:PPI}}};
    \node[bg_post] (dim) at (8.25,0.75) {$\dim(\mathcal{H})\ge 3$};

    \draw[{Circle[length=2pt]}-Stealth] (-0.6,1.2) to [bend left=2] (0.2,1.7);
    \draw[{Circle[length=2pt]}-Stealth] (7.8, 0.5) to [bend left=10] (7.0,0);
    \draw[{Circle[length=2pt]}-Stealth] (-0.9, -2.9) to [bend right=0] (-0.8,-2.4);
    \draw[{Circle[length=2pt]}-Stealth] (-1.1,-2.9) to [bend right=5] (-2.8,-1.2);
    \draw[{Circle[length=2pt]}-Stealth] (0.85, 0.65) to [bend right=5] (0.6,1.6);
    
\end{tikzpicture}
\caption{A map of the logical relationships between 13 postulates. Acronyms in blue link to the section of the paper in which the given postulate is formalized. A thin arrow running from a background assumption to an implication indicates that the implication depends on that assumption, e.g., the proof of $\mathrm{\ref{WPL}}\Rightarrow\mathrm{\ref{QF}}$ applies only to systems whose Hilbert space $\mathcal{H}$ satisfies $\dim(\mathcal{H}) \ge 3$. A circled number over an implication arrow indicates the section in which it is proved.}
\label{Implications}
\end{figure}
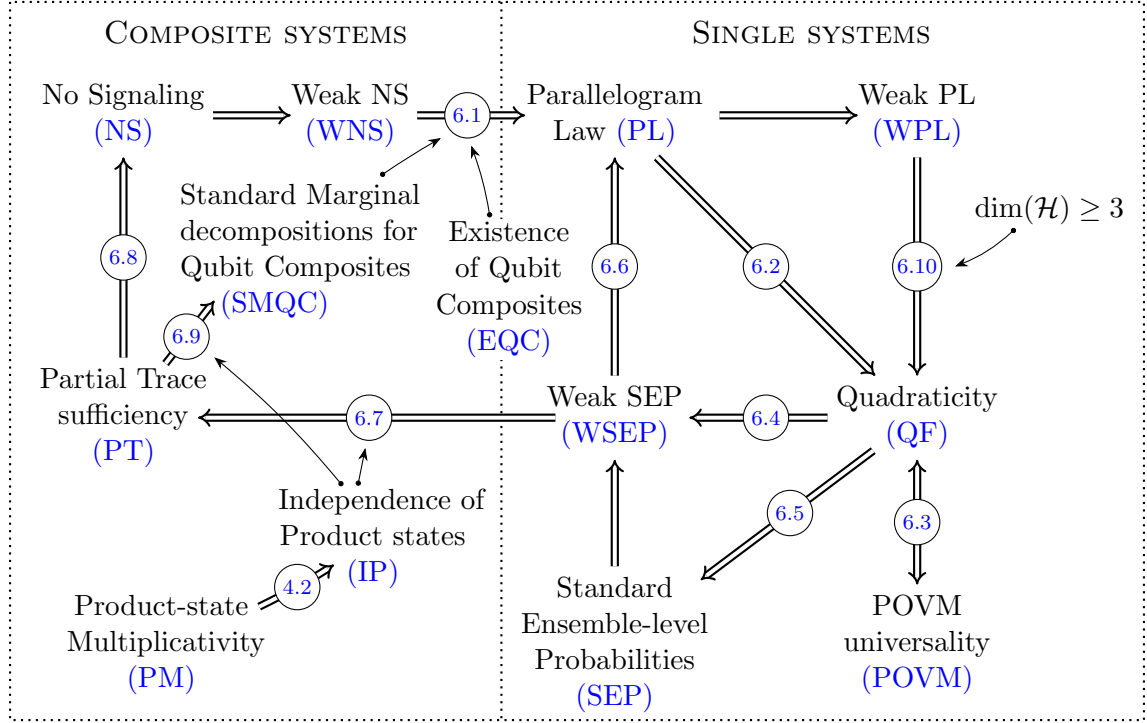 

Many of the proofs in this paper are variations on arguments appearing elsewhere in the literature (e.g., \cite{gisin1989stochastic, svetlichny1998quantum, simon2001no, caves2004gleason, barnum2005influence, nenashev2014blackbox}). The arguments that may be less familiar are those depending on the parallelogram law (Postulates \ref{PL} and \ref{WPL} below), which draw on results of Gleason \cite{gleason1957}, and those concerning general proper mixtures, which require Bochner integration. Once the results shown in Figure \ref{Implications} are proved, many further results can be read off the diagram by following directed paths of implications. Theorem \ref{mainthm}, which gives a precise formulation of the equivalence between the no-signaling hypothesis and POVM universality, is one such result (see Section \ref{sec:proofOfMainTheorem}).\footnote{In Theorem \ref{mainthm}, and elsewhere in this paper, I will use the term ``textbook NRQM'' to loosely indicate the claims, formalism, and methods presented in a typical introductory undergraduate course in quantum mechanics. The contents of Griffiths's textbook \cite{griffiths2005introduction} are what I have in mind.}

\begin{theorem}\label{mainthm} The no-signaling hypothesis (Postulate \ref{NS}) is equivalent to POVM universality (Postulate \ref{POVM}), provided that textbook spin measurements on qubits can be composed with any available experiment on an isolated second system (Postulate \ref{EQC}), and the following two empirical claims of textbook NRQM hold:
\begin{itemize}
\item When Alice and Bob perform isolated experiments on a joint system prepared in a product state, Bob's marginal outcome probabilities are those his experiment assigns to his own state, whatever Alice's state or choice of experiment (Postulate \ref{IP}).
\item In the special case that Alice performs a spin measurement on a qubit, Bob's marginal outcome probabilities can be written as a sum over Alice's outcome probabilities in the standard way (Postulate \ref{SMQC}).
\end{itemize}
\end{theorem}

Given this network of implications, it is natural to wonder whether the single-system postulates in Figure \ref{Implications}, together with Postulate \ref{IP}, suffice to characterize the probability-generating fragment of textbook NRQM. In Section \ref{sec:PMIndep}, I show this is not the case by translating Popescu and Rohrlich's PR-box \cite{popescu1994quantum} example into the present framework. This yields a model whose CHSH value exceeds Tsirelson's bound and that satisfies every postulate in Figure \ref{Implications} except Postulate \ref{PM}.

The paper is organized as follows. Part \ref{part:map} opens with an analysis of Kennedy's puzzle. What signaling worries does NRQM introduce, and how do ``no-signaling'' arguments address them? Sections \ref{sec:PandM} and \ref{sec:compositeExps} then formalize the postulates mapped in Figure \ref{Implications}, and Section \ref{sec:discussion} discusses their significance. The proofs of the implications in Figure \ref{Implications} (Section \ref{mathImp}) and the construction of the PR-box model (Section \ref{sec:PMIndep}) are collected in Part \ref{part:proofs}.

\part{No-Signaling Arguments and Their Formalization}\label{part:map}

\section{Signaling Scenarios in NRQM}\label{sec:background}

To characterize the type of ``signaling'' no-signaling arguments aim to rule out, it helps to look closely at a standard ``signaling'' scenario. After doing this in the next subsection, I discuss how different ways of modifying or extending NRQM relate to the worries this scenario raises. This discussion motivates an alternative scenario, which I call the ``bunker game'' (Section \ref{sec:TheBunkerGame}), involving attempts to signal between mutually isolated systems at subluminal speeds. The bunker game separates questions about spacetime structure from more general signaling concerns and underlies my approach to isolating and formalizing the relevant fragment of standard quantum theory. Finally, I discuss the core features of my approach, highlighting which parts of textbook NRQM I will take as background assumptions (e.g., Hilbert space structure), and which I will promote to the foreground postulates appearing in Figure \ref{Implications}.

\subsection{The textbook signaling scenario}\label{subsection:whatworry}

In this subsection, I present the textbook argument against signaling in NRQM (see, e.g., \cite[Sec.~12.2]{griffiths2005introduction}), taking its description of the scenario at face value. The aim is to understand in detail why signaling is impossible \emph{according to textbook NRQM}. The scenario I will consider involves two experimenters, Alice and Bob, who attempt to communicate. I will assume that Alice attempts to signal \emph{to} Bob, so they will play asymmetric roles in this story (if not in the underlying physics). As discussed by Maudlin, ``when we characterize one system as a transmitter and another as a receiver of a signal, we require that a \emph{controllable} aspect of the transmitter be correlated with some \emph{observable} aspect of the receiver'' \cite{maudlin2011quantum}. We will therefore be concerned with what Alice can \emph{do}, and what Bob can \emph{infer}.

In the version of the story I will tell, Alice and Bob each control a single particle, and are separated by some arbitrary distance $d_{AB}$. I will also assume that, once the pair of particles has been prepared, any interactions between them are negligible. This is represented in NRQM by the absence of interaction terms in the Hamiltonian.  At some predetermined time $t_{\mathrm{transmit}}$, Alice performs one of several available actions on her particle. Subsequently (and note that ``subsequently'' has a clear meaning, since NRQM is a \emph{non-relativistic} theory), Bob performs an experiment on his particle and observes the outcome. No-signaling arguments typically assume that Alice can choose what to \emph{do} to her particle, and that Bob can choose what to \emph{observe} about his, within the bounds of what NRQM allows. These choices are usually treated as free variables in Bell's sense \cite{bell1985exchange}. The goal of a no-signaling argument is to show that there is no protocol by which Bob can learn anything about Alice's actions.

As Kennedy's remark suggests, these arguments seem to ask NRQM to forbid something physically possible. After all, in most familiar circumstances, Bob can learn what Alice did by going to her lab and asking her. Worries about \emph{superluminal} signaling are, implicitly, worries about timing. The question is not whether Bob can learn what Alice did, but whether he can learn it before a signal from Alice could reach him. To make this precise, let $t_{\mathrm{receive}}$ be the time at which Bob infers what Alice did, so that $\Delta t=t_{\mathrm{receive}}-t_{\mathrm{transmit}}$ is the time that elapses between Alice's action and Bob's inference. If $\Delta t$ were bounded above by a value that did not depend on $d_{AB}$, this would threaten \emph{superluminal} signaling, since the protocol would work even if $d_{AB}$ were larger than $c\,\Delta t$, where $c$ is the speed of light. With Alice and Bob sufficiently far apart, a $d_{AB}$-independent protocol would let Bob learn what Alice did faster than \emph{anything} traveling at any fixed finite speed could traverse the space between them. This is why no-signaling arguments, which rule out such protocols, never need to invoke $c$, and end up prohibiting signaling ``at any speed.''

What features of NRQM suggest that such protocols might be available? The first ingredient concerns how NRQM represents the wave function of a pair of particles. The canonical example is Bohm's version of the EPR experiment \cite{bohm1955quantum}, which involves the singlet state for a pair of spin-1/2 particles. This state is typically given the mathematical representation $\frac{1}{\sqrt{2}}\big(\ket{\uparrow}\!\otimes\!\ket{\downarrow}-\ket{\downarrow}\!\otimes\!\ket{\uparrow}\big)$, where one imagines that Alice's particle is represented by the left-hand tensor factor, and Bob's particle is represented by the right-hand tensor factor. An important feature of this representation is that space and time do not enter into it at all. This means that, to the extent this representation is adequate for describing our scenario, $d_{AB}$-independence is guaranteed. One might object that the spin-only expression is merely shorthand for a richer representation that does mention space and time. If Alice and Bob's particles were each in the ground state of identical traps with energy $E_0$, for instance, their joint state might be represented by
$$
\Psi(q_A,q_B,t)=
\tfrac{1}{\sqrt{2}}\,e^{-2iE_0t/\hbar}\left[\left(\!\begin{smallmatrix} \psi(q_A)\\ 0\end{smallmatrix}\!\right)\!\otimes\!\left(\!\begin{smallmatrix} 0 \\ \phi(q_B) \end{smallmatrix}\!\right) -\left(\!\begin{smallmatrix} 0 \\ \psi(q_A)\end{smallmatrix}\!\right)\!\otimes\!\left(\!\begin{smallmatrix} \phi(q_B) \\ 0 \end{smallmatrix}\!\right)\right]\,,
$$
where $|\psi(q_A)|^2$ and $|\phi(q_B)|^2$ give the position distributions of the two particles, concentrated in Alice and Bob's labs respectively. This representation \emph{does} contain information about $d_{AB}$. But NRQM tells us that the spin-only representation suffices for predictions involving only spin, so at least in this case there is an empirically adequate representation of the state that exhibits the sort of $d_{AB}$-independence at issue.

The second ingredient is what NRQM tells us about how to represent the joint state after Alice performs a spin measurement. The textbook Born rule tells us the probability that Alice obtains a given outcome as a function of the joint state before her measurement. Then, the collapse postulate (or L\"uders rule in the case of eigenspaces of dimension greater than $1$) tells us how to represent the joint state after her measurement given the outcome she obtains. The post-measurement state is, in turn, what we use to compute the probabilities of any subsequent measurement on the pair of particles. Read literally, the collapse postulate says that Alice's measurement affects the quantum state instantaneously. Since this joint state describes Bob's particle, this change has consequences for Bob.

The third ingredient is that the post-measurement joint state depends on which measurement Alice chooses. The formalism associates Alice's choice of measurement with an orthonormal basis for the Hilbert space $\mathcal{H}_A$ describing her particle, which is $\mathbb{C}^2$ in this case. For simplicity, we can imagine restricting Alice to one of two measurement choices, one described by $\mathcal{B}_z=\{\ket{\uparrow},\ket{\downarrow}\}$ and the other by $\mathcal{B}_x=\{\ket{+},\ket{-}\}$, where $\ket{\pm}=\frac{1}{\sqrt{2}}\big(\ket{\uparrow}\pm \ket{\downarrow}\big)$. If she chooses the measurement represented by $\mathcal{B}_z$, the collapse rule says that the joint state of the system after her measurement is described by either $\ket{\uparrow}\!\otimes\!\ket{\downarrow}$ or $\ket{\downarrow}\!\otimes\! \ket{\uparrow}$. In either case, the joint state is a product state, and for product states NRQM tells us that, at least for the purposes of describing Bob's particle alone, we can forget about the state of Alice's particle entirely by dropping the left-hand tensor factor. This means that we can meaningfully describe Bob's particle by a state of its own after Alice's measurement, namely an element of $\{\ket{\uparrow},\ket{\downarrow}\}$ if she measures in $\mathcal{B}_z$, and of $\{\ket{+},\ket{-}\}$ if she measures in $\mathcal{B}_x$.

We can now see what theorists who prove no-signaling theorems have to worry about. According to NRQM, the state Bob's particle ends up in depends on Alice's choice of measurement. Of course, the randomness of Alice's outcome means she cannot control which of the possible states Bob's particle is left in. In the case of the singlet state, each outcome occurs with probability $\frac{1}{2}$ regardless of the basis Alice chooses. However, she can control the fact \emph{that} Bob's particle ends up in $\{\ket{\uparrow},\ket{\downarrow}\}$, by measuring in $\mathcal{B}_z$. By choosing the measurement associated with $\mathcal{B}_x$, she instead ensures that it lies in $\{\ket{+},\ket{-}\}$. The control Alice has in this situation is sometimes referred to as ``steering,'' a term that goes back to Schr\"odinger, who noted that the theory tells us not only that the state of the other particle changes, but also that ``the state arrived at \emph{depends} quite decidedly on \emph{what} measurements one chooses to take---not only on the results they yield'' \cite{Schrodinger1936} (emphasis original).

At this point in the story, it seems that our only hope of avoiding signaling is for Bob to be unable to distinguish between the two processes Alice can bring about. Process Z leaves his particle in either $\ket{\uparrow}$ or $\ket{\downarrow}$ with equal probability, and Process X leaves it in $\ket{+}$ or $\ket{-}$ with equal probability. What we now need is a theory of what Bob can \emph{observe} about his particle, and what sort of inferences such observations license. This is the climax of our story. Can NRQM prevent Bob from learning what Alice has done? Note that the pressure has landed on epistemic features of single-system NRQM. After Alice's measurement, she drops out of the picture, and the remaining question is what Bob alone can learn about an unentangled particle well-localized to his lab. 

As it happens, NRQM's highly structured theory of observables restricts what Bob can learn in just the way that is needed. For any observable $M_B$ Bob chooses to measure, a particle produced by Process Z and a particle produced by Process X yield the same outcome probabilities for $M_B$. Since everything Bob can observe is, by assumption, described by an observable, he cannot tell which process Alice has brought about. No-signaling arguments generalize this story to any joint initial state $\Psi$, any choice of measurement by Alice, and any choice of observable by Bob (see, e.g., \cite{shimony1984controllable}).

So, according to NRQM, the peace is maintained. But only just, and for a surprising reason. One might expect that Alice's actions cannot influence the state of Bob's particle at all. However, this is not what NRQM says. The limit on what Alice can control is not enough to prevent signaling on its own. What ultimately does the work in the story above is the limit NRQM places on what Bob can observe.

\subsection{Which features of these scenarios are worth taking seriously?}\label{subsec:features}

The tale told in the last subsection used NRQM's mathematical resources to raise the possibility of signaling, then used other claims of NRQM to rule it out. Given that no one treats NRQM as a comprehensive final theory, what does this tell us about how NRQM can be generalized or modified?

The first thing to understand is that the concern raised by this scenario cannot simply be ``interpreted away.'' For instance, it might seem that the problem dissolves if we adopt an interpretation of NRQM on which collapse does not ``really happen.'' This move is, of course, unavailable to collapse theorists, who take collapse to be a real physical process. They certainly need to take the threat of signaling seriously, which may help explain why they have contributed extensively to the no-signaling literature \cite{ghirardi1979some, ghirardi1980general, ghirardi1983quantum, ghirardi1988experiments, ghirardi1989cp, bassi2008recent, ghirardi2012proposal, ghirardi2013entanglement}. However, even approaches to quantum theory that deny collapse face a version of the same problem. In explicitly nonlocal theories like Bohmian mechanics, no-signaling depends both on the quantum equilibrium hypothesis \cite{durr1992quantum}, without which signaling becomes possible \cite{valentini1991signal},
and on the claim that outcome statistics of experiments are always given by POVMs \cite{durr2004quantum}, either because an exact ``POVM theorem'' can be proved, as some argue \cite{tumulka2022foundations}, or because the claim holds ``to a very good approximation'' \cite{beck2025povm}. That non-POVM arrival-time statistics (see Section \ref{subsec:observability}) might permit signaling has been an explicit concern in the recent Bohmian literature \cite{goldstein2024arrival, amblard2026arrival}. 

Similar issues arise in other approaches. Some Everettians not only reject collapse, but also deny that Alice's choices change anything about Bob's subsystem at all \cite{page1982einstein}, which might seem to dissolve the concern. However, the relative state of Bob's particle, conditional on Alice's recorded outcome, depends on her choice of measurement. That relative state is what determines Bob's subsequent outcome probabilities in a given branch of the universal wave function. What needs to be ruled out is the existence of a protocol under which Alice and Bob tend to end up in branches in which Bob happens to receive the signal Alice attempts to send. 

Interpretations in which collapse reflects only a change in our information, rather than a physical change, face the same concern. In these interpretations, Alice's measurement updates our \emph{description} of Bob's particle, and the updated description, which fixes Bob's subsequent outcome probabilities, depends on her choice of measurement. If that description predicts that Bob receives Alice's message more often than chance, it predicts signaling, however it is interpreted. Nor does the concern go away if NRQM is regarded as merely an effective approximation to a relativistic theory such as QFT. To the extent that the deeper theory recovers NRQM's predictions in the relevant regime, it inherits the same problem. This is why no-signaling arguments are needed in QFT as well (see \cite{eberhard1989quantum}).

What must be explained in each setting is why Bob cannot distinguish the scenarios Alice can bring about. This requires an account of what Bob can distinguish, which is the role observables play in NRQM. Rules about what Bob can observe can simply be postulated, as in operational treatments that define observables as POVMs (e.g., \cite{busch1995operational}), or one can try to derive them from more primitive assumptions. Approaches differ in how they answer this question, but none can avoid it entirely if they grant that the empirical content of NRQM's description of the scenario is correct.

The formalism I adopt below invokes neither collapse nor conditional states. Its relevant postulate (Postulate \ref{SMQC}, Section \ref{subsec:semiorthodox}) is an empirical claim about the relationship between outcome probabilities for different experiments, and that claim alone makes signaling possible if Bob can observe too much.

\subsection{What can Bob observe?}\label{subsec:observability}

In NRQM's telling of the story, the measurement postulates describe both what Alice does and what Bob can observe. Alice's side, which triggers collapse, has received considerable scrutiny. Bob's side has received less, because it belongs to the part of the measurement postulates theorists tend to be comfortable with, namely, the ``probability-generating'' fragment, which is often identified simply with ``the Born rule.'' This fragment not only makes prescriptive claims about outcome probabilities but also restrictive claims about ``observability.'' The most important of these is that \emph{everything} Bob can observe is described by an ``observable,'' a word designed to carry exactly this connotation.

Many theorists can and do question this claim, not just because they are generally dissatisfied with the measurement postulates, but also because there are concrete empirical questions that are difficult, if not impossible, to address using the textbook framework. One illustrative example is the still-unsolved arrival-time problem, which asks how to describe time-of-flight data for quantum particles (see \cite{muga2000arrival}). Part of the difficulty is the apparent impossibility of defining a ``time operator.'' This point is often attributed to Pauli, who wrote a famous footnote claiming that ``the introduction of an operator $t$ is basically forbidden and the time $t$ must necessarily be considered as an ordinary number (`c'-number)'' \cite{pauli2012general}.

However one interprets Pauli's claim, it does not mean that we cannot \emph{physically} perform laboratory procedures that measure duration. We can and do. At most, it implies that the standard theory of self-adjoint operators cannot \emph{describe} such experiments, at least under certain assumptions about how that theory's resources are to be used. The challenge of describing the resulting data cannot be met by insisting that experimentalists do not ``really'' observe what they say they do. Such a claim may be theoretically defensible in a given interpretation, but it does not answer the empirical question the arrival-time problem poses.

Theorists working on the arrival-time problem try to extend the orthodox formalism in theoretically motivated ways, typically proposing descriptions that are \emph{not} given by self-adjoint operators. Often these proposals can be framed in terms of POVMs, but sometimes they go beyond the POVM framework altogether \cite{marchewka2001survival, ruggenthaler2005times, das2019arrival}. Given the story in Section \ref{subsection:whatworry}, however, any extension of the class of available ``observables'' threatens to introduce signaling. In the extreme case, if Bob had access to a ``state readout device'' that told him what his state was, he could easily tell what Alice did.\footnote{Cf.\ \cite{ghirardi1983quantum} for discussion of an early proposal to evade no-signaling arguments ``by widening the class of observables,'' and \cite{goldstein2024arrival, amblard2026arrival} for discussion of a recent example \cite{das2019arrival} from the arrival-time literature.} This is how a seemingly mundane empirical question comes to bear on signaling. Not every extension actually introduces signaling, however. POVMs, for instance, go beyond self-adjoint operators, yet the standard no-signaling arguments still apply to them. How can we tell which extensions of NRQM's theory of observables threaten to introduce signaling, and which are benign?

\subsection{What are the stakes?}\label{subsec:stakes}

The stakes of this question are widely thought to be extremely high. Many theorists worry that causality itself is threatened, since, as Griffiths memorably puts it, according to SR,
\begin{quote}
a \emph{causal} influence that propagated faster than light would carry unacceptable implications. For according to special relativity there exist inertial frames in which such a signal propagates \emph{backward in time}---the effect preceding the cause---and this leads to inescapable logical anomalies. (You could, for example, arrange to kill your infant grandfather. Not a good idea!) \cite{griffiths2005introduction}
\end{quote}
On the other hand, Kent has argued that it is not logically necessary to put this sort of burden on the theory of observables \cite{kent2005nonlinearity, kent2021quantum}. In his ``local state'' formalism, the description of spatially separated subsystems incorporates spacetime structure in a way that standard NRQM does not. This decouples questions about single-system observability from questions about superluminal signaling, so that even introducing quantum state readout devices would not allow faster-than-light signaling (see Section \ref{sec:TheBunkerGame}).

That spacetime structure and ``observability'' can be decoupled in this way should not be surprising, given the limited role spacetime notions played in the no-signaling argument of Section \ref{subsection:whatworry}. There, concerns about superluminality translated into concerns about the existence of $d_{AB}$-independent signaling protocols. But what if we set aside questions of signal speed and put Alice and Bob close together? In this setting, we can no longer rely on the decay of familiar potentials with distance to justify neglecting interactions between their systems, as standard no-signaling arguments typically do (see, e.g., \cite{ghirardi1980general}).\footnote{It has been argued that the reliance of no-signaling proofs on non-interaction makes them question-begging as defenses of the peaceful coexistence of NRQM and SR \cite{peacock1998edge}.} Still, particles that are close together can also be dynamically isolated from one another, or very nearly so. We can therefore ask whether the resources of NRQM can be used to rule out signaling between mutually isolated systems, even when they are close together.

Reframing the signaling question in terms of isolation rather than distance disentangles it from questions of spacetime structure. What is at issue is whether there can be ``signaling through opaque media'' \cite{kent2021quantum}. No-signaling arguments answer this question in the negative, a point Lindblad makes with particular clarity in reply to a proposal for superluminal signaling between the detectors on the two wings of an EPR experiment.\footnote{Lindblad's comment was one of several replies to the proposal of Datta, Home, and Raychaudhuri \cite{datta1987curious} (see also \cite{finkelstein1987cp, squires1988non}). Hall similarly locates the essential assumption in the independence of the measuring apparatus from the distant subsystem, concluding that ``a non-local effect can occur only if a non-local measurement is made'' \cite{hall1987imprecise}.}
\begin{quote}
In a non-relativistic theory such a signal is not forbidden, but it can exist only if there is an interaction between the two detectors. However, here the particles cease to interact after the creation of the pair, so there can be no such interaction if the detectors are independent, and hence there should be no signal whatever, regardless of velocity. \cite{lindblad1987comment}
\end{quote}
Unlike a prohibition on superluminal signaling, this ``no-signaling hypothesis'' can be formulated within NRQM proper,\footnote{Cf.~\cite{barnum2003no} on the difference between no-signaling and no faster-than-light signaling.} and it has nothing to do with the safety of our infant grandfathers.

\subsection{The bunker game}\label{sec:TheBunkerGame}

To make the no-signaling hypothesis concrete, consider the following cooperative two-player game, to be played by Alice and Bob. The game requires a highly secure bunker, designed to prevent what goes on inside it from being inferred from outside. No expense has been spared. The bunker is equipped with lead-lined walls, vacuum-gapped thermal isolators, Faraday cages, etc., and is buried deep underground with no phone lines or other communication infrastructure connecting it to the outside world. It can only be accessed via a single tunnel, which can be closed off by doors designed to be as impermeable as the bunker itself.  

Each round, Alice enters the bunker, seals the tunnel behind her, and flips a fair coin whose sides are marked $\mathrm{Y}$ and $\mathrm{N}$.\footnote{The outcome of the coin flip is assumed to be independent of any systems Alice and Bob have prepared, and of any resources or information they shared before the round.} Alice may then do whatever she likes to try to communicate the outcome of the coin flip to Bob. The only rule is that she must stay inside the bunker, with the tunnel sealed, until the current round of the game is over. Bob, who remains outside throughout, can do whatever he likes to try to infer the outcome of Alice's coin flip. Before the tunnel is reopened, Bob must guess whether the coin landed $\mathrm{Y}$ or $\mathrm{N}$. He can take as much time as he likes to guess. In particular, the spacetime region in which Bob arrives at his guess can be fully inside the future light cone of any relevant action Alice takes. Alice can bring anything she wants into the bunker. This includes any laboratory equipment (e.g., a Stern--Gerlach apparatus), and any physical system she likes that can have interacted in any way with something remaining outside with Bob. One situation to have in mind is the analog of the signaling scenario of Section \ref{subsection:whatworry}, where Alice goes into the bunker with one particle from an entangled pair, the partner of which stays with Bob.

The winning conditions of the bunker game are as follows. A positive integer $N$ and a parameter $\delta\in \big(0,\tfrac{1}{2}\big)$ are chosen. Alice and Bob play $N$ independent rounds of the game, each of which ends with Alice coming out of the bunker and revealing the answer so that Bob's guess can be checked. After $N$ rounds, if Bob has guessed correctly more than $N\times\big(\tfrac{1}{2}+\delta\big)$ times, Alice and Bob win. I will call this the $(N,\delta)$ bunker game. The question is whether there is a strategy Alice and Bob can carry out that gives them a $95\%$ chance\footnote{If Alice and Bob can find a strategy that ensures Bob's guesses are correlated with the coin flips, they can reach any fixed probability of winning by taking $N$ large enough. For instance, if Bob's guess is correct with probability $\frac{1}{2}+\varepsilon$ on each round, then for $\delta=\varepsilon/2$, Chebyshev's inequality bounds the probability of losing by $1/(N\varepsilon^2)$, so $N\ge 20/\varepsilon^2$ rounds suffice for a $95\%$ chance of winning.} of winning the $(N,\delta)$ bunker game for at least one pair $(N,\delta)$.

Let's now analyze an illustrative strategy using NRQM. As in Section \ref{subsection:whatworry}, we describe Alice's particle by its spin alone, so $\mathcal{H}_A=\mathbb{C}^2$ with basis $\{\ket{\uparrow},\ket{\downarrow}\}$. Bob, by contrast, may be far more resourceful, performing more complicated experiments (e.g., arrival-time experiments). We will therefore allow the Hilbert space $\mathcal{H}_B$ describing Bob's particle to be infinite-dimensional, like the space $L^2(\mathbb{R}^3)$ used for a particle's spatial degrees of freedom, assuming only that it is separable. Any joint state $\Psi\in\mathcal{H}_A\otimes\mathcal{H}_B$ can then be written as $\Psi=\ket{\uparrow}\!\otimes \psi + \ket{\downarrow}\!\otimes \phi$ for some $\psi,\phi\in\mathcal{H}_B$, where normalization requires $\|\psi\|^2+\|\phi\|^2=1$.

Suppose Alice and Bob carry out the following protocol. If the coin lands $\mathrm{Y}$, Alice measures spin in the $z$-direction, which is associated with the basis $\mathcal{B}_z=\{\ket{\uparrow},\ket{\downarrow}\}$, and if it lands $\mathrm{N}$, she measures spin in the $x$-direction, which is associated with $\mathcal{B}_x=\{\ket{+},\ket{-}\}$. What does NRQM tell us about whether this protocol works?

Let's call the situation Bob is in after Alice's measurement Condition $\mathrm{Y}$ when the coin lands $\mathrm{Y}$, and Condition $\mathrm{N}$ when it lands $\mathrm{N}$. Analyzing Condition $\mathrm{Y}$ is straightforward. The Born rule tells us that the probability Alice gets ``$z$-up'' is $\|\psi\|^2$ and the probability she gets ``$z$-down'' is $\|\phi\|^2$. In the former case, the collapse rule tells us Bob's system is subsequently described by $\psi/\|\psi\|$, and in the latter by $\phi/\|\phi\|$.\footnote{If $\psi=0$, the corresponding outcome occurs with probability zero, and we simply omit that term from $\mu_z$. The same convention applies to $\phi$, and to $\psi\pm\phi$ in $\mu_x$ below.} From Bob's perspective, when the coin lands $\mathrm{Y}$, his particle behaves exactly as though it had been produced by a random state-generating machine that returns $\psi/\|\psi\|$ with probability $\|\psi\|^2$ and returns $\phi/\|\phi\|$ with probability $\|\phi\|^2$. This can be described by a discrete probability distribution\footnote{Section \ref{sec:ProperMixture} explains why I do not use density operators here.} $\mu_z$ over $\mathcal{H}_B$, which we write as
$$
\mu_z=\left\{\left(\|\psi\|^2,\frac{\psi}{\|\psi\|}\right),\left(\|\phi\|^2,\frac{\phi}{\|\phi\|}\right)\right\}\,.
$$
To analyze Condition $\mathrm{N}$, we first express the joint state $\Psi$ in terms of $\mathcal{B}_x$, which gives us  $\Psi= \ket{+}\!\otimes\!\left(\tfrac{\psi+\phi}{\sqrt{2}}\right) + \ket{-}\!\otimes\!\left(\tfrac{\psi-\phi}{\sqrt{2}}\right)$. We then find that, from Bob's perspective, his particle behaves as though it had been produced by a random state-generating machine described by
$$
\mu_x = \left\{\left(\frac{\|\psi+\phi\|^2}{2},\frac{\psi+\phi}{\|\psi+\phi\|}\right),\left(\frac{\|\psi-\phi\|^2}{2},\frac{\psi-\phi}{\|\psi-\phi\|}\right)\right\}\,.
$$
Whether this protocol lets them win the bunker game comes down to whether Bob can reliably distinguish the output of the random state-generating machine described by $\mu_z$ from that of the one described by $\mu_x$. 

This is where NRQM's restrictive claim about observability does its work. Everything Bob can learn about his particle comes from measuring an observable, represented by a self-adjoint operator $M_B$ on $\mathcal{H}_B$. By the spectral theorem, $M_B$ gives rise to a projection-valued measure, which associates a projection $E(O)$ with each set of possible outcomes $O$. Suppose, then, that on each round of the game, Bob measures some such $M_B$ and bases his guess on whether the outcome lies in some set $O$.\footnote{Bob is free to choose different measurements and outcome sets on different rounds.} Let $\mathbb{P}(O\,;\,\mu_z)$ denote the probability that his outcome lies in $O$ when the coin has landed $\mathrm{Y}$, so that his particle is described by a random wave function drawn according to $\mu_z$, and define $\mathbb{P}(O\,;\,\mu_x)$ analogously. Conditioning on Alice's outcome, the standard law of total probability\footnote{This application of the law of total probability is unproblematic, since Alice's measurement unambiguously occurs and its outcome can be recorded. This makes it, on all accounts, the sort of event that one can condition on (unlike, e.g., unobserved which-path alternatives in a double-slit experiment).} gives
\begin{align*}\label{eq:muzprob}
\mathbb{P}(O\,;\,\mu_z)
&=\left\| E(O)\,\frac{\psi}{\|\psi\|}\right\|^2 \|\psi\|^2 +\left\| E(O)\,\frac{\phi}{\|\phi\|}\right\|^2 \|\phi\|^2
=\left\| E(O)\,\psi\right\|^2 +\left\| E(O)\,\phi\right\|^2\,.\numberthis
\end{align*}
Similarly, conditioning on Alice's outcome and applying the parallelogram law gives
\begin{align*}\label{eq:muxprob}
\mathbb{P}(O\,&;\,\mu_x)
=\left\| E(O)\,\frac{\psi+\phi}{\|\psi+\phi\|}\right\|^2 \frac{\|\psi+\phi\|^2}{2} +\left\| E(O)\,\frac{\psi-\phi}{\|\psi-\phi\|}\right\|^2\frac{\|\psi-\phi\|^2}{2}\\
&=\frac{1}{2}\left\| E(O)\,\psi + E(O)\,\phi\right\|^2 +\frac{1}{2}\left\| E(O)\,\psi - E(O)\,\phi\right\|^2= \left\| E(O)\,\psi\right\|^2 +\left\| E(O)\,\phi\right\|^2\,.\numberthis
\end{align*}
Putting together Equations \eqref{eq:muzprob} and \eqref{eq:muxprob}, we have that $\mathbb{P}(O\,;\,\mu_z)=\mathbb{P}(O\,;\,\mu_x)$ for every self-adjoint operator $M_B$ and every set of outcomes $O$. Bob's outcome is therefore distributed identically regardless of how the coin landed, so any rule for turning his outcome into a guess is expected to be correct exactly half the time. No matter how large $N$ is, then, Alice and Bob cannot be confident of winning the $(N,\delta)$ bunker game for any $\delta>0$ using this protocol.

The above argument gives the essential structure of an NRQM-based no-signaling argument (see Section \ref{OPFPosts} for a more formal discussion). Its conclusion concerns only Bob's marginal statistics, but the derivation passed through post-collapse states that depend on Alice's choice of measurement. If this dependence is physically real, the influence involved requires no spatial intermediary, and no bunker, however carefully designed, can completely seal off Alice's ability to influence the outside world. Can she use this influence to signal through barriers? NRQM says she cannot.

Not every theory agrees. An analysis of the bunker game in Kent's local state formalism \cite{kent2005nonlinearity, kent2021quantum}, equipped with quantum state readout devices, reaches the opposite conclusion. According to that formalism, as soon as Alice's measurement lies in Bob's past light cone, the local state assigned to Bob's particle changes. Since Bob can wait as long as he likes before guessing, he can wait until he is in the future light cone of Alice's relevant actions, and then use the quantum state readout device to determine whether his particle is described by one of the states produced by $\mu_z$ or by $\mu_x$. If Alice and Bob choose an initial joint state for which these sets of states are disjoint (the singlet state, for instance), Bob can guess correctly every time. So if state readout devices exist and the local state formalism is valid, Alice and Bob cannot send \emph{superluminal} signals, but they can nevertheless win the bunker game.

\subsection{A minimalistic framework for discussing signaling in NRQM}\label{sec:overview}

We can now state the central question of this paper more precisely. Suppose that Alice and Bob \emph{cannot} win the bunker game. What does this tell us about what is \emph{not} ``observable'' about a quantum system? In particular, how does this assumption constrain the probability-generating fragment of quantum mechanics?

As I have emphasized, the answer depends on which assumptions are made and how they are formalized. The relevant class of experiments, discussed in detail in Section \ref{sec:PandM}, is not \emph{all} possible experiments, but only ``prepare-and-measure experiments,''\footnote{In this paper, I will try to follow Bell's guidance in \cite{bell1990against} and use the word ``experiment'' rather than ``measurement'' whenever I am no longer working within textbook NRQM. The term ``prepare-and-measure experiment'' departs from this practice, but I will use it because it is standard (see, e.g., \cite{leifer2014quantum}) and appropriately circumscribes the class of experiments discussed. As I use it, the term is not meant to imply that some preexisting property of the prepared system is being ``measured.'' Experiments resembling Rube Goldberg machines, in which a prepared particle sets off a complicated chain of events (and where there is no expectation that observations of downstream phenomena amount to ``measuring'' a property of the particle that set off the process), are still within scope.} which fit the following idealized description. First, a procedure is followed that ``prepares a target system,'' e.g., a single particle. Then, some further procedure is carried out on the prepared system, e.g., an arrival-time experiment. The whole procedure is then repeated many times, and outcome frequencies are tabulated.

The aim is to quantitatively describe outcome probabilities for such experiments. I will take an ``operationalist'' perspective that models experiments as black boxes that take in descriptions of ``prepared systems'' and output outcome probabilities. These black boxes will be described by ``outcome probability functions'' (OPFs), a term borrowed from Galley and Masanes \cite{galley2017classification}. This perspective does not reflect a commitment to the view that ``any claim that cannot be settled by experiment is meaningless.'' Instead, it reflects the view that empirical questions ought to be the common denominator between different approaches to quantum theory. Indeed, whatever their interpretational stance, most theorists agree that making quantitative predictions about outcome probabilities is \emph{among} the things a satisfactory theory of quantum mechanics ought to be able to do.

Modeling experiments as black boxes sets aside dynamical questions. I have chosen to do this for two reasons. First, beyond the assumption that the subsystems do not interact (which the framework captures through a notion of isolation, see Section \ref{sec:PPI}), the dynamics of state evolution plays no essential role in the signaling story told above. The only change of state the argument used was the one produced by the collapse postulate. Second, assumptions about dynamics are contested. For instance, whether to restrict attention to unitary evolution or to allow the non-unitary evolution proposed by collapse theories is something proponents of different interpretations disagree about. Keeping such assumptions out of the framework lets its results apply more broadly. Many important results relate no-signaling to dynamical considerations such as linearity \cite{gisin1989stochastic, gisin1990weinberg, polchinski1991weinberg, simon2001no}. These typically take the standard account of observables for granted and ask what no-signaling implies about dynamics, whereas this paper sets dynamics aside and asks what no-signaling implies about observables.

I will make the standard assumption that ``prepared states'' can be labeled by elements of a Hilbert space. This allows us to discuss ``orthogonality'' and ``superpositions'' in the usual way. Among other things, this means the OPFs we are looking for are functions from a unit sphere $S(\mathcal{H})$ in a Hilbert space $\mathcal{H}$ to a space of probability measures on an outcome space $\Omega$. Another important background assumption is that composite quantum systems are described using the tensor product formalism (see Section \ref{sec:compositeExps}). From the perspective of generalized probabilistic theories (GPTs), both of these representational choices are non-trivial assumptions (see \cite{plavala2023general}), and the same is true in other frameworks such as Sorkin's quantum measure theory \cite{sorkin1994quantum}. Many theorists have tried to derive them from assumptions they take to be more fundamental. Some derive Hilbert space structure, mostly in finite dimensions \cite{hardy2001quantum, chiribella2011informational, masanes2011derivation}, while others derive the tensor product representation of composite systems \cite{aerts1978physical, carcassi2021four}. I take both as given, and ask what follows when no-signaling is added.

One motivation for the present work comes from studying proposed solutions to the arrival-time problem discussed in Section \ref{subsec:observability}. I have therefore aimed for a framework that is as simple as possible while still general enough to handle arrival-time models. Since these models are typically formulated using infinite-dimensional Hilbert spaces like $L^2(\mathbb{R})$, I will \emph{not} assume that the Hilbert spaces used to describe prepared states are finite-dimensional, only that they are separable. The remaining assumptions that the story above relied on will be stated as explicit postulates, with single-system postulates given in Section \ref{sec:PandM}, and composite-system postulates given in Section \ref{sec:compositeExps}.

\section{Prepare-and-Measure Experiments}\label{sec:PandM}

As described informally in Section \ref{sec:overview}, a prepare-and-measure experiment has three ingredients: (1) a procedure to ``prepare the target system,'' (2) a preparation-independent experimental protocol, and (3) a means of recording which of a predetermined set of possible ``outcomes'' has occurred.

What exactly licenses the claim that a certain ``preparation'' has been achieved is not a question I will address in detail, though there are several common justifications available. One is post-selection on a measurement outcome, e.g., only keeping particles that go up in a Stern--Gerlach apparatus \cite[Ch.~5]{feynman1965lectures}. Another is performing certain laboratory procedures, like ``cooling a particle to the ground state of a trap'' \cite{diedrich1989laser}, that are thought to reliably produce a given quantum state. Techniques from quantum state tomography \cite{heinosaari2013quantum} can then be used to corroborate the claim that a given procedure ``prepares'' the intended state. 

Once this preparation stage is complete, a preparation-independent protocol is applied. For instance, the protocol might involve switching off a trap, allowing the particle to propagate freely for a predetermined interval of time, firing a laser pulse over the region the particle is expected to occupy, and then detecting the scattered or transmitted light. Such a protocol is ``preparation independent'' in the sense that it can be followed no matter what state the trapped particle began in. The ``outcome'' of the experiment need not be a single discrete value, but can be something much richer, e.g., an image captured by an array of optical detectors, ultimately encoded as brightness levels for a family of pixels. The prepared state describes the prepared system at the time the preparation-independent protocol is initiated.\footnote{Unlike in textbook applications of the Born rule to time-dependent states, where an observable is measured at a specific instant of time, the outcome of a prepare-and-measure experiment need not be associated with any ``time of measurement.'' The notion of an ``initiation time'' for an experiment (e.g., the time at which a particle is released from a trap) provides a natural reference time in models where the relative timing of Alice and Bob's actions plays a role.\label{fn:initiation}}

The theoretical goal is to describe the probability of a given set $O$ of outcomes as a function of the preparation procedure or, in more familiar terms, of the state it prepares. As the sections below describe in more detail, the mathematical model I associate with a prepare-and-measure experiment is a tuple $\mathcal{E}=(\Omega_{\mathcal{E}},\Sigma_{\mathcal{E}},\mathcal{H}_\mathcal{E}, \mathbb{P}_{\mathcal{E}})$. Here $(\Omega_{\mathcal{E}},\Sigma_{\mathcal{E}})$ is a measurable space describing the possible experimental outcomes, $\mathcal{H}_{\mathcal{E}}$ is a non-zero separable complex Hilbert space used to model preparation procedures, and $\mathbb{P}_{\mathcal{E}}:\Sigma_{\mathcal{E}}\times S(\mathcal{H}_\mathcal{E})\to [0,1]$ is a function such that the map $O\mapsto \mathbb{P}_{\mathcal{E}}(O\,;\psi)$ is a probability measure on $(\Omega_{\mathcal{E}},\Sigma_{\mathcal{E}})$ for each unit vector $\psi\in S(\mathcal{H}_\mathcal{E})$. The function $\mathbb{P}_{\mathcal{E}}$ is also required to be phase-invariant, i.e., $\mathbb{P}_{\mathcal{E}}(O\,;e^{i\theta}\psi)=\mathbb{P}_{\mathcal{E}}(O\,;\psi)$ for all $\theta\in \mathbb{R}$. I will call such a tuple a \emph{Hilbert space-based Prepare-and-Measure experiment Model}, or HPMM.

As the density of adjectives involved in this definition suggests, this need not be understood as a completely general framework for modeling \emph{all} quantum-mechanical experiments.\footnote{Examples of experimental protocols that are not directly covered are ``weak measurement'' protocols involving post-selection \cite{aharonov1988result}. In such protocols, the directly observed empirical data are processed using post-selection, and data points that are post-selected out are not assigned an outcome. This violates the requirement of the present framework, stated below, that each run have exactly one outcome. It also has structural consequences. The conditional outcome probabilities in such protocols are \emph{ratios} of quadratic forms, and so generally fail Postulate \ref{QF} below. However, this is simply an artifact of how the empirical data are processed, and not the sort of violation of Postulate \ref{QF} that anyone would find surprising.} An HPMM formalizes a collection of successful modeling conventions that many different approaches to NRQM share. The scope of HPMMs, however, is the class of experiments to which the textbook Born rule is applied. Since the Born rule is often thought to capture the full empirical content of quantum mechanics, many theorists \emph{do} take this class of models to apply to all ``quantum experiments.''\footnote{For example, Tumulka asserts that ``every quantum physical experiment'' is described by a POVM \linebreak \cite[Sec.~5.1.2]{tumulka2022foundations} (cf.~Postulate \ref{POVM} below).}

In the framework that follows, the mathematical objects over which we will quantify are HPMMs, their composites, and families thereof. The postulates are stated in terms of two predicates. The first, \emph{availability}, applies to HPMMs. The second, \emph{describing a post-preparation independent composite experiment}, applies to composite HPMMs and is introduced in Section \ref{sec:PPI}. For brevity, a composite HPMM satisfying the second predicate will sometimes be called ``post-preparation independent.'' Informally, an HPMM is meant to be ``available'' when it describes a prepare-and-measure experiment, such as a Stern--Gerlach measurement, that Alice or Bob can actually carry out. Formally, both should be read as primitive predicates. The postulates below are constraints on these predicates, but the mathematical framework itself does not determine which HPMMs are available or which composites are post-preparation independent. Availability assumptions come in through explicit postulates, in particular Postulate \ref{EQC}. The logical relations mapped in Figure \ref{Implications} therefore hold regardless of how these predicates are interpreted, and a model of a given set of postulates amounts to specifying which HPMMs are available and which composites describe post-preparation independent experiments (see Section \ref{subsec:PRmodel} for an example). 

\subsection{Outcome spaces}

The first ingredient needed to model a prepare-and-measure experiment is an \emph{experimental outcome space}, which we formalize via the standard Kolmogorov framework for probability theory \cite{kolmogorov2018foundations}. A space of possible experimental outcomes is represented by a measurable space $(\Omega_{\mathcal{E}},\Sigma_{\mathcal{E}})$, where $\Omega_{\mathcal{E}}$ is understood to be the set of all possible outcomes, and $\Sigma_{\mathcal{E}}$ is a $\sigma$-algebra of subsets of $\Omega_{\mathcal{E}}$, interpreted as the sets of outcomes to which we can meaningfully ascribe probabilities.

As an example, consider the standard Stern--Gerlach experiment with a planar glass detection plate oriented along the $xy$-plane. Suppose each run uses a single silver atom. The atom either adheres to the glass plate at some position $(x,y)$ or it does not, so we can take $\Omega_{\mathcal{SG}}$ to be $\{\mathrm{ND}\}\cup \mathbb{R}^2$, where $\mathrm{ND}$ indicates ``no detection.'' A standard choice for $\Sigma_{\mathcal{SG}}$ is the $\sigma$-algebra generated by the open sets of $\mathbb{R}^2$ together with the singleton set $\{\mathrm{ND}\}$. The finite size of any real glass plate is not a problem here. Our outcome space $\Omega_{\mathcal{SG}}$ just needs to be able to describe the information we want to record about our experiment, so points outside the plate can simply be assigned probability zero.

Note that the outcome space $\Omega_{\mathcal{E}}$ is \emph{not} given to us by the physical experiment itself, but rather reflects our choices about what to model. In the example of the Stern--Gerlach experiment, for instance, we could equally well choose the coarser outcome space $\Omega_{\mathcal{SG}}'=\{\mathrm{ND},\mathrm{UP},\mathrm{DOWN}\}$, and describe any atom that adheres to the plate at a position $(x,y)$ with $y>0$ as an $\mathrm{UP}$ outcome, and any atom that adheres elsewhere as a $\mathrm{DOWN}$ outcome. This is, in fact, the coarse-grained description theorists typically use. Likewise, we could choose a refinement of $\Omega_{\mathcal{SG}}$, e.g., by recording data about \emph{when} the atom adheres to the plate.

In this framework, we assume that each run of the experiment has \emph{exactly one} ``outcome.'' This requires that $\Omega_{\mathcal{E}}$ comprehensively describe all possible outcomes. In practice, this often means giving the experiment a fixed finite duration $T$ and adding to the outcome space a formal $\mathrm{ND}$ result indicating a ``null'' outcome, as in the Stern--Gerlach example above. In experiments whose natural duration depends on the prepared state, a temporal cutoff gives an unambiguous criterion for declaring that the null outcome has occurred. For many experiments, this is not a substantive restriction. Experiments with no natural temporal cutoff, however, such as waiting for the decay of a long-lived particle, are a different matter. Each truncation of such an experiment at a finite time $T$ can be modeled by an HPMM, but it is not always clear how to stitch such a family of HPMMs together into a single model for the untruncated experiment.

\subsection{State preparation}\label{subsec:stateprep}

The next thing we need is a mathematical description of preparation procedures. Outcome spaces are described just as they would be in any probabilistic theory, but preparations are where the ``quantum-mechanical'' character of the experiments becomes important. In quantum mechanics, ``prepared states'' are typically labeled by elements of an inner product space, which lets us make sense of ``adding together'' the labels of preparation procedures to form superpositions. In the present framework, I will assume that the preparation procedures for the quantum experiments of interest can be labeled by elements of the unit sphere $S(\mathcal{H}_\mathcal{E})$ of a separable Hilbert space $\mathcal{H}_\mathcal{E}$. I will call $\mathcal{H}_\mathcal{E}$ ``the space of preparations.''\footnote{Since we will only be interested in phase-invariant outcome probability functions $\psi\mapsto\mathbb{P}_{\mathcal{E}}(O\,;\psi)$, we would lose no generality by assuming that prepared states are described by elements of the projective Hilbert space $\mathcal{P}(\mathcal{H}_{\mathcal{E}})$. However, in the context of Postulates \ref{WPL} and \ref{PL}, working with rays rather than unit vectors becomes cumbersome, since addition of rays in $\mathcal{P}(\mathcal{H}_{\mathcal{E}})$ is ill-defined. I therefore work with the redundant space $S(\mathcal{H}_{\mathcal{E}})$ and impose phase invariance as an additional requirement.} Section \ref{sec:ProperMixture} generalizes this labeling to ``proper mixtures,'' described by elements of the space $\mathrm{M}_1(S(\mathcal{H}_{\mathcal{E}}))$ of Borel probability measures on $S(\mathcal{H}_{\mathcal{E}})$.

Not every real-world experiment involves a clean ``preparation'' phase, or even any well-defined preparation phase at all. For instance, in the original Stern--Gerlach experiment (as opposed to an idealized Stern--Gerlach experiment performed with a single atom carefully prepared in a definite spin state), the closest thing to a ``preparation phase'' is the heating of the oven that emits the silver atoms, and it is not clear what initial wave function, or even what probability measure over initial wave functions, should be assigned. Other quantum-mechanical experiments, e.g., those involving cosmic rays, do not naturally fit into this framework, since the particles involved are not ``prepared'' by the experimenter at all.

\subsection{Outcome probability functions (OPFs)}

The final component of an HPMM is its OPF (cf.~\cite{galley2017classification}). Given an outcome space $(\Omega_{\mathcal{E}},\Sigma_{\mathcal{E}})$ and a space of preparations $\mathcal{H}_{\mathcal{E}}$, an \emph{outcome probability function} is a function $f_{\mathcal{E}}$ that assigns to each $\psi\in S(\mathcal{H}_{\mathcal{E}})$ a probability measure $f_{\mathcal{E}}(\psi)$ on $(\Omega_{\mathcal{E}},\Sigma_{\mathcal{E}})$. Equivalently, an outcome probability function $f_{\mathcal{E}}$ can be encoded by $\mathbb{P}_{\mathcal{E}}:\Sigma_{\mathcal{E}}\times S(\mathcal{H}_{\mathcal{E}})\to [0,1]$ given by $\mathbb{P}_{\mathcal{E}}(O\,;\psi)= f_{\mathcal{E}}(\psi)(O)$. This function is informally understood to describe the probability of an outcome in the set $O\in \Sigma_{\mathcal{E}}$ when the preparation label is $\psi$. OPFs are defined on all of $S(\mathcal{H}_\mathcal{E})$, and the postulates below quantify over all of $S(\mathcal{H}_\mathcal{E})$. This amounts to adopting the standard idealization that every unit vector in $\mathcal{H}_\mathcal{E}$ can be treated as a label for a possible preparation.\footnote{As noted at the start of this section, I am setting aside the question of when one is justified in claiming that a given state can be or has been ``prepared.'' If some states cannot be prepared, even approximately, this idealization fails, and the instances of the postulates involving those states have no direct empirical content. For instance, a violation of Postulate \ref{NS} that occurred only for joint states that cannot be prepared could not be used to signal. What the postulates imply when restricted to proper subsets of $S(\mathcal{H})$ is a question this paper leaves open.}

The interpretation of probability in this framework deserves a brief comment. Formally, the probabilities in this paper are simply real numbers, which we can add, multiply, and compare in the standard way. Claims about relationships between outputs of an OPF, such as Postulate \ref{PL} below, are standard arithmetic claims. The framework requires only that these numbers obey the Kolmogorov axioms, and is otherwise agnostic about what they mean. Interpretive assumptions enter only in motivating the postulates, e.g., in the link between probabilities and observed frequencies used in Section \ref{sec:signalingViaChoice}.

\subsubsection{Postulates on OPFs: WPL, PL, and QF}\label{OPFPosts}

We can now state three postulates on the structure of OPFs, beginning with the simplest.

\begin{postulatecustom}{WPL}[The weak parallelogram law for OPFs]\label{WPL} An HPMM $\mathcal{E}$ has \emph{Property \ref*{WPL}} if for any pair of orthogonal states $\psi, \phi\in S(\mathcal{H}_{\mathcal{E}})$, and for any set of outcomes $O\in \Sigma_{\mathcal{E}}$, 
$$
\mathbb{P}_{\mathcal{E}}\left(O\,;\psi\right) +  \mathbb{P}_{\mathcal{E}}\left(O\,;\phi\right) =\mathbb{P}_{\mathcal{E}}\left(O\,;\frac{\psi + \phi}{\sqrt{2}}\right) +  \mathbb{P}_{\mathcal{E}}\left(O\,;\frac{\psi - \phi}{\sqrt{2}}\right)\,.
$$
Under \emph{Postulate \ref*{WPL}}, all available HPMMs satisfy Property \ref*{WPL}. 
\end{postulatecustom}

\noindent The weak parallelogram law concerns only equally weighted superpositions of orthogonal states.\footnote{I have stated Postulate \ref{WPL} for phase-invariant functions on $S(\mathcal{H}_{\mathcal{E}})$ for simplicity. It can also be stated in terms of rays in $\mathcal{P}(\mathcal{H}_{\mathcal{E}})$. A function $f:\mathcal{P}(\mathcal{H}_{\mathcal{E}})\to [0,1]$ satisfies the \emph{projective weak parallelogram law} if for any pair of orthogonal rays $r_1,r_2\in\mathcal{P}(\mathcal{H}_{\mathcal{E}})$, for any $\tilde{r}_1,\tilde{r}_2\in S(\mathcal{H}_{\mathcal{E}})$ such that the natural projection map $\pi : \mathcal{H}_{\mathcal{E}}\setminus\{0\}\to\mathcal{P}(\mathcal{H}_{\mathcal{E}})$ satisfies $\pi(\tilde{r}_i)=r_i$ for $i\in\{1,2\}$, $f(r_1)+f(r_2) = f\big(\pi(\tilde{r}_1+\tilde{r}_2)\big) + f\big(\pi(\tilde{r}_1-\tilde{r}_2)\big)$. It is straightforward to check that a phase-invariant function $\tilde{f}:S(\mathcal{H}_{\mathcal{E}})\to [0,1]$ satisfies the weak parallelogram law if and only if there exists a function $f$ satisfying the projective weak parallelogram law such that $\tilde{f}=f\circ \pi$. A similar discussion applies to Postulate \ref{PL}.} The following stronger formulation of this postulate will also be useful.

\begin{postulatecustom}{PL}[The parallelogram law for OPFs]\label{PL} An HPMM $\mathcal{E}$ has \emph{Property \ref*{PL}} if for any pair of elements $\psi,\phi\in S(\mathcal{H}_{\mathcal{E}})$, any set of outcomes $O\in \Sigma_{\mathcal{E}}$, and any unitary matrix $U = \left(\begin{smallmatrix}
a & b \\
c & d
\end{smallmatrix}\right)$ such that $a\psi+b\phi$ and $c\psi+d\phi$ are both non-zero, 
\begin{align*}
\mathbb{P}_{\mathcal{E}}\left(O\,;\psi\right) &+  \mathbb{P}_{\mathcal{E}}\left(O\,;\phi\right)=
\| a \psi + b \phi\|^2\, \mathbb{P}_{\mathcal{E}}\left(O\,;\frac{a\psi + b\phi}{\|a\psi+b\phi\|}\right)\\
&\qquad\qquad\qquad\qquad\qquad\qquad\qquad\qquad + \|c\psi + d \phi\|^2\,\mathbb{P}_{\mathcal{E}}\left(O\,;\frac{c\psi + d\phi}{\|c \psi + d\phi\|}\right).
\end{align*}
Under \emph{Postulate \ref*{PL}}, all available HPMMs satisfy Property \ref*{PL}. 
\end{postulatecustom}

Postulate \ref{PL} is directly related to the bunker game in Section \ref{sec:TheBunkerGame} above. Under the standard measurement postulates of NRQM, it is equivalent to the assertion that if Alice and Bob share the joint state $\tfrac{1}{\sqrt{2}}\big(\ket{\uparrow}\otimes \psi+\ket{\downarrow}\otimes \phi\big)$, Bob's probability of seeing an experimental outcome in the set $O$ after Alice measures does not depend on whether Alice measures in the basis $\{\ket{\uparrow},\ket{\downarrow}\}$ or in its image $\{V\ket{\uparrow},V\ket{\downarrow}\}$ under a unitary transformation $V$. Note that Postulate \ref{WPL} is the special case of Postulate \ref{PL} where $a=b=c=\tfrac{1}{\sqrt{2}}$, $d=-\tfrac{1}{\sqrt{2}}$, and $\psi$ and $\phi$ are orthogonal. Conversely, Postulate \ref{WPL} implies Postulate \ref{PL} whenever $\dim(\mathcal{H}_{\mathcal{E}})\ge 3$ (see Figure \ref{Implications}).

The next postulate concerns the algebraic structure of OPFs. Recall that a \emph{sesquilinear form} on a Hilbert space $\mathcal{H}$ is a function $q:\mathcal{H}\times\mathcal{H}\to\mathbb{C}$ (not assumed to be bounded) such that for all $w\in\mathcal{H}$, $v\mapsto q(w,v)$ is linear and $v\mapsto q(v,w)$ is conjugate-linear. A \emph{quadratic form} on $\mathcal{H}$ is a function $Q:\mathcal{H}\to\mathbb{C}$ of the form $Q(v)=q(v,v)$ for some sesquilinear form $q$.

\begin{postulatecustom}{QF}[Quadratic form structure of experimental outcome probability functions]\label{QF} An HPMM $\mathcal{E}$ has \emph{Property \ref*{QF}} if for any $O\in \Sigma_{\mathcal{E}}$, there exists a quadratic form \linebreak $Q_O:\mathcal{H}_{\mathcal{E}}\to \mathbb{R}$ such that for all $\psi\in S(\mathcal{H}_{\mathcal{E}})$, $\mathbb{P}_{\mathcal{E}}(O\,;\psi)= Q_O(\psi)$. Under \emph{Postulate \ref*{QF}}, all available HPMMs satisfy Property \ref*{QF}.
\end{postulatecustom}

Postulate \ref{QF} expresses the familiar fact that probabilities in textbook NRQM are quadratic in the wave function.

\subsubsection{POVM universality}\label{sssec:POVM}

The POVM framework is widely thought to provide the most general description of measurement outcome probabilities in quantum theory (see, e.g., \cite{peres1995quantum, holevo2011probabilistic, heinosaari2011mathematical}). The following definition follows Busch et al.~in \cite{busch1995operational}. Let $\mathcal{L}(\mathcal{H})$ denote the space of bounded operators on a Hilbert space $\mathcal{H}$. Given a measurable space $(\Omega,\Sigma)$, a \emph{positive operator-valued measure acting on $\mathcal{H}$} is a function $E:\Sigma\to \mathcal{L}(\mathcal{H})$ such that (1) for all $O\in \Sigma$, $E(O)$ and $I-E(O)$ are positive operators, and (2) for all sequences of pairwise disjoint sets $\{O_i\}\subset \Sigma$, the partial sums $\sum_{i=1}^N E(O_i)$ converge in the weak operator topology to $E(\bigcup_{i=1}^\infty O_i)$. If $E(\Omega)=I$, $E$ is said to be \emph{normalized}.

\begin{postulatecustom}{POVM}[POVM universality]\label{POVM} An HPMM $\mathcal{E}$ has \emph{Property \ref*{POVM}} if there exists a normalized POVM $E_{\mathcal{E}}:\Sigma_{\mathcal{E}}\to \mathcal{L}(\mathcal{H}_{\mathcal{E}})$ such that for all $\psi \in S(\mathcal{H}_{\mathcal{E}})$ and $O\in \Sigma_{\mathcal{E}}$, $\mathbb{P}_{\mathcal{E}}(O\,;\psi)=\langle \psi,E_{\mathcal{E}}(O)\psi\rangle$. Under \emph{Postulate \ref*{POVM}}, all available HPMMs satisfy Property \ref*{POVM}.
\end{postulatecustom}

\subsubsection{Orthodox HPMMs}\label{sec:OrthodoxOPFs}

I will call an HPMM $\mathcal{E}$ \emph{orthodox} if it has Property \ref{POVM} and $E_{\mathcal{E}}(O)$ is a projection, i.e., $E_{\mathcal{E}}(O)^2=E_{\mathcal{E}}(O)$, for each $O\in\Sigma_{\mathcal{E}}$. In this case, $E_{\mathcal{E}}$ is called a \emph{projection-valued measure} (PVM). The main examples of orthodox HPMMs we will need are those describing textbook measurements of a qubit, where $\mathcal{H}_{\mathcal{E}}=\mathbb{C}^2$ and the outcome space $\Omega_{\mathcal{E}}$ is the two-element set $\{\mathrm{Y},\mathrm{N}\}$ equipped with the discrete $\sigma$-algebra. Suppose such an HPMM is nontrivial, i.e., $\mathbb{P}_{\mathcal{E}}(\{\mathrm{Y}\}\,;\psi)\notin \{0,1\}$ for some $\psi\in S(\mathcal{H}_\mathcal{E})$. Then there exists an element $u\in S(\mathcal{H}_{\mathcal{E}})$ (unique up to phase) such that for all $\phi\in S(\mathcal{H}_\mathcal{E})$, $\mathbb{P}_{\mathcal{E}}(\{\mathrm{Y}\}\,;\phi)=|\langle u,\phi\rangle|^2$. I will therefore denote such an orthodox HPMM by $\mathcal{E}_{\overline{u}}$, where $\overline{u}$ denotes the ray defined by $u$.

The set $\mathcal{P}(\mathbb{C}^2)$ of such rays can be identified with the Bloch sphere $S^2$, and $\overline{u}\in S^2$ can be thought of as a setting of the measuring apparatus. For a massive spin-$\frac{1}{2}$ particle, $\overline{u}$ corresponds to an oriented axis along which the spin is measured, and for a photon qubit, it parameterizes arrangements of wave plates and polarizing beam splitters that implement a textbook ``polarization measurement.''

\subsection{Proper mixtures of quantum states}\label{sec:ProperMixture}

In many settings, a preparation procedure is better described by a probability measure $\mu\in \mathrm{M}_1\left(S(\mathcal{H}_{\mathcal{E}})\right)$ than by a single element $\psi\in S(\mathcal{H}_{\mathcal{E}})$. This is the case when the prepared state is itself selected by a random process, as happens in the NRQM model of the bunker game (see Section \ref{sec:TheBunkerGame}). Such measures are called ``proper mixtures.''\footnote{A ``proper mixture'' describes a probability distribution over prepared states. It models a preparation procedure whose output is random, with the understanding that on repeated runs the states are sampled independently from this distribution. This matters because ensembles of fixed composition can carry information that independent samples do not. D'Espagnat noted that a collection of $2N$ spin-$\tfrac{1}{2}$ particles, exactly $N$ of which are $z$-spin up and $N$ of which are $z$-spin down, can be distinguished with high confidence from a collection in which exactly $N$ are $x$-spin up and $N$ are $x$-spin down, even though both might naively be assigned the density operator $\tfrac{1}{2}I$ \cite[Sec.~8.3]{d2018veiled}. Measuring $z$-spin on every particle gives exactly $N$ up outcomes for the former collection, but is unlikely to yield exactly $N$ up outcomes for the latter when $N$ is large. This test relies on the fixed composition of the two collections, however. For $2N$ independent samples from $\big\{\big(\tfrac12,\ket{\uparrow_z}\big),\big(\tfrac12,\ket{\downarrow_z}\big)\big\}$ and from $\big\{\big(\tfrac12,\ket{\uparrow_x}\big),\big(\tfrac12,\ket{\downarrow_x}\big)\big\}$, textbook NRQM predicts identical statistics.} A ``pure'' state $\psi$ can be thought of as a special case of a proper mixture, where the probability measure $\mu$ is the Dirac measure $\delta_\psi$ supported at $\psi$. The most familiar proper mixtures are countably supported measures $\mu$, i.e., measures concentrated on a countable set $\{\psi_i\}\subset S(\mathcal{H}_{\mathcal{E}})$. I will denote such a measure $\sum_i p_i\,\delta_{\psi_i}$ by $\{(p_i,\psi_i)\}$, where $\sum_i p_i=1$ and the $\psi_i$ need not be distinct. Realistic preparation procedures, however, do not typically produce proper mixtures of this form. For instance, a somewhat noisy procedure designed to prepare $\ket{\uparrow}$ is likely to produce states spread continuously around $\ket{\uparrow}$ on the Bloch sphere.

Let $\mathcal{E}$ be an HPMM such that for every $O\in \Sigma_{\mathcal{E}}$ the map $\psi\mapsto \mathbb{P}_{\mathcal{E}}(O\,;\psi)$ is Borel measurable on $S(\mathcal{H}_{\mathcal{E}})$. Then we can extend $\mathbb{P}_{\mathcal{E}}$ to a function $\widetilde{\mathbb{P}}_{\mathcal{E}}:\Sigma_{\mathcal{E}}\times \mathrm{M}_1\left(S(\mathcal{H}_{\mathcal{E}})\right)\to [0,1]$ via the formula
\begin{equation}\label{eq:affinity}
\widetilde{\mathbb{P}}_\mathcal{E}(O\,;\mu)\equiv\int_{\psi\in S(\mathcal{H}_{\mathcal{E}})}\!d\mu(\psi)\,\mathbb{P}_\mathcal{E}(O\,;\psi)\,.
\end{equation}
For countably supported proper mixtures $\mu=\{(p_i,\psi_i)\}$, we define\footnote{This definition does not require any measurability assumptions on the OPFs. When the OPFs are measurable, it agrees with Equation \eqref{eq:affinity}.}
$$
\widetilde{\mathbb{P}}_\mathcal{E}(O\,;\mu)\equiv\sum_{i} p_i \,\mathbb{P}_\mathcal{E}(O\,;\psi_i)\,.
$$
Heuristically, $\widetilde{\mathbb{P}}_\mathcal{E}(O\,;\mu)$ is the probability that a two-step random process, like the one in the bunker game, yields an outcome in $O$. First ``randomly sample'' a wave function $\psi$ according to $\mu$, then ``randomly sample'' an outcome according to $\mathbb{P}_{\mathcal{E}}(-\,;\psi)$.

In textbook NRQM, it is common practice to represent elements of $\mathrm{M}_1(S(\mathcal{H}_{\mathcal{E}}))$ using density operators, i.e., elements of the set $\mathcal{D}(\mathcal{H}_{\mathcal{E}})$ of trace-class, positive-semidefinite operators $\mathcal{H}_{\mathcal{E}}\to \mathcal{H}_{\mathcal{E}}$ with trace equal to $1$. This representation can be thought of in terms of the ``mixture map'' $\mathcal{M}:\mathrm{M}_1\big(S(\mathcal{H}_{\mathcal{E}})\big) \to \mathcal{D}(\mathcal{H}_{\mathcal{E}})$ defined by
$$
\mathcal{M}(\mu) = \int_{\psi \in S(\mathcal{H}_{\mathcal{E}})}\!d\mu(\psi)\, \ket{\psi}\!\bra{\psi}\,,
$$
where the integral above is the Bochner integral of the function $\psi \mapsto \ket{\psi}\!\bra{\psi}$ considered as a map from $S(\mathcal{H}_{\mathcal{E}})$ to the Banach space $\mathrm{B}_1(\mathcal{H}_{\mathcal{E}})$ of trace-class operators equipped with the trace norm.\footnote{See \cite{hytonen2016analysis} for an introduction to Bochner integrals, and \cite{holevo2006continuous} for their use in defining the density operator of a general ensemble of quantum states.} That this gives a well-defined element of $\mathcal D(\mathcal H_{\mathcal{E}})$ is shown in Lemma \ref{tcintegral} below.

When $\mu=\{(p_i,\psi_i)\}$, its image $\mathcal{M}(\mu)$ is the familiar density operator $\sum_i p_i \ket{\psi_i}\!\bra{\psi_i}$. Since many different probability measures on $S(\mathcal{H}_{\mathcal{E}})$ have the same image under $\mathcal{M}$, the mixture map loses information about a proper mixture. For finitely supported measures in finite-dimensional settings, the preimage $\mathcal{M}^{-1}(\rho)$ of a density operator can be characterized exactly \cite{Schrodinger1936, gisin1989stochastic, hughston1993complete}. The amount of information lost is substantial. For instance, the space of probability measures on the Bloch sphere is infinite-dimensional, whereas the space of density operators on $\mathcal{H}_{\mathcal{E}}\cong \mathbb{C}^2$ is $3$-dimensional. It is therefore far from obvious that the empirical content of an arbitrary proper mixture is fully captured by its image under $\mathcal{M}$.\footnote{For a nonlinear model in which distinct finitely supported proper mixtures can always be distinguished, even when they have the same density operator, see \cite{haag1978comments}. For discussion of the role this distinction plays in derivations of the linearity of state evolution from no-signaling assumptions, such as \cite{simon2001no}, see \cite{bona2003comment, svetlichny2002critique, ziman2002quantum}.}

\subsubsection{Postulates on proper mixtures}\label{subsec:propmix}

The next two postulates assert that distinguishing ensembles associated with identical density operators (DEIDO) is impossible. A version of this assumption was put forward as Postulate K by Peres in \cite[p.~76]{peres1995quantum}, and is also a consequence of the ``Mixture Equivalence Principle'' formulated by Fedida and Kent \cite{fedida2025mixture}. As we note below, the affine dependence of outcome probabilities on density operators, which is often built into statements of this kind, is automatic.

\begin{postulatecustom}{WSEP}[Standard ensemble-level probabilities for countably supported proper mixtures]\label{WSEP} An HPMM $\mathcal{E}$ has \emph{Property WSEP} if for any $O\in\Sigma_{\mathcal{E}}$, there exists a function $f_O:\mathcal{D}(\mathcal{H}_{\mathcal{E}})\to [0,1]$ such that for all $\mu\in \mathrm{M}_1(S(\mathcal{H}_{\mathcal{E}}))$ of the form $\{(p_i,\psi_i)\}$,
$$
\widetilde{\mathbb{P}}_{\mathcal{E}}(O\,;\mu)=f_O\bigg(\sum_{i} p_i\ket{\psi_i}\!\bra{\psi_i}\bigg)\,.
$$
Under \emph{Postulate WSEP}, all available HPMMs satisfy Property WSEP.
\end{postulatecustom}

Postulate \ref{WSEP} will be sufficient for the purposes of proving Theorem \ref{mainthm}. However, the following more general statement applying to all proper mixtures is arguably more natural, and turns out to be unconditionally equivalent to it (see Figure \ref{Implications}).

\begin{postulatecustom}{SEP}[Standard ensemble-level probabilities for general proper mixtures]\label{SEP} An HPMM $\mathcal{E}$ has \emph{Property SEP} if for any $O\in \Sigma_{\mathcal{E}}$, there is a function $f_O:\mathcal{D}(\mathcal{H}_{\mathcal{E}})\to [0,1]$ such that the following diagram commutes:
$$
\begin{tikzcd}[column sep={2.5cm,between origins}, row sep=0.5cm]
\mathrm{M}_1(S(\mathcal{H}_{\mathcal{E}})) \arrow{rr}{\widetilde{\mathbb{P}}_{\mathcal{E}}(O\,;-)~~~~} \arrow{dr}[swap]{\mathcal{M}} && {[0,1]} \\
& \mathcal{D}(\mathcal{H}_{\mathcal{E}}) \arrow{ur}[swap]{f_O}
\end{tikzcd}
$$
Under \emph{Postulate SEP}, all available HPMMs satisfy Property SEP.\footnote{Note that Property \ref{SEP} requires that $\psi\mapsto\mathbb{P}_{\mathcal{E}}(O\,;\psi)$ is Borel measurable for every $O$, otherwise $\widetilde{\mathbb{P}}_{\mathcal{E}}$ is undefined. Borel measurability is thus part of the assertion, not a background hypothesis. Property \ref{WSEP} carries no such requirement, so the equivalence of the two postulates includes the claim that Postulate \ref{WSEP} implies measurability. This follows by way of Postulates \ref{PL} and \ref{QF}, which together make each OPF a continuous function of $\psi$ (see Lemmas \ref{DIPL}, \ref{PLtoQF}, and \ref{densityReductionGeneral}).}
\end{postulatecustom}

\noindent Postulate \ref{WSEP} is simply the special case of Postulate \ref{SEP} where $\mu$ has countable support.

The function $f_O$ appearing in the above postulates is necessarily affine, a property inherited from the definition of $\widetilde{\mathbb{P}}_\mathcal{E}$ in Equation \eqref{eq:affinity}. To see this, note first that $\mathcal{M}$ is surjective. Indeed, given any $\rho\in \mathcal{D}(\mathcal{H}_\mathcal{E})$, since $\rho$ is trace-class, and hence compact, it has a spectral decomposition of the form $\rho=\sum_{i} p_i\, \ket{\psi_i}\!\bra{\psi_i}$ where $\{\psi_i\}$ is orthonormal and $\sum_i p_i=1$. This shows that $\rho$ is the image under $\mathcal{M}$ of the countably supported element $\{(p_i,\psi_i)\}\in \mathrm{M}_1(S(\mathcal{H}_\mathcal{E}))$. Let $\rho_1,\rho_2\in \mathcal{D}(\mathcal{H}_{\mathcal{E}})$ and $\mu_i\in \mathcal{M}^{-1}(\rho_i)$. The affinity of the mixture map $\mathcal{M}$ then gives that for any $\alpha \in [0,1]$,
\begin{align*}
f_O\big(\alpha \rho_1+(1-\alpha)\rho_2\big)&= f_O\big( \mathcal{M}\big(\alpha \mu_1 +(1-\alpha)\mu_2\big)\big) \stackrel{\textrm{\ref{SEP}}}{=}
\widetilde{\mathbb{P}}_\mathcal{E}(O\,;\alpha \mu_1 +(1-\alpha)\mu_2)\\
&= \int_{\psi\in S(\mathcal{H}_{\mathcal{E}})}\!d\big(\alpha \mu_1+(1-\alpha)\mu_2\big)(\psi)\,\mathbb{P}_\mathcal{E}(O\,;\psi)\\
&=\alpha \int_{\psi\in S(\mathcal{H}_{\mathcal{E}})}\!d\mu_1(\psi)\,\mathbb{P}_\mathcal{E}(O\,;\psi)+(1-\alpha)\int_{\psi\in S(\mathcal{H}_{\mathcal{E}})}\!d\mu_2(\psi)\,\mathbb{P}_\mathcal{E}(O\,;\psi)\\
&\stackrel{\textrm{\ref{SEP}}}{=}\alpha f_O\big(\mathcal{M}(\mu_1)\big)+(1-\alpha)f_O\big(\mathcal{M}(\mu_2)\big)=\alpha f_O(\rho_1)+(1-\alpha)f_O(\rho_2)\,.
\end{align*}
The same argument applies to countably supported mixtures, so the $f_O$ appearing in Postulate \ref{WSEP} must be affine as well. This shows that what looks on the surface like a purely informational restriction (factoring through the mixture map) entails an algebraic restriction (affinity).

\section{Composite Experiments}\label{sec:compositeExps}

Relating single-system measurement postulates to signaling requires formalizing \emph{composite experiments}, i.e., experiments that can be decomposed into component parts, each of which can be understood as an experiment in its own right. In the present framework, a composite experiment is modeled as an HPMM whose outcome space and space of preparations are built from those of its components using product and tensor-product constructions. An HPMM $\mathcal{J}$ will be called \emph{a composite of HPMMs $\mathcal{E}_A$ and $\mathcal{E}_B$} if $\Omega_\mathcal{J}=\Omega_{\mathcal{E}_A}\!\times\Omega_{\mathcal{E}_B}$, $\Sigma_{\mathcal{J}}=\Sigma_{\mathcal{E}_A}\otimes \Sigma_{\mathcal{E}_B}$, and $\mathcal{H}_\mathcal{J}=\mathcal{H}_{\mathcal{E}_A}\otimes \mathcal{H}_{\mathcal{E}_B}$,
i.e., the joint outcome space is the product of the outcome spaces for the component HPMMs, the joint $\sigma$-algebra is the tensor product of the component $\sigma$-algebras,\footnote{The tensor product $\sigma$-algebra is the $\sigma$-algebra generated by elements of the form $O_A\times O_B$ where $O_A\in \Sigma_{\mathcal{E}_A}$ and $O_B\in \Sigma_{\mathcal{E}_B}$.} and the associated Hilbert space is the tensor product of the component Hilbert spaces.

On the level of outcomes, this just amounts to declaring that an outcome of the joint experiment is a pair of outcomes, one for each component experiment. In particular, the event that Bob obtains an outcome in $O_B\in \Sigma_{\mathcal{E}_B}$ corresponds to the set $\Omega_{\mathcal{E}_A}\times O_B$ of joint outcomes, so Bob's outcome probabilities can be read off the composite HPMM. On the level of preparations, however, the assumption that joint preparations are labeled by unit vectors in $\mathcal{H}_{\mathcal{E}_A}\otimes \mathcal{H}_{\mathcal{E}_B}$ is substantive.

This definition constrains only the outcome spaces and Hilbert spaces of $\mathcal{E}_A$ and $\mathcal{E}_B$, not their OPFs. If $\mathcal{J}$ is a composite of $\mathcal{E}_A$ and $\mathcal{E}_B$, it is therefore also a composite of any other pair of HPMMs with the same outcome spaces and Hilbert spaces, whatever their OPFs. Several postulates below, such as Postulate \ref{IP}, relate $\mathbb{P}_\mathcal{J}$ to $\mathbb{P}_{\mathcal{E}_A}$ and $\mathbb{P}_{\mathcal{E}_B}$, so whether they hold for $\mathcal{J}$ depends on which pair of HPMMs it is taken to be a composite of. I will therefore use the term \emph{composite HPMM} for a triple $(\mathcal{J};\mathcal{E}_A,\mathcal{E}_B)$, where $\mathcal{J}$ is a composite of $\mathcal{E}_A$ and $\mathcal{E}_B$. Such a triple is \emph{available} when $\mathcal{J}$ is. When the components are clear from context, I will refer to the triple simply as $\mathcal{J}$.

\subsection{The composition structure of experiments}\label{sec:composition}

In discussions of bipartite experiments, it is typically taken for granted that a pair of experimental procedures, one for each wing of a bipartite setup, can be performed together as a single joint experiment. In Bohm's spin-based variation of the EPR experiment \cite{bohm1955quantum}, for example, one specifies a joint experiment simply by choosing a spin measurement for each particle. The underlying assumption is that if two prepare-and-measure experiments can be performed on mutually isolated systems, then there is a joint experiment consisting of performing both together and recording the pair of outcomes. Physically, one can simply do one experiment on one particle and the other experiment on the other, provided the particles are far enough apart to be experimented on separately. The argument below will only need a special case of this claim, in which one of the two experiments is a textbook qubit measurement and the other is any available experiment. Postulate \ref{EQC} in Section \ref{subsec:semiorthodox} formalizes this special case in terms of HPMMs.

Such discussions likewise typically take for granted that a pair of procedures specifies a \emph{unique} joint experiment. In Bohm's experiment, for instance, choosing a basis for each wing is taken to fix the joint experiment completely. Operationally, this seems uncontroversial. Once one has fully specified what Alice is to do and what Bob is to do, there is nothing left to specify about their joint activities.\footnote{I assume that the initiation times of the two experiments (cf.\ footnote~\ref{fn:initiation}), with respect to agreed-upon time coordinates, are specified as part of the joint procedure. In many models the relative timing of Alice and Bob's actions does not matter, but the present framework does not assume this.} In the present framework, however, a pair of models $\mathcal{E}_A$ and $\mathcal{E}_B$ need not determine a unique model of the joint experiment. The model of Section \ref{subsec:PRmodel} contains two composites with the same components but different OPFs.

\subsection{Post-preparation independence of the components of a composite experiment}\label{sec:PPI}

The present framework also requires a way to formalize the claim that Alice and Bob's experimental procedures are ``independent of one another'' once the preparation phase is complete, a condition I will call \emph{post-preparation independence}. This condition is meant to rule out influence or communication between the two wings, not correlations between their results. A classical example illustrates the difference. Suppose Alice and Bob have each received a vial of a mystery substance from a crime scene. Because their vials come from the same source, we would expect their results to be strongly correlated, at least in forensically useful situations. Yet we would still say that they performed independent analyses, provided neither could influence or communicate with the other.

Influence and communication can take several forms. The most obvious is that Bob learns the outcome of Alice's experiment and changes his procedure accordingly. Another is direct physical interaction. If Alice and Bob were working side by side with charged particles, for instance, the attraction or repulsion between their particles could become non-negligible. A subtler case is a shared power supply, through which the triggering of Alice's apparatus by a detection event could affect how Bob's apparatus behaves. The bunker of Section \ref{sec:TheBunkerGame} is designed to isolate Alice's experiment from Bob's in all of these respects, while still allowing anything they shared beforehand, such as an entangled pair.

Developing a complete taxonomy of the ways general experiments can influence one another is beyond our current scope. The core implications need only one aspect of post-preparation independence, which the following postulate captures. It says that whenever the joint preparation is a product state, Bob's outcome probabilities depend neither on the prepared state of Alice's system nor on the experiment she chooses to perform.

\begin{postulatecustom}{IP}[Independence of one-sided experimental outcome probabilities for product states]\label{IP} A composite HPMM $(\mathcal{J};\mathcal{E}_A,\mathcal{E}_B)$ has \emph{Property \ref*{IP}} if for any $\psi\in S(\mathcal{H}_{\mathcal{E}_A})$, $\phi\in S(\mathcal{H}_{\mathcal{E}_B})$, and $O_B\in \Sigma_{\mathcal{E}_B}$, $\mathbb{P}_{\mathcal{J}}(\Omega_{\mathcal{E}_A}\times O_B\,;\psi\otimes\phi) = \mathbb{P}_{\mathcal{E}_B}(O_B\,;\phi)$. Under \emph{Postulate \ref*{IP}}, every available composite HPMM that describes a post-preparation independent composite experiment has Property \ref*{IP}.
\end{postulatecustom}
In words, Property \ref{IP} says that, when the joint preparation is a product state $\psi\otimes\phi$, Bob's marginal probabilities are those his single-system model assigns to $\phi$. This postulate holds in textbook NRQM when Alice and Bob's experiments are mutually isolated. Indeed, if the joint state is a product state $\psi\otimes\phi$ at some time $t_0$, then, absent interaction terms in the Hamiltonian, it remains a product state at all later times, and Bob's factor evolves exactly as it would if Alice's system were absent. This holds whether a ``measurement'' on either side is modeled unitarily or by collapse, since each acts only on its own factor. Because textbook NRQM represents Bob's outcomes on the joint space by operators that act as the identity on Alice's factor, the Born rule then assigns them the same probabilities that Bob's single-system experiment assigns to $\phi$.

The following stronger version of Postulate \ref{IP} is needed only in Section \ref{sec:PMIndep}:

\begin{postulatecustom}{PM}[Product-state multiplicativity for composite experiments]\label{PM} A composite HPMM $(\mathcal{J};\mathcal{E}_A,\mathcal{E}_B)$ has \emph{Property \ref*{PM}} if for any $\psi\in S(\mathcal{H}_{\mathcal{E}_A})$, $\phi\in S(\mathcal{H}_{\mathcal{E}_B})$, $O_A\in \Sigma_{\mathcal{E}_A}$, and $O_B\in \Sigma_{\mathcal{E}_B}$, $\mathbb{P}_{\mathcal{J}}(O_A\times O_B\,;\psi\otimes\phi) = \mathbb{P}_{\mathcal{E}_A}(O_A\,;\psi)\,\mathbb{P}_{\mathcal{E}_B}(O_B\,;\phi)$. Under \emph{Postulate \ref*{PM}}, every available composite HPMM that describes a post-preparation independent composite experiment has Property \ref*{PM}.
\end{postulatecustom}

\noindent Postulate \ref{PM} implies Postulate \ref{IP}, since
\begin{align*}
\mathbb{P}_{\mathcal{J}}\big(\Omega_{\mathcal{E}_A}\!\times\!O_B\,;\psi\otimes\phi\big)
\stackrel{(\textrm{\ref{PM}})}{=}
\mathbb{P}_{\mathcal{E}_A}\big(\Omega_{\mathcal{E}_A}\,;\psi\big)\,\mathbb{P}_{\mathcal{E}_B}\big( O_B\,;\phi\big)
&=\mathbb{P}_{\mathcal{E}_B}\big( O_B\,;\phi\big)\,.
\end{align*}
Section \ref{subsec:PRmodel} shows that the converse fails, even with the other postulates. In textbook NRQM, Postulate \ref{PM} follows from the standard quantum-mechanical description of composite systems, together with the hypothesis that a joint outcome $O_A\times O_B$ of mutually isolated experiments is represented by the product $E_{\mathcal{E}_A}(O_A)\otimes E_{\mathcal{E}_B}(O_B)$ of the local operators \cite[App.~D]{eberhard1978bell}. This treats the prepared state as the only source of correlation between the two wings. If the apparatuses share an additional resource not encoded in the joint state, Alice's and Bob's outcomes can be correlated even for product preparations. For instance, both detectors might be programmed to ignore the particle and simply display the value of a pre-shared random bit. This would violate Postulate \ref{PM} without any interaction between the wings.\footnote{Such shared resources are fair game in the bunker game, since the outcome of Alice's coin flip is assumed to be independent of anything shared beforehand.}

Postulate \ref{PM} also bears on the uniqueness question raised in Section \ref{sec:composition}. Together with Postulate \ref{QF}, it implies that, for given components, at most one available composite HPMM describes a post-preparation independent composite experiment (Lemma \ref{lem:decompUniqueness}). None of the core implications in Figure \ref{Implications} need this uniqueness claim, however.

\subsection{Semiorthodox experiments}\label{subsec:semiorthodox}

When Alice and Bob conduct a joint experiment in which Alice's side is well described by an orthodox ``measurement,'' I will say that the joint experiment is \emph{semiorthodox}. Formally, a composite HPMM $(\mathcal{J};\mathcal{E}_A, \mathcal{E}_B)$ is \emph{semiorthodox} when $\mathcal{E}_A$ is an orthodox HPMM in the sense of Section \ref{sec:OrthodoxOPFs}. The semiorthodox experiments the argument relies on are those in which Alice's component is a qubit HPMM $\mathcal{E}_{\overline{u}}$ of the kind introduced in that subsection. The argument requires that such composites exist for every direction in which Alice might measure and every experiment available to Bob.

\begin{postulatecustom}{EQC}[Existence of qubit composites]\label{EQC}
For every $\overline{u}\in S^2$ and every available HPMM $\mathcal{E}_B$, there exists an available composite HPMM $(\mathcal{J};\mathcal{E}_{\overline{u}},\mathcal{E}_B)$ describing a post-preparation independent composite experiment.
\end{postulatecustom}

The signaling scenario of Section \ref{subsection:whatworry} and the bunker game of Section \ref{sec:TheBunkerGame} both involve such composites, since Alice performs a textbook spin measurement on a qubit, while Bob performs any experiment he likes on his own system. Postulate \ref{EQC} is the only existence claim the derivation of Postulate \ref{POVM} from Postulate \ref{NS} requires. It is needed because Postulates \ref{NS} and \ref{WNS} are universal claims about whatever composite experiments are available, and would be vacuously satisfied if no such composites were available.

The next postulate concerns how outcome probabilities for such joint experiments relate to outcome probabilities for Bob's single-system experiments.

\begin{postulatecustom}{SMQC}[Standard marginal probabilities for semiorthodox post-preparation independent qubit experiments]\label{SMQC} Let $(\mathcal{J};\mathcal{E}_{\overline{u}},\mathcal{E}_B)$ be a semiorthodox composite HPMM for some $u\in S(\mathbb{C}^2)$. The composite $\mathcal{J}$ has \emph{Property SMQC} if for any $u^\perp\in S(\mathbb{C}^2)$ such that $\langle u, u^\perp\rangle=0$, any $O_B\in \Sigma_{\mathcal{E}_B}$, and any $\Psi \in S\big(\mathbb{C}^2\otimes\mathcal{H}_{\mathcal{E}_B}\big)$, 
$$
\mathbb{P}_{\mathcal{J}}\big(\Omega_{\mathcal{E}_{\overline{u}}}\times O_B\,;\Psi\big)
=
|\alpha|^2\,\mathbb{P}_{\mathcal{E}_B}\left(O_B\,;\psi\right)+|\beta|^2\,\mathbb{P}_{\mathcal{E}_B}\left(O_B\,;\phi \right)\,,
$$
where $\Psi$ is expanded as $\Psi=\alpha\, u \otimes \psi + \beta\,u^\perp \otimes \phi$ with $\alpha,\beta\in \mathbb{C}$ and $\psi,\phi\in S(\mathcal{H}_{\mathcal{E}_B})$. Under \emph{Postulate SMQC}, every available semiorthodox composite HPMM
of this form that describes a post-preparation independent composite experiment has Property SMQC.
\end{postulatecustom}

Postulate \ref{SMQC} is stated operationally in terms of outcome probabilities. It makes no claim about the conditional state of Bob's system after Alice's measurement or about the temporal order of the two measurements. It constrains only Bob's marginal probabilities, relating those in the composite experiment $\mathcal{J}$, in which Alice measures in the basis $\{u,u^\perp\}$ on the joint preparation $\Psi$, to those in the single-system experiment $\mathcal{E}_B$, performed on preparations of $\psi$ and of $\phi$. For example, if $\alpha=\beta=\frac{1}{\sqrt{2}}$, Postulate \ref{SMQC} asserts that, when Alice performs her measurement, Bob's outcome statistics for the joint preparation $\frac{1}{\sqrt{2}}(u\otimes\psi+u^\perp\otimes\phi)$ coincide with those for the proper mixture in which $\psi$ and $\phi$ are each prepared with probability $\frac{1}{2}$, whether or not $\psi$ and $\phi$ are orthogonal.

Postulate \ref{SMQC} holds in textbook NRQM by the reasoning of Section \ref{sec:TheBunkerGame}. Suppose first that Alice completes her measurement before Bob's experiment is initiated, as in the protocol analyzed there. She obtains the outcome corresponding to $u$ (respectively $u^\perp$) with probability $|\alpha|^2$ (respectively $|\beta|^2$), and Bob's system is then prepared in $\psi$ (respectively $\phi$), so his outcome probabilities are those on the right-hand side of Property \ref{SMQC}. Since in textbook NRQM Bob's marginal probabilities do not depend on the order of the two measurements, the same holds whatever their temporal order. As Figure \ref{Implications} shows, Postulate \ref{SMQC} also follows from Postulates \ref{PT} and \ref{IP} (Section \ref{sec:PTtoSQ}), and hence from Postulate \ref{IP} together with POVM universality.

\subsection{Signaling via experimental choice}\label{sec:signalingViaChoice}

Next, we need to formalize \emph{signaling via experimental choice}, which we discussed in informal terms in Section \ref{sec:TheBunkerGame}. The standard formalization of the no-signaling condition \cite{eberhard1978bell, popescu1994quantum} is via an expression of the form $\sum_{x}\mathrm{Prob}(x,y~|~a,b)=\sum_{x}\mathrm{Prob}(x,y~|~a',b)$, where $x$ and $y$ denote outcomes for Alice and Bob, respectively, and $a$, $a'$, and $b$ denote their experimental settings. This amounts to the claim that the marginal probability of Bob getting a given outcome in a joint experiment does not depend on Alice's experimental choice.

The relevant situations are those in which Alice and Bob perform a post-preparation independent composite experiment drawn from a family $\{(\mathcal{J}_\alpha;\mathcal{E}_A^\alpha,\mathcal{E}_B)~|~\alpha\in \mathcal{I}\}$ of composite HPMMs, where Alice's component $\mathcal{E}_A^\alpha$ may vary with the index $\alpha\in\mathcal{I}$ but Bob's component $\mathcal{E}_B$ is fixed. I will call such a family \emph{preparation-compatible} if there is a single Hilbert space $\mathcal{H}_A$ such that $\mathcal{H}_{\mathcal{E}_A^\alpha}=\mathcal{H}_A$ for all $\alpha\in\mathcal{I}$. This ensures that the same joint state $\Psi\in S(\mathcal{H}_A\otimes\mathcal{H}_{\mathcal{E}_B})$ makes sense for every experiment in the family. The definition itself requires only this common Hilbert space, but it is intended to describe situations in which the same joint preparation is used throughout, and Alice chooses among the different members of the family by varying her procedure while Bob's stays fixed, as in the protocols examined in Section \ref{sec:TheBunkerGame}. Phrased in these terms, the standard no-signaling condition reads as follows.

\begin{postulatecustom}{NS}[No signaling via experimental choice]\label{NS} A preparation-compatible family \linebreak ${\{(\mathcal{J}_\alpha;\mathcal{E}_A^\alpha,\mathcal{E}_B)~|~\alpha\in \mathcal{I}\}}$ with common left-hand space of preparations $\mathcal{H}_A$ has \emph{Property NS} if for all $\alpha_1,\alpha_2\in\mathcal{I}$, for all $O_B\in \Sigma_{\mathcal{E}_B}$, and for all $\Psi\in S(\mathcal{H}_A\otimes\mathcal{H}_{\mathcal{E}_B})$,
$$
\mathbb{P}_{\mathcal{J}_{\alpha_1}}(\Omega_{\mathcal{E}_{A}^{\alpha_1}}\!\times O_B\,;\Psi)=\mathbb{P}_{\mathcal{J}_{\alpha_2}}(\Omega_{\mathcal{E}_{A}^{\alpha_2}}\!\times O_B\,;\Psi)\,.
$$
Under \emph{Postulate NS}, every preparation-compatible family of available composite HPMMs, each describing a post-preparation independent composite experiment, has Property NS.  
\end{postulatecustom}

\noindent A preparation-compatible family that lacks Property \ref{NS} will be said to \emph{allow for signaling via experimental choice}. To see why the name is apt, suppose that Property \ref{NS} fails for some preparation-compatible family. Then there exist $\alpha_1,\alpha_2\in\mathcal{I}$, $O_B\in\Sigma_{\mathcal{E}_B}$, and $\Psi\in S(\mathcal{H}_A\otimes\mathcal{H}_{\mathcal{E}_B})$ such that the marginal probabilities $p_i\equiv\mathbb{P}_{\mathcal{J}_{\alpha_i}}\big(\Omega_{\mathcal{E}_{A}^{\alpha_i}}\!\times O_B\,;\Psi\big)$ for $i\in \{1,2\}$ differ, i.e., $\delta\equiv|p_1-p_2|>0$. Suppose further that $\Psi$ can be reliably prepared on as many independent trials as Alice and Bob like, that both know $\Psi$ and the OPFs $\mathbb{P}_{\mathcal{J}_{\alpha_1}}$ and $\mathbb{P}_{\mathcal{J}_{\alpha_2}}$ (e.g., because they have run these experiments many times before), and that Alice can control the parameter $\alpha$.\footnote{The argument given here relies on the operational assumption discussed in Section \ref{sec:composition}, that the joint experiment is fixed by the procedures chosen in each wing, and on the intended reading of preparation-compatibility described above. These assumptions are used only to motivate Postulates \ref{NS} and \ref{WNS} and play no role in the formal mathematical arguments below.} Then Alice can send Bob a one-bit signal with arbitrarily high fidelity as follows. First, they agree that $\mathcal{J}_{\alpha_1}$ corresponds to the bit ``$\mathrm{0}$'' and $\mathcal{J}_{\alpha_2}$ to the bit ``$\mathrm{1}$,'' and pick a (perhaps quite large) number $N$ of trials. Alice picks the bit she intends to communicate and selects $\alpha$ accordingly, keeping it fixed across the $N$ trials. Bob records the number $n(O_B)$ of trials on which his outcome lies in $O_B$, and guesses ``$\mathrm{0}$'' if $n(O_B)/N$ is closer to $p_1$ than to $p_2$, and ``$\mathrm{1}$'' otherwise. By the law of large numbers, for any $\eta>0$ there exists $N$ such that if the joint HPMM is $\mathcal{J}_{\alpha_i}$, then $n(O_B)/N$ lies within $\delta/3$ of $p_i$ with probability at least $1-\eta$, so Bob's guess is wrong with probability at most $\eta$.

To derive POVM universality, we only need the following weak version of Postulate \ref{NS}, which forbids signaling via experimental choice when Alice performs a qubit measurement.

\begin{postulatecustom}{WNS}[No signaling via experimental choice for semiorthodox experiments where Alice controls a qubit]\label{WNS}
Under \emph{Postulate WNS}, every preparation-compatible family of available composite HPMMs of the form $\{(\mathcal{J}_\alpha;\mathcal{E}_{\overline{u}_\alpha},\mathcal{E}_B)~|~\alpha\in\mathcal{I}\}$, each describing a post-preparation independent composite experiment, has Property \ref{NS}.
\end{postulatecustom}

\subsection{The partial trace}\label{subsec:PT}

The final postulate asserts that the outcome probabilities for any experiment on Bob's system can be predicted by applying an affine function to the partial trace of the joint prepared state over Alice's subsystem. Recall that the partial trace over $\mathcal{H}_A$ is the linear map $\mathrm{Tr}_A$ from trace-class operators on $\mathcal{H}_A\otimes\mathcal{H}_B$ to trace-class operators on $\mathcal{H}_B$ satisfying $\mathrm{Tr}\big(\mathrm{Tr}_A(T)\,X\big)=\mathrm{Tr}\big(T\,(I_A\otimes X)\big)$ for all trace-class $T$ and bounded $X$ on $\mathcal{H}_B$ \cite[Def.~2.68]{heinosaari2011mathematical}.

\begin{postulatecustom}{PT}[Partial trace sufficiency for one-sided outcome probabilities]\label{PT} A preparation-compatible family $\{(\mathcal{J}_\alpha;\mathcal{E}_A^\alpha,\mathcal{E}_B)~|~\alpha\in\mathcal{I}\}$ has \emph{Property PT} if for any $O_B\in\Sigma_{\mathcal{E}_B}$, there exists an affine function $f_{O_B}$ from the space of density operators $\mathcal{D}(\mathcal{H}_{\mathcal{E}_B})$ to $[0,1]$ such that for any $\alpha \in \mathcal{I}$ and any $\Psi\in S(\mathcal{H}_A\otimes \mathcal{H}_{\mathcal{E}_B})$, $\mathbb{P}_{\mathcal{J}_\alpha}(\Omega_{\mathcal{E}_A^\alpha}\times O_B\,;\Psi) = f_{O_B}\big(\mathrm{Tr}_A(\ket{\Psi}\!\bra{\Psi})\big)$. Under \emph{Postulate PT}, every preparation-compatible family of available composite HPMMs, each describing a post-preparation independent composite experiment, has Property PT.
\end{postulatecustom}

\noindent Note that, unlike in Postulates \ref{WSEP} and \ref{SEP}, the affinity of $f_{O_B}$ is not automatic and must be explicitly assumed. Property \ref{PT} concerns only pure joint preparations, and nothing relates the outcome probabilities of different pure states whose reduced states are convex combinations of one another. Affinity is needed in Section \ref{sec:PTtoSQ}.

The common assumption that the partial trace correctly describes one part of a multipartite quantum system has strong consequences that are easy to miss, since the reduced density matrix is routinely treated as a complete description of a subsystem. This treatment is often the starting point of no-signaling proofs (see, e.g., \cite{ghirardi2013entanglement, simon2001no}). Yet justifying the claim that the partial trace provides a ``complete and exhaustive description of the measurement statistics'' of a subsystem \cite[p.~45]{schlosshauer2007decoherence} requires much of the standard theory of quantum measurement. As Zurek puts it in motivating his derivation of the Born rule from ``envariance,'' ``Both the physical interpretation of the trace and of the reduced density matrix are justified assuming Born's rule,'' so a derivation of the Born rule cannot rely on them without circularity \cite{zurek2005probabilities}.\footnote{For a discussion of the assumptions Zurek's derivation uses in their place, see \cite{schlosshauer2003zurek}.} The same holds for Postulate \ref{PT}, which, given Postulates \ref{EQC} and \ref{IP}, is equivalent to POVM universality (see Figure \ref{Implications}). Attempts to derive the Born rule or its generalizations should therefore not take it to be a neutral assumption.\footnote{As an example, Sebens and Carroll's derivation of the Born rule in Everettian quantum mechanics relies on a principle formulated in terms of reduced density matrices, which they call ESP-QM \cite[pp.~42--43]{sebens2018self}. Dawid and Friederich argue that ESP-QM cannot be motivated independently of the Born rule \cite{dawid2020epistemic}.}

\section{Discussion}\label{sec:discussion}

\subsection{Two readings of Figure \ref*{Implications}}

The web of implications in Figure \ref{Implications} lends partial support to the view that quantum mechanics is an ``island in theoryspace'' \cite{aaronson2004quantum}, or, at the very least, ``a tighter package than one might have first thought'' \cite{fuchs2011coming}. Given relatively uncontroversial background assumptions, many common quantum-mechanical postulates are mutually equivalent. We can therefore axiomatize the theory in a variety of ways depending on which assumptions we take to be most fundamental, and because these alternative axioms are so closely interrelated, many different choices converge on essentially the same theory. On this view, the no-signaling hypothesis is a natural candidate for an axiom. It rests on clear physical intuitions, and although it is not strong enough to recover all of textbook NRQM (see Section \ref{sec:PMIndep}), it is strong enough to derive much of its probability-generating fragment, in particular the generalized Born rule. Figure \ref{Implications} then suggests that the empirical core of NRQM is, if not inevitable, derivable from well-motivated foundations.

At the same time, Figure \ref{Implications} shows that many common assumptions are more closely tied to the measurement postulates of textbook NRQM than they appear. Since the measurement postulates are the aspect of NRQM that theorists tend to question most, one can interpret Figure \ref{Implications} as suggesting that we should be wary of the whole package. For instance, the idea that density operators ``generalize'' pure states is unproblematic if one simply means that the map $\mathcal{P}(\mathcal{H})\to \mathcal{D}(\mathcal{H})$ given by $\overline{\psi}\mapsto \ket{\psi}\!\bra{\psi}$ is injective. But the claim that \emph{probability measures} over pure states can be replaced by their associated density operators without losing ``observable information'' (Postulate \ref{SEP}) is more loaded. Given textbook NRQM's theory of ``observability,'' this claim holds. Yet in the present framework, Postulate \ref{SEP} is equivalent to POVM universality outright (see Figure \ref{Implications}), so it can be seen as a measurement postulate in disguise.

Similar commitments can also enter more subtly through the assumptions used to derive formalisms that relax NRQM's idealizations. For instance, the description of open quantum systems by the Lindblad master equation is often presented as a generalization of Schr\"odinger dynamics that removes the idealized assumption that the system under consideration is closed \cite{cuffaro2021open}. Deriving this master equation, however, typically relies on tracing out environmental degrees of freedom \cite{davies1974markovian}. Again, if one accepts NRQM's theory of observables, this is justifiable under appropriate conditions, because taking the partial trace does not destroy any ``observable information.'' Promoting this use of the partial trace to an axiom yields something like Postulate \ref{PT}, which, together with Postulates \ref{EQC} and \ref{IP}, is strong enough to derive all of the foreground postulates other than Postulate \ref{PM} (see Section \ref{subsec:PT}). So, while the open-systems approach removes the idealized assumption of perfect isolation, it also builds contested aspects of the measurement postulates into the formalism.

Both views described above are reasonable reactions to the strong interrelationships shown in Figure \ref{Implications}. What one theorist sees as rigidity, another may see as brittleness. Some or all of these postulates may well be correct. But whichever assumptions we decide to take as foundational, we should be ready to defend them on their own terms.

\subsection{The parallelogram law as a Rorschach test}

Empirical questions related to the postulates in Figure \ref{Implications} can seem either exciting or mundane, depending on how one looks at them. Consider the arrival-time problem discussed in Section \ref{subsec:observability}. On its face, this is an unremarkable scientific question about how to best describe certain empirical data. Some proposed models predict arrival-time distributions that are not given by any POVM. If these models turned out to provide the best fit to the data, would that really be a paradigm-shattering discovery? To theorists who treat the postulates in Figure \ref{Implications} as a unified package, the suggestion of non-POVM statistics can seem to call into question ``quantum mechanics itself,'' or ``causality itself,'' or some other principle that only a radical skeptic could doubt. But the stakes need not be so dramatic. Non-POVM statistics seem to threaten principles like causality only when we bundle together assumptions from different physical theories and insist that each holds exactly and universally.

Postulate \ref{PL} offers a way to make such commitments explicit. Its simplest special case, Postulate \ref{WPL}, asserts that equalities of the following form hold whenever $\psi$ and $\phi$ are orthogonal unit vectors:
\begin{equation}\label{eq:PLReprise}
\mathbb{P}_{\mathcal{E}}\left(O\,;\psi\right) +  \mathbb{P}_{\mathcal{E}}\left(O\,;\phi\right) =\mathbb{P}_{\mathcal{E}}\left(O\,;\tfrac{1}{\sqrt{2}}(\psi + \phi)\right) +  \mathbb{P}_{\mathcal{E}}\left(O\,;\tfrac{1}{\sqrt{2}}(\psi - \phi)\right)\,.
\end{equation}
This is an \emph{exact} claim about the relationship between four real numbers: the probabilities that a given experiment modeled by $\mathcal{E}$ yields an outcome in $O$ for four different preparation procedures labeled by $\psi$, $\phi$, $\tfrac{1}{\sqrt{2}}(\psi+\phi)$, and $\tfrac{1}{\sqrt{2}}(\psi-\phi)$. Equation \eqref{eq:PLReprise} applies to \emph{every} prepare-and-measure experiment, \emph{every} outcome set $O$, and \emph{every} pair of preparation procedures with orthogonal labels $\psi$ and $\phi$. Note that the four procedures can be carried out independently of one another, perhaps on different days, or even in different laboratories. Yet their outcome probabilities are supposed to satisfy this equality \emph{exactly}. This is a strikingly precise and general family of empirically falsifiable claims.

This family of claims cannot simply be dismissed as interpretation-dependent. The outcome probabilities it relates are precisely those to which the Born rule applies, so a theorist who accepts that the Born rule makes testable predictions has everything needed to unambiguously interpret Postulate \ref{PL}. And since the Born rule, or its POVM generalization, is typically taken to hold exactly, a theorist who takes it this way is thereby committed to Postulate \ref{PL} holding exactly as well.

Two attitudes toward Postulate \ref{PL} are available, and both are respectable. The first is to see Postulate \ref{PL} as an empirical regularity that we expect to hold to a good approximation in a certain regime, but that future theories may modify. This attitude is familiar enough to need no elaboration. The second is to see Postulate \ref{PL} as an exact, universal law. The tight-package reading of Figure \ref{Implications} can make this attitude seem natural, since Postulate \ref{PL} is logically equivalent to many other standard postulates. This position nevertheless needs more explanation, since promoting an empirical regularity to the status of a physical law is a much more consequential decision than saying that a certain regularity holds, to use Bell's terminology, ``for all practical purposes'' (FAPP) \cite{bell1990against}. 

One way to argue that Postulate \ref{PL} deserves law-like status would be to show that NRQM (including its measurement postulates) can be recovered as an effective limit of a deeper theory in which some sort of gauge equivalence underwrites Equation \eqref{eq:PLReprise}. This deeper theory would, perhaps, reveal that the two sides of this equation are different mathematical representations of the same physics. Relativistic principles might seem to supply such an equivalence. One might note, for instance, that when Alice and Bob are spacelike separated, different reference frames disagree about whether Alice's measurement occurs first. So, simply insisting that every inertial frame tell a consistent story in which causes precede their effects seems to force Bob's outcome statistics to be the same regardless of Alice's choice. In the case of the bunker game, however, Bob can carry out his entire experiment inside the future light cone of Alice's actions, in which case every frame agrees that Alice measures first. The bunker game may well be unwinnable, but if so, relativistic symmetries alone cannot explain why.

Whether Postulate \ref{PL} and its logically equivalent single-system relatives are exact laws of nature or merely good approximations in certain regimes should, in my view, be considered an open question. It is, in the end, an empirical question, and theorizing alone cannot settle it. What theoretical results like those in this paper can do is help us think through our commitments before the data come in. So take a moment to ask yourself the following. Do you believe that Postulate \ref{PL} holds \emph{exactly}? More precisely, to the extent the relevant states can be prepared and systematic errors can be eliminated, do you believe that any instance of Postulate \ref{PL} will be found to hold to any desired degree of precision? Suppose experimental data showed that an instance of Equation \eqref{eq:PLReprise} held to ten decimal places, but not twelve. Would that fundamentally challenge your understanding of physics?

\section*{Acknowledgments}
\phantomsection
\addcontentsline{toc}{section}{Acknowledgments}
I am grateful to Justin Clarke-Doane, Larry Frolov, and Fabian Nolte for comments on earlier drafts.

\pagebreak
\part{Proofs}\label{part:proofs}

\section{Proving the Implications in Figure \ref*{Implications}}\label{mathImp}

This section gives the proofs of the labeled implications in Figure \ref{Implications} other than Postulate \ref{PM} $\Rightarrow$ Postulate \ref{IP}, which was shown in Section \ref{sec:PPI}. Derivations of the Born rule or POVM universality from no-signaling or locality considerations, which roughly correspond to the path from Postulate \ref{NS} to Postulate \ref{POVM} in the figure, have been widely studied (e.g., \cite{zurek2003environment, barnum2003no, zurek2005probabilities, nenashev2014blackbox, nenashev2014quantum, nenashev2016state, han2016quantum, joshi2016violation, mcqueen2019defence}, \cite[Sec.~VI.B]{sekatski2015no}). No-signaling has also been used more broadly to constrain the structure of quantum theory and the limits of information-processing tasks (e.g., \cite{clifton2003characterizing, bae2011no}). Where I am aware of an argument in the literature that is particularly close to one below, I provide a reference at the relevant step.

\subsection{Postulate \ref*{WNS} implies Postulate \ref*{PL}}\label{sec:WNStoPL}

The argument below uses ingredients related to ``steering'' \cite{Schrodinger1936,hughston1993complete} that appear throughout the literature, e.g., in arguments for the linearity of quantum dynamics \cite{gisin1989stochastic, simon2001no}. The earliest close precedent for Lemma \ref{imp1} I am aware of is an argument of Svetlichny \cite{svetlichny1998quantum}, which relies on Gleason's theorem and requires dimension at least three.

\begin{lemma}\label{imp1} Assuming Postulates \ref{SMQC} and \ref{EQC}, Postulate \ref{WNS} implies Postulate \ref{PL}.
\end{lemma}

\begin{proof}
Let $\mathcal{E}_B$ be an available HPMM, and let $O_B\in \Sigma_{\mathcal{E}_B}$. Let $\mathcal{H}_A=\mathbb{C}^2$, let $\{u,u^\perp\}$ be an orthonormal basis for $\mathcal{H}_A$, and let $U = \left(\begin{smallmatrix}
a & b \\
c & d
\end{smallmatrix}\right)$ be a unitary matrix, which we identify with a linear transformation $\mathcal{H}_A\to \mathcal{H}_A$ using the basis $\{u,u^\perp\}$.

Applying Postulate \ref{EQC} to the rays $\overline{u}$ and $U^\dagger\overline{u}$ (where $U^\dagger\overline{u} \equiv \overline{U^\dagger u}$), we obtain available composite HPMMs $\mathcal{J}_{\overline{u}}$ and $\mathcal{J}_{U^\dagger\overline{u}}$, each describing a post-preparation independent composite experiment, with right-hand component $\mathcal{E}_B$ and left-hand components $\mathcal{E}_{\overline{u}}$ and $\mathcal{E}_{U^\dagger\overline{u}}$, respectively. Let $\psi,\phi \in S(\mathcal{H}_{\mathcal{E}_B})$ be such that $a\psi+b\phi\ne 0$ and $c\psi+d\phi\ne0$, and let $\Psi\in S\big(\mathbb{C}^2\otimes \mathcal{H}_{\mathcal{E}_B}\big)$ be given by $\frac{1}{\sqrt{2}}\left(u\otimes \psi + u^\perp\otimes \phi\right)$. Consider the composite HPMM $\mathcal{J}_{\overline{u}}$. It is available and semiorthodox, since its left-hand component is orthodox, and by construction it describes a post-preparation independent composite experiment. Postulate \ref{SMQC} therefore applies, and we have
\begin{align}\label{uProb}
\mathbb{P}_{\mathcal{J}_{\overline{u}}}\big(\Omega
_{\mathcal{E}_{\overline{u}}}\times O_B\,;\Psi\big)
&\stackrel{(\mathrm{\ref{SMQC}})}{=}
\frac{1}{2}\,\mathbb{P}_{\mathcal{E}_B}\big(O_B\,;\psi\big)+\frac{1}{2}\,\mathbb{P}_{\mathcal{E}_B}\big(O_B\,;\phi\big)\,.
\end{align}
If, on the other hand, we consider the HPMM $\mathcal{J}_{U^\dagger \overline{u}}$, then 
Alice's orthodox component is a measurement with respect to the basis
$\big\{U^\dagger u,U^\dagger u^\perp\big\}$, which is orthonormal because $U$ is unitary. The HPMM $\mathcal{J}_{U^\dagger\overline u}$ is likewise available and semiorthodox, and by construction it describes a post-preparation independent composite experiment, so Postulate \ref{SMQC} applies to $\mathcal{J}_{U^\dagger\overline u}$ as well. Since $Uu = au + cu^\perp$ and $Uu^\perp = bu + du^\perp$, applying $U^\dagger$ gives $u = a\,U^\dagger u + c\,U^\dagger u^\perp$ and $u^\perp = b\,U^\dagger u + d\,U^\dagger u^\perp$, so
\begin{align}\label{uDag}
u\otimes \psi + u^\perp \otimes \phi = U^\dagger u \otimes (a\psi + b \phi) + U^\dagger u^\perp \otimes (c \psi + d \phi)\,.
\end{align}
We can rewrite $\Psi$ in the following form:
\begin{align*}
\Psi
&=
\frac{1}{\sqrt{2}}\left(u\otimes \psi + u^\perp\otimes \phi\right)~
\stackrel{\eqref{uDag}}{=}~\frac{1}{\sqrt{2}}\left(
U^\dagger u\otimes (a \psi + b \phi) + U^\dagger u^\perp\otimes (c \psi + d \phi)\right)\\
&=\frac{1}{\sqrt{2}}\left(\|a \psi + b \phi\|\,U^\dagger u\otimes \frac{a \psi + b \phi}{\|a \psi + b \phi\|} + \|c \psi + d \phi\|\,U^\dagger u^\perp\otimes \frac{c \psi + d \phi}{\|c \psi + d \phi\|}\right)\,.
\end{align*}
Applying Postulate \ref{SMQC} again gives us 
\begin{align}\label{UdaguProb}
\mathbb{P}_{\mathcal{J}_{U^\dagger \overline{u}}}\left(\Omega_{\mathcal{E}_{U^\dagger\overline{u}}}\!\times O_B\,;\Psi\right)
&\stackrel{(\mathrm{\ref{SMQC}})}{=}
\frac{\|a \psi + b \phi\|^2}{2}\,\mathbb{P}_{\mathcal{E}_B}\left(O_B\,;\frac{a \psi + b \phi}{\|a \psi + b \phi\|}\right)\nonumber\\
&\qquad\qquad+\frac{\|c \psi + d \phi\|^2}{2}\,\mathbb{P}_{\mathcal{E}_B}\left(O_B\,;\frac{c \psi + d \phi}{\|c \psi + d \phi\|}\right)\,.
\end{align}
The pair $\{\mathcal{J}_{\overline u},\mathcal{J}_{U^\dagger\overline u}\}$ is preparation-compatible, since both left-hand components have preparation space $\mathbb{C}^2$, and it satisfies the remaining hypotheses of Postulate \ref{WNS} by construction, so the left-hand sides of \eqref{uProb} and \eqref{UdaguProb} are equal. Equating the right-hand sides and multiplying through by $2$ yields
\begin{align*}
\mathbb{P}_{\mathcal{E}_B}\!&\left(O_B\,;\psi\right)+\mathbb{P}_{\mathcal{E}_B}\!\left(O_B\,;\phi\right)\\
&=\|a \psi + b \phi\|^2\,\mathbb{P}_{\mathcal{E}_B}\!\left(O_B\,;\frac{a \psi + b \phi}{\|a \psi + b \phi\|}\right)+\|c \psi + d \phi\|^2\,\mathbb{P}_{\mathcal{E}_B}\!\left(O_B\,;\frac{c \psi + d \phi}{\|c \psi + d \phi\|}\right)\,.~\qedhere
\end{align*}
\end{proof}

\subsection{Postulate \ref*{PL} implies Postulate \ref*{QF}}\label{sec:PLtoQF}

The central ingredients of the proof of this implication can be phrased in terms of abstract functions $f:S(\mathcal{H})\to \mathbb{R}$ on the unit sphere of a Hilbert space $\mathcal{H}$. Following our terminology in Section \ref{OPFPosts}, we will say that such a function \emph{satisfies the parallelogram law} if, for any unitary matrix $U  = \left(\begin{smallmatrix} a & b \\ c & d \end{smallmatrix}\right)$ and any vectors $\psi$ and $\phi$ in $S(\mathcal{H})$ such that $a\psi+b\phi$ and $c\psi+d\phi$ are non-zero,
\begin{align*}
f\left(\psi\right) &+ f\left(\phi\right)=
\| a \psi + b \phi\|^2 \, f\left(\frac{a\psi + b\phi}{\|a\psi+b\phi\|}\right) + \|c \psi + d \phi\|^2 \,f\left(\frac{c\psi + d\phi}{\|c \psi + d\phi\|}\right).
\end{align*}
The strategy in this subsection is as follows. We first extend $f$ from the unit sphere to a function $Q_f$ on all of $\mathcal{H}$ that is homogeneous of degree $2$, and show that $Q_f$ satisfies the familiar parallelogram identity $Q_f(\psi+\phi)+Q_f(\psi-\phi)=2\,Q_f(\psi)+2\,Q_f(\phi)$. An argument of Jordan and von Neumann then shows that the form obtained from $Q_f$ by polarization is sesquilinear over $\mathbb{Q}[i]$, and boundedness, via Ostrowski's theorem on additive functions, extends sesquilinearity to all of $\mathbb{C}$.

\begin{lemma}\label{PLtoPhaseInvt} Any function $f:S(\mathcal{H})\to \mathbb{R}$ satisfying the parallelogram law is phase invariant.
\end{lemma}

\begin{proof}
Let $\theta\in\mathbb{R}$ and $\psi,\phi\in S(\mathcal{H})$. For $U=\mathrm{diag}(e^{i\theta},1)$, we have $a\psi+b\phi=e^{i\theta}\psi$ and $c\psi+d\phi=\phi$, both unit vectors, so the parallelogram law reads $f(\psi)+f(\phi)=f(e^{i\theta}\psi)+f(\phi)$. Subtracting $f(\phi)$ gives $f(\psi)=f(e^{i\theta}\psi)$.
\end{proof}

Let $Q_f:\mathcal{H}\to \mathbb{R}$ be defined by $Q_f(\psi)=\|\psi\|^2 \, f\left(\frac{\psi}{\|\psi\|}\right)$ when $\psi\ne 0$, and $0$ otherwise.

\begin{lemma}\label{lem:homog} If $f$ satisfies the parallelogram law, then $Q_f$ is absolutely homogeneous of degree $2$, i.e., for all $\lambda \in \mathbb{C}$, $Q_f(\lambda\,\psi)=|\lambda|^2 Q_f(\psi)$.
\end{lemma}

\begin{proof} Let $\lambda \in \mathbb{C}$. The cases $\psi=0$ or $\lambda=0$ are immediate, so suppose $\psi\neq 0$ and $\lambda\neq 0$. Since $f$ is phase invariant by Lemma \ref{PLtoPhaseInvt},
\[
Q_f(\lambda \psi) = |\lambda|^2 \|\psi\|^2\, f\!\left(\frac{\lambda}{|\lambda|}\frac{\psi}{\|\psi\|}\right) = |\lambda|^2 \|\psi\|^2\, f\!\left(\frac{\psi}{\|\psi\|}\right)=|\lambda|^2 Q_f(\psi)\,.\qedhere
\]
\end{proof}

\begin{lemma}\label{lem:unnormPL} If $f$ satisfies the parallelogram law, then for any unitary matrix $U =  \left(\begin{smallmatrix}
a & b \\
c & d
\end{smallmatrix}\right)$, and $\psi,\phi\in\mathcal{H}$ such that $a\psi+b\phi\ne0$ and $c\psi+d\phi\ne0$, 
$$
Q_f(\psi)+Q_f(\phi)= Q_f(a\psi+b\phi)+ Q_f(c\psi+d\phi)\,.
$$
\end{lemma}

\begin{proof} If $\psi=0$, then $\phi\ne 0$ by hypothesis. Lemma \ref{lem:homog}, together with the fact that the columns of $U$ are orthonormal so $|b|^2+|d|^2=1$, gives
$$
Q_f(a\psi+b\phi)+Q_f(c\psi+d\phi)=Q_f(b\phi)+Q_f(d\phi)=\left(|b|^2+|d|^2\right)Q_f(\phi)=Q_f(\psi)+Q_f(\phi)\,.
$$
The case $\phi=0$ holds by an identical argument. 

Suppose now that $\psi\ne0$ and $\phi\ne 0$. The idea from here is to rescale $\psi$ and $\phi$ and then mix them into a pair of unit vectors $\eta_\pm$, to which the parallelogram law for $f$ applies directly. Applying the parallelogram law twice relates both sides of the desired identity to $f(\eta_+)+f(\eta_-)$. The rest of the proof carries out this strategy. 

Let $\ell = \tfrac{1}{\sqrt{2}}\sqrt{\|\psi\|^2+\|\phi\|^2}$, $\psi'=\frac{\psi}{\ell}$, and $\phi'=\frac{\phi}{\ell}$, so that $\|\psi'\|^2+\|\phi'\|^2=2$. Choose $\theta\in\mathbb{R}$ such that $e^{i\theta}\langle\psi',\phi'\rangle$ is purely imaginary, let $W$ be the unitary matrix $\tfrac{1}{\sqrt2}\left(\begin{smallmatrix} 1 & e^{i\theta} \\ 1 & -e^{i\theta}\end{smallmatrix}\right)$, and let $w_{ij}$ be the coefficients of $W$. Define $\eta_+$ and $\eta_-$ by $W \big(\begin{smallmatrix}\psi' \\ \phi'\end{smallmatrix}\big) =\big(\begin{smallmatrix}\eta_+\phantom{'} \\ \eta_-\phantom{'}\end{smallmatrix}\!\big)$. Then 
\begin{equation}\label{eq:unitvec}
\|\eta_\pm\|^2=
\frac{\|\psi'\|^2+\|\phi'\|^2}{2}\pm \mathrm{Re}\big(e^{i\theta}\langle\psi',\phi'\rangle\big)=1\,,
\end{equation}
so $\eta_+$ and $\eta_-$ are unit vectors (but need not be orthogonal). From the unitarity of $W^\dagger$, we have the identity
\begin{equation}\label{eq:Winv}
\begin{pmatrix}w_{11}^\ast \eta_+ + w_{21}^\ast \eta_- \\ w_{12}^\ast \eta_+ + w_{22}^\ast \eta_- \end{pmatrix}=W^\dagger \begin{pmatrix}\eta_+ \\ \eta_-\end{pmatrix}=W^\dagger W\begin{pmatrix}\psi' \\ \phi'\end{pmatrix}=\begin{pmatrix}\psi' \\ \phi'\end{pmatrix}\,.
\end{equation}
It follows that $w_{11}^\ast \eta_+ + w_{21}^\ast\eta_-$ and $w_{12}^\ast \eta_+ + w_{22}^\ast\eta_-$ equal $\psi'$ and $\phi'$ respectively, which are non-zero because $\psi$ and $\phi$ are non-zero. Applying the parallelogram law for $f$, we have
\begin{align*}
f\left(\eta_+\right)+f\left(\eta_-\right)&~\stackrel{\eqref{eq:unitvec}}=~\|\eta_+\|^2f\left(\frac{\eta_+}{\|\eta_+\|}\right)+\|\eta_-\|^2f\left(\frac{\eta_-}{\|\eta_-\|}\right)\\
&\stackrel{(\mathrm{PL})}{=} \|w_{11}^\ast \eta_++w_{21}^\ast\eta_-\|^2f\left(\frac{w_{11}^\ast \eta_++w_{21}^\ast\eta_-}{\|w_{11}^\ast \eta_++w_{21}^\ast\eta_-\|}\right)\\
&\qquad\qquad+\|w_{12}^\ast \eta_++w_{22}^\ast\eta_-\|^2f\left(\frac{w_{12}^\ast \eta_++w_{22}^\ast\eta_-}{\|w_{12}^\ast \eta_++w_{22}^\ast\eta_-\|}\right)\\
&~\stackrel{\eqref{eq:Winv}}{=}~ \|\psi'\|^2f\left(\frac{\psi'}{\|\psi'\|}\right)+\|\phi'\|^2f\left(\frac{\phi'}{\|\phi'\|}\right)=Q_f( \psi')+Q_f(\phi')\,.\numberthis \label{eq:psiphieta}
\end{align*}
The matrix $V\equiv UW^{-1}$ is unitary. Writing $v_{ij}$ for the coefficients of $V$,
$$
\begin{pmatrix}v_{11} \eta_+ + v_{12} \eta_-\\ v_{21} \eta_+ + v_{22} \eta_-\end{pmatrix} = V\begin{pmatrix}\eta_+ \\ \eta_-\end{pmatrix}=UW^{-1}W\begin{pmatrix}\psi'\\ \phi'\end{pmatrix}=U\begin{pmatrix}\psi'\\ \phi'\end{pmatrix}=\begin{pmatrix}a\psi'+b\phi'\\ c\psi'+d\phi'\end{pmatrix}\,.
$$
By hypothesis, $a\psi+b\phi \ne 0$ and $c\psi+d\phi \ne 0$. Since $\psi' = \psi/\ell$ and $\phi' = \phi/\ell$ for a non-zero scalar $\ell$, it follows that $a\psi'+b\phi'$ and $c\psi'+d\phi'$ are also non-zero. 

Now, applying the parallelogram law for $f$ once again to the unit vectors $\eta_+$ and $\eta_-$ using the unitary matrix $V$, we have
\begin{align*}
f\left(\eta_+\right)&+f\left(\eta_-\right)\\
&=
\|v_{11} \eta_+ + v_{12} \eta_- \|^2 f\left(\frac{v_{11} \eta_+ + v_{12} \eta_-}{\|v_{11} \eta_+ + v_{12} \eta_-\|}\right) +\|v_{21} \eta_+ + v_{22} \eta_- \|^2 f\left(\frac{v_{21} \eta_+ + v_{22} \eta_-}{\|v_{21} \eta_+ + v_{22} \eta_-\|}\right)\\
&= Q_f(v_{11} \eta_+ + v_{12} \eta_- )+ Q_f(v_{21} \eta_+ + v_{22} \eta_-)\,.
\end{align*}
Plugging the identities for the $V$-combinations into this equation and using Equation \eqref{eq:psiphieta}, we have $Q_f(\psi')+Q_f(\phi') = f(\eta_+)+f(\eta_-) = Q_f(a\psi'+b\phi' )+ Q_f(c\psi'+d\phi')$. Since $Q_f$ is homogeneous by Lemma \ref{lem:homog}, multiplying through by $\ell^2$ and simplifying yields
\[
Q_f(\psi)+Q_f(\phi) =Q_f(a\psi+b\phi )+ Q_f(c\psi+d\phi)\,.\qedhere
\]
\end{proof}

\begin{lemma}\label{PLSum} If $f: S(\mathcal{H})\to \mathbb{R}$ satisfies the parallelogram law, then for all $\psi$ and $\phi$ in $\mathcal{H}$, $Q_f(\psi+\phi)+ Q_f(\psi-\phi)= 2\,Q_f(\psi)+ 2\,Q_f(\phi)$.
\end{lemma}

\begin{proof} If $\phi=\pm\psi$, both sides equal $4Q_f(\psi)$ by Lemma~\ref{lem:homog}. Otherwise, $\psi+\phi\ne0$ and $\psi-\phi\ne0$, so applying Lemma~\ref{lem:unnormPL} with the Hadamard matrix $U = \frac{1}{\sqrt{2}}
\left(\begin{smallmatrix}
1 & \phantom{-}1 \\
1 & -1
\end{smallmatrix}\right)$ and using the homogeneity of $Q_f$, we get
$$
Q_f(\psi) + Q_f(\phi) =Q_f\left(\frac{\psi+\phi}{\sqrt{2}}\right)  +Q_f\left(\frac{\psi-\phi}{\sqrt{2}}\right)= \frac{1}{2}Q_f(\psi+\phi) + \frac{1}{2}Q_f(\psi-\phi)\,.
$$
Multiplying by 2 gives the standard parallelogram identity 
\[
2Q_f(\psi) + 2Q_f(\phi) = Q_f(\psi+\phi) + Q_f(\psi-\phi)\,.\qedhere
\]
\end{proof}

\noindent In the following lemmas, we will let $B_f: \mathcal{H} \times \mathcal{H} \to \mathbb{C}$ be defined by
$$
B_f(\psi,\phi)= \frac{1}{4}\Big(Q_f\!\left(\psi+\phi\right) - i\, Q_f(\psi+ i\phi) -Q_f\!\left(\psi-\phi\right) + i\, Q_f(\psi- i\phi)\Big)\,.
$$

The proof of the next lemma follows a well-known argument of Jordan and von Neumann in \cite{jordan1935inner}.

\begin{lemma}\label{sesqOverQ} If $f$ satisfies the parallelogram law, then $B_f$ is sesquilinear over $\mathbb{Q}[i]$.
\end{lemma}

\begin{proof} Let $R_f(\psi,\phi) = \frac{1}{4}\Big( Q_f(\psi+\phi) - Q_f(\psi-\phi)\Big)$, so $B_f(\psi,\phi)= R_f(\psi,\phi) - i\, R_f(\psi,i\phi)$. By Lemma \ref{lem:homog}, $Q_f(\lambda v)=Q_f(v)$ whenever $|\lambda|=1$, and we will use this below with $\lambda=-1$ and $\lambda=i$. Applying Lemma \ref{PLSum} twice, we have
\begin{align*}
R_f&(\psi_1+\psi_2,\phi) + R_f(\psi_1-\psi_2,\phi)\\
&= \frac{1}{4}\Big(Q_f(\psi_1+\psi_2+\phi) - Q_f(\psi_1+\psi_2-\phi)+ Q_f(\psi_1-\psi_2+\phi) -  Q_f(\psi_1-\psi_2-\phi)\Big)\\ 
&=  \frac{1}{4}\Big( 2\,Q_f(\psi_1+\phi) + 2\,Q_f(\psi_2) - 2\,Q_f(\psi_1-\phi) -  2\,Q_f(\psi_2) \Big) = 2\,R_f(\psi_1,\phi)\,.\numberthis \label{eq:sumFirst}
\end{align*}
We have $R_f(0,\phi)=\frac{1}{4}\big(Q_f(\phi)-Q_f(-\phi)\big)=0$. It follows that
\begin{align*}
2R_f(\psi,\phi)
\stackrel{\eqref{eq:sumFirst}}{=}
R_f(\psi+\psi,\phi) +R_f(\psi-\psi,\phi)
=R_f(\psi+\psi,\phi)+ R_f(0,\phi)
=
R_f(2\psi,\phi)\,,\numberthis \label{eq:mulTwo}
\end{align*}
and
\begin{align*}
R_f(\psi_1,\phi) + R_f(\psi_2,\phi)&= R_f\left(\frac{\psi_1+\psi_2}{2} + \frac{\psi_1-\psi_2}{2},\phi \right)+ R_f\left(\frac{\psi_1+\psi_2}{2} - \frac{\psi_1-\psi_2}{2},\phi \right)\\ 
&\stackrel{\eqref{eq:sumFirst}}{=}
2\,R_f\left(\frac{\psi_1+\psi_2}{2} ,\phi\right)\stackrel{\eqref{eq:mulTwo}}{=}R_f(\psi_1+\psi_2,\phi)\,.
\end{align*}
This shows that $R_f$ is additive in its first argument. By additivity, $R_f(k\psi,\phi)=k\,R_f(\psi,\phi)$ for every positive integer $k$ and every $\psi\in\mathcal{H}$. Hence, for positive integers $n$ and $m$,
$$
\frac{m}{n}\,R_f\Big(\frac{n}{m}\psi,\phi\Big)=\frac{1}{n}R_f(n\psi,\phi)=\frac{1}{n}\big(n\,R_f(\psi,\phi)\big)=R_f(\psi,\phi)\,,
$$
where the first equality applies the integer case with $k=m$ to $\frac{n}{m}\psi$. Together with $R_f(0,\phi)=0$, this shows that $R_f$ is homogeneous over the non-negative rationals in its first argument. By additivity, $R_f(\psi,\phi)+R_f(-\psi,\phi)=R_f(0,\phi)=0$, so $R_f(-q\psi,\phi)=-R_f(q\psi,\phi)$ for every non-negative rational $q$, and homogeneity of $R_f$ extends to all of $\mathbb{Q}$. Finally, $R_f$ is symmetric, since
\begin{align*}
R_f(\psi,\phi)= \frac{1}{4}\Big( Q_f(\psi+\phi) - Q_f(\psi-\phi)\Big)=\frac{1}{4}\Big( Q_f(\psi+\phi) - Q_f(\phi-\psi)\Big)= R_f(\phi,\psi)\,.\numberthis \label{eq:RSym}
\end{align*}

We now show the sesquilinearity of $B_f$ in the first argument. We have $R_f(i\psi,i\phi)=R_f(\psi,\phi)$ and
\begin{align*}
R_f(i\psi,\phi) = \frac{1}{4}\Big(Q_f(i\psi+\phi)-Q_f(i\psi-\phi)\Big)
=\frac{1}{4}\Big(Q_f(\psi-i\phi)-Q_f(\psi+i\phi)\Big)
=-R_f(\psi,i\phi)\,.
\end{align*}
Putting these identities together, we have
\begin{align*}
B_f(i\psi,\phi) &= R_f(i\psi,\phi)-iR_f(i\psi, i\phi)= -i\Big(R_f(\psi, \phi) - iR_f(\psi,i\phi)\Big)= -i B_f(\psi,\phi)\,.
\end{align*}
Furthermore, since $B_f(\psi,\phi) = R_f(\psi,\phi) - iR_f(\psi,i\phi)$, $B_f$ inherits additivity and homogeneity over $\mathbb{Q}$ in its first argument from $R_f$. We therefore have that for $r_1,r_2\in\mathbb{Q}$,
$$
B_f\big((r_1+ir_2)\psi,\phi\big)=r_1B_f(\psi,\phi)+r_2B_f(i\psi,\phi)=(r_1-ir_2)B_f(\psi,\phi)\,,
$$
so $B_f$ is conjugate linear over $\mathbb{Q}[i]$ in the first argument. Linearity over $\mathbb{Q}[i]$ for the second argument follows from the identity
\begin{align*}
B_f(\psi,\phi) &= R_f(\psi,\phi)-iR_f(\psi,i\phi)= R_f(\psi,\phi) + i R_f(i\psi,\phi)\\
&= R_f(\phi,\psi) + i R_f(\phi,i\psi) = B_f(\phi,\psi)^\ast\,,
\end{align*}
which shows that $B_f$ is conjugate-symmetric. By conjugate symmetry, $B_f$ is additive in its second argument, and
\begin{align*}
B_f(\psi,(r_1+ir_2)\phi)&=B_f((r_1+ir_2)\phi,\psi)^\ast
=\Big(r_1 B_f(\phi,\psi)-ir_2 B_f(\phi,\psi)\Big)^\ast \\
&=(r_1+ir_2) B_f(\psi,\phi)\,.\qedhere
\end{align*}
\end{proof}

The next ingredient we will need is the following special case of a theorem of Ostrowski \cite{ostrowski1929mathematische} on \emph{additive} functions, which are functions $g$ such that ${g(u+v)=g(u)+g(v)}$.

\begin{lemma}[Ostrowski's Theorem on Additive Functions]\label{ostrowski} Let $g:\mathbb{R}\to\mathbb{R}$ be an additive function that is bounded on some open interval $(a,b)$. Then $g$ is linear.
\end{lemma}

\begin{proof} Let $(a,b)$ be an interval on which $|g|\le M$. As in the proof of Lemma~\ref{sesqOverQ}, additivity (which gives $g(0)=0$) implies $g(qx)=qg(x)$ for all $q\in\mathbb{Q}$ and $x\in\mathbb{R}$, and in particular $g(q)=qg(1)$. Consider the function $r(x)=g(x)-xg(1)$. The function $r$ is additive and has absolute value bounded on $(a,b)$ by $M'=M + \max(|a|,|b|)\big|g(1)\big|$. Furthermore, $r$ has every rational number as a period, since for all $q\in \mathbb{Q}$,
$$
r(x+q) = r(x)+ r(q)=r(x) + g(q) - qg(1) = r(x) + qg(1) - qg(1)=r(x)\,.
$$
Suppose by way of contradiction that $r(x_0)\ne 0$ for some $x_0$. Then there exists $n>0$ such that $|nr(x_0)|>M'$. Since the set of translates by rationals of any real number is dense in $\mathbb{R}$, there exists $q\in \mathbb{Q}$ such that $nx_0+q\in (a,b)$. It follows that $|r(nx_0+q)|=|r(nx_0)|=|nr(x_0)|>M'$, which is a contradiction. Thus $r$ is the $0$ function and $g(x)=g(1)\,x$.
\end{proof}

\noindent As the following lemma shows, Ostrowski's theorem allows us to go from sesquilinearity over $\mathbb{Q}[i]$ to sesquilinearity over $\mathbb{C}$ (cf.~\cite[Sec.~3]{caves2004gleason} and \cite{wright2019gleason}, where closely related additive-to-linear arguments are used to prove Gleason-type theorems).

\begin{lemma}\label{sesqOverC} If $f$ is bounded and satisfies the parallelogram law, then $B_f$ is sesquilinear over $\mathbb{C}$.
\end{lemma}

\begin{proof} Consider the function $F_{\psi,\phi}:\mathbb{R}\to \mathbb{R}$ given by $F_{\psi,\phi}(\lambda) = \mathrm{Re}\big(B_f(\psi,\lambda\phi)\big)$. By the additivity of $B_f$ in its second argument (Lemma \ref{sesqOverQ}), applied to the vectors $\lambda_1\phi$ and $\lambda_2\phi$, the function $F_{\psi,\phi}$ is additive. Since $f$ is bounded, $|Q_f(v)| \le C\|v\|^2$ for all $v$, where $C=\sup|f|$. Since the arguments passed to $Q_f$ in the expansion of $B_f(\psi, \lambda\phi)$ have norms bounded by $\|\psi\| + \|\phi\|$ for $\lambda \in (0,1)$, it follows from the triangle inequality that $F_{\psi,\phi}$ restricted to this interval is bounded in absolute value by $C(\|\psi\|+\|\phi\|)^2$. By Lemma \ref{ostrowski}, $F_{\psi,\phi}(\lambda)=\lambda F_{\psi,\phi}(1)$ for all $\lambda\in\mathbb{R}$, and this holds for every pair $\psi,\phi$. Since $B_f(\psi,i\lambda\phi)=iB_f(\psi,\lambda\phi)$ by Lemma \ref{sesqOverQ}, we have $\mathrm{Im}\big(B_f(\psi,\lambda\phi)\big)=-F_{\psi,i\phi}(\lambda)$, and hence for all $\lambda\in\mathbb{R}$,
$$
B_f(\psi,\lambda\phi) = F_{\psi,\phi}(\lambda) - iF_{\psi,i\phi}(\lambda)=\lambda\big(F_{\psi,\phi}(1) - iF_{\psi,i\phi}(1)\big)=\lambda B_f(\psi,\phi)\,.
$$
Using homogeneity over $\mathbb{R}$ and sesquilinearity over $\mathbb{Q}[i]$, it follows that for all $a,b\in \mathbb{R}$
\begin{align*}
B_f(\psi,(a+ib)\phi)&=B_f(\psi,a\phi)+ B_f(\psi, ib\,\phi)
=aB_f(\psi,\phi)+ bB_f(\psi,i\,\phi)
=(a+ib) B_f(\psi,\phi)\,,
\end{align*}
so $B_f$ is homogeneous over $\mathbb{C}$ in the second argument. For the first argument, we have that for all $a,b\in\mathbb{R}$,
\[
B_f((a+ib)\psi,\phi) = B_f(\phi,(a+ib)\psi)^\ast = \big((a+ib)B_f(\phi,\psi)\big)^\ast = (a-ib)B_f(\psi,\phi)\,.\qedhere
\]
\end{proof}

We are now ready to prove the main result of this section.

\begin{lemma}\label{PLtoQF} Postulate \ref{PL} implies Postulate \ref{QF}.
\end{lemma}

\begin{proof} Let $\mathcal{E}$ be an available HPMM, let $O\in \Sigma_{\mathcal{E}}$, and let $f:S(\mathcal{H}_\mathcal{E})\to \mathbb{R}$ be defined by $f(\psi)=\mathbb{P}_{\mathcal{E}}(O\,;\psi)$. By Postulate \ref{PL}, $f$ satisfies the parallelogram law. Since $f(\psi)$ is a probability, it is bounded above by $1$ and below by $0$. It follows from Lemma \ref{sesqOverC} that $B_f$ is a sesquilinear form over $\mathbb{C}$. Letting $Q_O$ be defined by $Q_O(\psi)=B_f(\psi,\psi)$, we have that for all $\psi\in \mathcal{H}_{\mathcal{E}}$,
$$
Q_O(\psi)=B_f(\psi,\psi)=\tfrac{1}{4}\Big(Q_f(2\psi)-Q_f(0)+i\big[Q_f\big((1-i)\psi\big)-Q_f\big((1+i)\psi\big)\big]\Big)=Q_f(\psi)\,,
$$
since $Q_f(2\psi)=4Q_f(\psi)$, $Q_f(0)=0$, and $Q_f\big((1\pm i)\psi\big)=2Q_f(\psi)$ by Lemma~\ref{lem:homog}. Hence $Q_O=Q_f$ is real-valued on all of $\mathcal{H}_{\mathcal{E}}$, and for $\psi\in S(\mathcal{H}_\mathcal{E})$, $Q_O(\psi)=f(\psi)=\mathbb{P}_{\mathcal{E}}(O\,;\psi)$. 
Since $B_f$ is sesquilinear, $Q_O$ is a quadratic form, so Property \ref{QF} holds.
\end{proof}

\subsection{Postulate \ref*{POVM} is equivalent to Postulate \ref*{QF}}\label{sec:POVM}

In this subsection, we show that POVM universality is, under the standard Kolmogorov axioms for probability, mathematically equivalent to Postulate \ref{QF}. The proof is a variant of the argument given in \cite[Sec.~4.4]{durr2004quantum} for the correspondence between POVMs and normalized measure-valued quadratic maps, with the details adapted to the present framework.

\begin{lemma}\label{lem:POVMeqQF} Postulate \ref{POVM} is equivalent to Postulate \ref{QF}.
\end{lemma}

\begin{proof} The claim that POVM universality implies Postulate \ref{QF} is almost immediate. Indeed, let $\mathcal{E}$ be an available HPMM. Under POVM universality, $\mathcal{E}$ is associated with a POVM $E_{\mathcal{E}}$ such that $\mathbb{P}_{\mathcal{E}}(O\,;\psi)=\langle \psi, E_{\mathcal{E}}(O)\psi\rangle$ for all $O\in\Sigma_{\mathcal{E}}$ and $\psi\in S(\mathcal{H}_{\mathcal{E}})$. Since each $E_{\mathcal{E}}(O)$ is positive, $Q_O(\psi)=\langle\psi,E_{\mathcal{E}}(O)\psi\rangle$ is a real-valued quadratic form on $\mathcal{H}_{\mathcal{E}}$, namely the one associated with the bounded sesquilinear form $(v,w)\mapsto \langle v,E_{\mathcal{E}}(O)w\rangle$.

Conversely, let $\mathcal{E}$ be an available HPMM. By Postulate \ref{QF}, for any $O\in\Sigma_{\mathcal{E}}$, there exists a quadratic form $Q_O(v)=q_O(v,v)$ associated with a sesquilinear form $q_O$ on $\mathcal{H}_{\mathcal{E}}$ such that $\mathbb{P}_{\mathcal{E}}(O\,;\psi)=Q_O(\psi)$ for all $\psi \in S(\mathcal{H}_{\mathcal{E}})$. Since probabilities take values in $[0,1]$, $|Q_O(\psi)| \le 1$ for $\psi\in S(\mathcal{H}_{\mathcal{E}})$. By the homogeneity of $Q_O$, $|Q_O(v)| \le \|v\|^2$ for all $v \in \mathcal{H}_{\mathcal{E}}$, and by polarization, $q_O$ is therefore a bounded sesquilinear form. By the Riesz representation theorem for bounded sesquilinear forms \cite[Thm.~II.2.2]{conway}, there exists a bounded operator $E_{\mathcal{E}}(O)$ such that $Q_O(\psi)=\langle \psi,E_{\mathcal{E}}(O)\psi\rangle$.

By Kolmogorov's first axiom, probabilities must be non-negative. It follows that for all $\psi\in S(\mathcal{H}_\mathcal{E})$, $\langle \psi,E_{\mathcal{E}}(O)\psi\rangle =\mathbb{P}_\mathcal{E}(O\,;\psi)\ge 0$. By the homogeneity of $\psi\mapsto\langle \psi,E_{\mathcal{E}}(O)\psi\rangle$, this holds for all $\psi\in \mathcal{H}_{\mathcal{E}}$, so $E_{\mathcal{E}}(O)$ is a positive operator. By Kolmogorov's second axiom, for any $\psi$, $\mathbb{P}_{\mathcal{E}}(\Omega_\mathcal{E}\,;\psi)=1$. It follows that for any $\psi\in S(\mathcal{H}_\mathcal{E})$, $\langle \psi,E_{\mathcal{E}}(\Omega_{\mathcal{E}})\psi\rangle = 1=\langle \psi,I\psi\rangle$. This shows that $E_{\mathcal{E}}(\Omega_\mathcal{E})$ agrees with the identity operator $I$ on $S(\mathcal{H}_{\mathcal{E}})$, so $E_{\mathcal{E}}(\Omega_{\mathcal{E}})=I$ by polarization, and $E_{\mathcal{E}}$ is therefore normalized. Furthermore, for any $\psi \in S(\mathcal{H}_\mathcal{E})$, 
$$
\langle \psi,\big(I-E_{\mathcal{E}}(O)\big)\psi\rangle = \langle \psi ,\psi\rangle-\langle \psi,E_{\mathcal{E}}(O)\psi\rangle\stackrel{(\mathrm{\ref{QF}})}{=} 1- \mathbb{P}_{\mathcal{E}}(O\,;\psi)\ge 0\,,
$$
so $I-E_{\mathcal{E}}(O)$ is a positive operator. 

The Kolmogorov axioms immediately give $\mathbb{P}_\mathcal{E}(\varnothing\,;\psi)=0$, so $\langle \psi ,E_{\mathcal{E}}(\varnothing)\psi\rangle = 0$ for all $\psi\in S(\mathcal{H}_\mathcal{E})$, from which it follows that $E_{\mathcal{E}}(\varnothing)$ is the zero operator.

Let $\{O_i~|~i\in\mathbb{N}\}$ be a countable collection of elements of $\Sigma_{\mathcal{E}}$ such that $O_i\cap O_j=\varnothing$ for $i\ne j$, and note that $\bigcup_{i=1}^\infty O_i \in \Sigma_{\mathcal{E}}$ since $\Sigma_{\mathcal{E}}$ is a $\sigma$-algebra. Letting $\mathrm{K3}$ denote Kolmogorov's third axiom, we have that for all $\psi\in S(\mathcal{H}_\mathcal{E})$ and $M\le N$,
\begin{equation}\label{eq:Kol1}
\left\langle \psi, \sum_{i=M}^N E_{\mathcal{E}}(O_i)\psi\right\rangle
=
 \sum_{i=M}^N\langle \psi,E_{\mathcal{E}}(O_i)\psi\rangle
\stackrel{(\mathrm{\ref{QF}})}{=}
\sum_{i=M}^N\mathbb{P}_{\mathcal{E}}(O_i\,;\psi)
\stackrel{(\mathrm{K3})}{=}
\mathbb{P}_{\mathcal{E}}\big( {\textstyle\bigcup_{i=M}^N} O_i\,;\psi\big)\,.
\end{equation}
Let $S_0=0$ and let $S_N= \sum_{i=1}^N E_{\mathcal{E}}(O_i)$ for $N>0$. By Equation \eqref{eq:Kol1}, $\langle\psi,(S_N-S_{M-1})\psi\rangle=\mathbb{P}_{\mathcal{E}}\big(\bigcup_{i=M}^N O_i\,;\psi\big)\in[0,1]$ for all $\psi\in S(\mathcal{H}_{\mathcal{E}})$ and $M\le N$, so the sequence $\{S_N\}$ is positive, monotone, and uniformly bounded. It follows from Vigier's theorem (the monotone convergence theorem for bounded operators \cite[Thm.~4.1.1]{murphy2014c}) that $S_N$ converges strongly to an operator $S_\infty$. By the definition of strong convergence and the continuity of the inner product, for all $\psi\in S(\mathcal{H}_{\mathcal{E}})$,
\begin{align}\label{eq:QFImp1}
\left\langle \psi,S_\infty\psi\right\rangle
&~=
\left\langle \psi,\lim_{N\to\infty}\left(\sum_{i=1}^N E_{\mathcal{E}}(O_i)\psi\right)\right\rangle = 
\lim_{N\to \infty}\sum_{i=1}^N\langle \psi,E_{\mathcal{E}}(O_i)\psi\rangle\nonumber\\
&\stackrel{(\mathrm{\ref{QF}})}{=}
\lim_{N\to \infty}
\sum_{i=1}^N\mathbb{P}_{\mathcal{E}}(O_i\,;\psi)\stackrel{(\mathrm{K3})}{=}
\mathbb{P}_{\mathcal{E}}\big({\textstyle\bigcup_{i=1}^\infty} O_i\,;\psi\big) \stackrel{(\mathrm{\ref{QF}})}{=}
\langle \psi,E_{\mathcal{E}}\left({\textstyle\bigcup_{i=1}^\infty} O_i\right)\psi\rangle\,.
\end{align}
Since Equation \eqref{eq:QFImp1} holds for any $\psi\in S(\mathcal{H}_{\mathcal{E}})$, $S_\infty = E_{\mathcal{E}}\left(\bigcup_{i=1}^\infty O_i\right)$ by polarization, so $\sum_{i=1}^N E_{\mathcal{E}}(O_i)$ converges strongly to $E_{\mathcal{E}}\left(\bigcup_{i=1}^\infty O_i\right)$. Since strong convergence implies weak convergence, this shows that $E_{\mathcal{E}}$ satisfies the countable additivity property for a POVM.
\end{proof}

\subsection{Postulate \ref*{QF} implies Postulate \ref*{WSEP}}\label{sec:QFtoWSEP}

As can be seen from Figure \ref{Implications}, the result in this subsection is not strictly necessary, since it follows easily from the implication \ref{QF}$\Rightarrow$\ref{SEP} and the fact that Postulate \ref{WSEP} is a special case of Postulate \ref{SEP}. However, the following short argument avoids Bochner integration, and so offers an easier route to the conclusion for those primarily interested in the countably supported case.

\begin{lemma}\label{densityReduction} Postulate \ref{QF} implies Postulate \ref{WSEP}. 
\end{lemma}

\begin{proof} Let $\mathcal{E}$ be an available HPMM and $O\in \Sigma_{\mathcal{E}}$. We want to construct $f_O:\mathcal{D}(\mathcal{H}_{\mathcal{E}})\to [0,1]$ such that for all $\mu\in\mathrm{M}_1\big(S(\mathcal{H}_\mathcal{E})\big)$ with countable support, $f_O\circ\mathcal{M}(\mu) =\widetilde{\mathbb{P}}_\mathcal{E}(O\,;\mu)$.

By Lemma \ref{lem:POVMeqQF}, Postulate \ref{QF} implies that $\mathcal{E}$ is associated with a POVM $E_{\mathcal{E}}$ such that $\mathbb{P}_{\mathcal{E}}(O\,;\psi) = \langle \psi , E_{\mathcal{E}}(O) \psi \rangle$ for all $\psi \in S(\mathcal{H}_{\mathcal{E}})$. Let $\mu=\{(p_i,\psi_i)\}$ be a countably supported element of $\mathrm{M}_1\big(S(\mathcal{H}_\mathcal{E})\big)$. By the definition of $\widetilde{\mathbb{P}}_{\mathcal{E}}$,  
$$
\widetilde{\mathbb{P}}_{\mathcal{E}}(O\,;\mu)=\sum_i p_i\,\mathbb{P}_{\mathcal{E}}(O\,;\psi_i)
=\sum_i p_i\langle \psi_i,E_{\mathcal{E}}(O)\psi_i\rangle\,.
$$
For each $i$, $\langle \psi_i,E_{\mathcal{E}}(O)\psi_i\rangle=\mathrm{Tr}\big(E_{\mathcal{E}}(O)\ket{\psi_i}\!\bra{\psi_i}\big)$. Since $\sum_i p_i=1$ and each $\ket{\psi_i}\!\bra{\psi_i}$ has trace norm $1$, the series $\sum_i p_i\ket{\psi_i}\!\bra{\psi_i}$ converges in trace norm. Since $E_{\mathcal{E}}(O)$ is bounded, the trace pairing $\rho\mapsto \mathrm{Tr}(E_{\mathcal{E}}(O)\rho)$ is continuous in the trace norm. It follows that we may pass the countable sum through the trace, so for all such $\mu$,
$$
\widetilde{\mathbb{P}}_{\mathcal{E}}(O\,;\mu)=\sum_ip_i \mathrm{Tr}\big(E_{\mathcal{E}}(O)\ket{\psi_i}\!\bra{\psi_i}\big)
=\mathrm{Tr}\Big(E_{\mathcal{E}}(O)\!\sum_i p_i \ket{\psi_i}\!\bra{\psi_i}\Big)\,.
$$
Let $f_O(\rho) \equiv \mathrm{Tr}(E_{\mathcal{E}}(O)\rho)$ for $\rho\in\mathcal{D}(\mathcal{H}_{\mathcal{E}})$. The display above shows that $f_O\circ \mathcal{M}(\mu)=\widetilde{\mathbb{P}}_{\mathcal{E}}(O\,;\mu)$ for every countably supported $\mu$. Since every $\rho\in\mathcal{D}(\mathcal{H}_{\mathcal{E}})$ is of the form $\mathcal{M}(\mu)$ for such a $\mu$ (see Section \ref{subsec:propmix}), it follows that $f_O$ takes values in $[0,1]$, as required.
\end{proof}

\subsection{Postulate \ref*{QF} implies Postulate \ref*{SEP}}\label{sec:QFtoSEP}

We will now prove that Postulate \ref{QF} implies the impossibility of DEIDO for general proper mixtures. The argument is similar to the one in Section \ref{sec:QFtoWSEP}, but requires more care with measurability and convergence. We will need the following preliminary lemma.

\begin{lemma}\label{tcintegral} Given a bounded operator $A$ on a separable Hilbert space $\mathcal{H}$ and a Borel probability measure $\mu$ on $S(\mathcal{H})$, $\int_{S(\mathcal{H})} A\ket{\psi}\!\bra{\psi}\,d\mu(\psi)$ defines a trace-class operator. When $A=I$, this operator lies in $\mathcal{D}(\mathcal{H})$.
\end{lemma}

\begin{proof}
The map $\psi\mapsto \ket{\psi}\!\bra{\psi}$ is continuous from $S(\mathcal{H})$ equipped with the norm topology into the trace class $\mathrm{B}_1(\mathcal{H})$ equipped with the topology induced by the trace-norm $\|\cdot\|_1$. Since $A$ is bounded, $T\mapsto AT$ is continuous on the trace class \cite[Sec.~IX.2, Exercise~20(g)]{conway}. It follows that $\psi\mapsto A\ket{\psi}\!\bra{\psi}$ is a composition of continuous maps, so it is continuous, Borel measurable, and weakly measurable. Since we have assumed that $\mathcal{H}$ is separable, the trace class $\mathrm{B}_1(\mathcal{H})$ is also
separable, so this map is separably valued. It follows that the map $\psi\mapsto A\ket{\psi}\!\bra{\psi}$ is strongly measurable by the Pettis measurability theorem \cite[Thm.~1.1.6]{hytonen2016analysis}.

To show that the map is Bochner integrable, it suffices by \cite[Prop.~1.2.2]{hytonen2016analysis} to show $\int_{S(\mathcal{H})} \big\|A\ket{\psi}\!\bra{\psi}\big\|_1\,d\mu(\psi)$ is finite. Since $A\ket{\psi}\!\bra{\psi}=\ket{A\psi}\!\bra{\psi}$ is rank one and $\|\psi\|=1$, $\big\|A\ket{\psi}\!\bra{\psi}\big\|_1=\|A\psi\|\,\|\psi\|\le\|A\|$. Because $\mu$ is a probability measure on $S(\mathcal{H})$, we find that $\int_{S(\mathcal{H})} \big\|A\ket{\psi}\!\bra{\psi}\big\|_1\,d\mu(\psi)$ is bounded by  $\int_{S(\mathcal{H})}\|A\|\,d\mu(\psi) = \|A\| < \infty$. Since the Bochner integral is defined as a trace-norm limit of finite sums of elements of $\mathrm{B}_1(\mathcal{H})$, and $\mathrm{B}_1(\mathcal{H})$ is complete in the trace norm, the integral is itself an element of $\mathrm{B}_1(\mathcal{H})$, i.e., a trace-class operator.

When $A=I$, since the trace \cite[Prop.~D.2.5]{hytonen2016analysis} and the maps $T\mapsto\langle\phi,T\phi\rangle$ are bounded linear functionals on $\mathrm{B}_1(\mathcal{H})$, they commute with the Bochner integral \cite[Eq.~(1.2)]{hytonen2016analysis}, so $\mathrm{Tr}\int\ket{\psi}\!\bra{\psi}\,d\mu(\psi)=\int 1\,d\mu(\psi)=1$ and for all $\phi\in \mathcal{H}$, 
$$
\left\langle\phi,\int\ket{\psi}\!\bra{\psi}\,d\mu(\psi)\,\phi\right\rangle=\int|\langle\phi,\psi\rangle|^2\,d\mu(\psi)\ge 0\,,
$$
so the operator $\int\ket{\psi}\!\bra{\psi}\,d\mu(\psi)$ lies in $\mathcal{D}(\mathcal{H})$.
\end{proof}

\noindent We are now ready to prove the general form of Lemma \ref{densityReduction}.

\begin{lemma}\label{densityReductionGeneral} Postulate \ref{QF} implies Postulate \ref{SEP}. 
\end{lemma}

\begin{proof} Let $\mathcal{E}$ be an available HPMM and $O\in\Sigma_{\mathcal{E}}$. By Lemma \ref{lem:POVMeqQF}, Postulate \ref{QF} implies that $\mathcal{E}$ is associated with a POVM $E_{\mathcal{E}}$, so that $0\le E_{\mathcal{E}}(O)\le I$ and $\mathbb{P}_{\mathcal{E}}(O\,;\psi) =\langle\psi,E_{\mathcal{E}}(O)\psi\rangle$ for all $\psi \in S(\mathcal{H}_{\mathcal{E}})$. By the trace identity $\langle\psi,E_{\mathcal{E}}(O)\psi\rangle=\mathrm{Tr}\big(E_{\mathcal{E}}(O)\ket{\psi}\!\bra{\psi}\big)$, it follows that given $\mu\in \mathrm{M}_1(S(\mathcal{H}_{\mathcal{E}}))$, we have
$$
\widetilde{\mathbb{P}}_{\mathcal{E}}(O\,;\mu)
=
\int_{S(\mathcal{H}_{\mathcal{E}})}\langle\psi,E_{\mathcal{E}}(O)\psi\rangle\, d\mu(\psi)
=
\int_{S(\mathcal{H}_{\mathcal{E}})}\mathrm{Tr}\big(E_{\mathcal{E}}(O)\ket{\psi}\!\bra{\psi}\big)\,d\mu(\psi)\,.
$$
Because the scalar integrand $\mathrm{Tr}\big(E_{\mathcal{E}}(O) \ket{\psi}\!\bra{\psi}\big) = \mathbb{P}_{\mathcal{E}}(O\,;\psi)$ is continuous (hence Borel measurable) and takes values in $[0,1]$, it is bounded and therefore $\mu$-integrable. In particular, $\widetilde{\mathbb{P}}_{\mathcal{E}}(O\,;-)$ is defined on all of $\mathrm{M}_1(S(\mathcal{H}_{\mathcal{E}}))$, as Property \ref{SEP} requires. Furthermore, by Lemma \ref{tcintegral}, the relevant operator-valued Bochner integrals exist and take values in the trace class. Since $\mathrm{Tr}$ is a bounded linear functional on the trace class \cite[Prop.~D.2.5]{hytonen2016analysis}, and since left multiplication by the bounded operator $E_{\mathcal E}(O)$ defines a bounded linear map on the trace class, both commute with the Bochner integral \cite[Eq.~(1.2)]{hytonen2016analysis}. By Lemma \ref{tcintegral} with \(A=I\), $\mathcal M(\mu)
=
\int_{S(\mathcal H_{\mathcal{E}})}
\ket{\psi}\!\bra{\psi}\,d\mu(\psi)$ lies in $\mathcal D(\mathcal H_{\mathcal{E}})$. Hence
\begin{align*}
\widetilde{\mathbb P}_{\mathcal E}(O\,;\mu)
&=
\int_{S(\mathcal H_{\mathcal{E}})}
\mathrm{Tr} \left(E_{\mathcal E}(O)\ket{\psi}\!\bra{\psi}\right)
d\mu(\psi)
=
\mathrm{Tr} \bigg(
\int_{S(\mathcal H_{\mathcal{E}})}
E_{\mathcal E}(O)\ket{\psi}\!\bra{\psi}\,d\mu(\psi)
\bigg)\\
&=
\mathrm{Tr} \bigg(
E_{\mathcal E}(O)
\int_{S(\mathcal H_{\mathcal{E}})}
\ket{\psi}\!\bra{\psi}\,d\mu(\psi)
\bigg)
=
\mathrm{Tr} \big(E_{\mathcal E}(O)\mathcal M(\mu)\big).
\end{align*}
Let $f_O$ be defined by $f_O(\rho)\equiv \mathrm{Tr}(E_{\mathcal{E}}(O)\rho)$. Since $0\le E_{\mathcal{E}}(O)\le I$, for every $\rho\in\mathcal{D}(\mathcal{H}_{\mathcal{E}})$ we have $0\le f_O(\rho)\le \mathrm{Tr}(\rho)=1$, so $f_O$ takes values in $[0,1]$. The preceding calculation therefore shows that for all $\mu$ in $\mathrm{M}_1(S(\mathcal{H}_{\mathcal{E}}))$, $\widetilde{\mathbb{P}}_{\mathcal{E}}(O\,;\mu) = f_O\circ \mathcal{M}(\mu)$, as required.
\end{proof}

\subsection{Postulate \ref*{WSEP} implies Postulate \ref*{PL}}\label{sec:WSEPtoPL}

In this subsection, we show that the impossibility of DEIDO directly implies PL. This conclusion is closely related to a standard result in operational treatments of quantum theory, where states are identified with density operators from the outset and measurements are defined as affine maps from states to probability distributions. In that setting, every such map is given by a POVM (see, e.g., \cite[Thm.~2.2.1]{holevo2011probabilistic}).

The argument below does not assume affinity, and uses Postulate \ref{WSEP} only to compare pairs of two-element ensembles with the same density operator, whose unnormalized vectors are related by a $2\times 2$ unitary matrix (cf.\ \cite[Sec.~2]{hughston1993complete}). The implication \ref{WSEP}$\,\Rightarrow\,$\ref{PL} will also be useful because, combined with Lemma \ref{PLtoQF}, it allows us to use the algebraic consequences of Postulate \ref{QF} when proving the implication \ref{WSEP}$\,\Rightarrow\,$\ref{PT}.

\begin{lemma}\label{DIPL} Postulate \ref{WSEP} implies Postulate \ref{PL}.
\end{lemma}

\begin{proof} Let $\mathcal{E}$ be an available HPMM, let $U=  \left(\begin{smallmatrix}
a & b \\
c & d
\end{smallmatrix}\right)$ be unitary, and let $\psi$ and $\phi$ be elements of $S(\mathcal{H}_{\mathcal{E}})$ such that $a\psi+b\phi$ and $c\psi+d\phi$ are non-zero. Let $\mu_1$ be the proper mixture given by $\mu_1=\left\{\left(\frac{1}{2},\psi\right),\left(\frac{1}{2},\phi\right)\right\}$ and let 
$$
\mu_2=\left\{\left(\frac{\|a\psi+b\phi\|^2}{2},\frac{a\psi+b\phi}{\|a\psi+b\phi\|}\right),\left(\frac{\|c\psi+d\phi\|^2}{2},\frac{c\psi+d\phi}{\|c\psi+d\phi\|}\right)\right\}\,.
$$
The weights in $\mu_2$ sum to $1$, since unitarity of $U$ gives
$$
\|a\psi+b\phi\|^2+\|c\psi+d\phi\|^2
=
\|\psi\|^2+\|\phi\|^2
=
2.
$$
Expanding the mixture map, we find that $\mathcal{M}(\mu_2)$ equals
\begin{align*}
&\frac{\|a\psi+b\phi\|^2}{2} \left|\frac{a\psi+b\phi}{\|a\psi+b\phi\|}\right\rangle\!\left\langle \frac{a\psi+b\phi}{\|a\psi+b\phi\|}\right|+ \frac{\|c\psi+d\phi\|^2}{2} \left|\frac{c\psi+d\phi}{\|c\psi+d\phi\|} \right\rangle\!\left\langle \frac{c\psi+d\phi}{\|c\psi+d\phi\|}\right|\\
&=\frac{1}{2}\left(|a\psi+b\phi\rangle\langle a\psi+b\phi|+|c\psi+d\phi\rangle\langle c\psi+d\phi|\right)\\
&=\frac{1}{2}\big(\left(|a|^2+|c|^2\right)\ket{\psi}\!\bra{\psi}+\left(|b|^2+|d|^2\right)\ket{\phi}\!\bra{\phi} +(a b^\ast + c d^\ast)\ket{\psi}\!\bra{\phi} + (b a^\ast + d c^\ast)\ket{\phi}\!\bra{\psi}\big)\\
&= \tfrac{1}{2}(\ket{\psi}\!\bra{\psi}+ \ket{\phi}\!\bra{\phi}) 
=\mathcal{M}(\mu_1)\,,\numberthis \label{eq:mixtureEq}
\end{align*}
where the cross-terms vanish because $a b^\ast + c d^\ast = \big(a^\ast b + c^\ast d\big)^\ast = 0$ by the orthonormality of the columns of $U$. This shows that $\mu_1$ and $\mu_2$ are associated with the same density operator.

By Postulate \ref{WSEP}, for any $O\in \Sigma_{\mathcal{E}}$, there exists ${f_O: \mathcal{D}(\mathcal{H}_{\mathcal{E}}) \to [0,1]}$ (not depending on $\mu_i$) such that
\begin{align*}
 &\frac{\|a\psi+b\phi\|^2}{2}\mathbb{P}_{\mathcal{E}}\!\left(O\,;\frac{a\psi+b\phi}{\|a\psi+b\phi\|}\right)+  \frac{\|c\psi+d\phi\|^2}{2}\mathbb{P}_{\mathcal{E}}\!\left(O\,;\frac{c\psi+d\phi}{\|c\psi+d\phi\|}\right)=\widetilde{\mathbb{P}}_{\mathcal{E}}(O\,;\mu_2)\\
 &\qquad\stackrel{(\mathrm{\ref{WSEP}})}{=}
 f_O(\mathcal{M}(\mu_2))
 \stackrel{\eqref{eq:mixtureEq}}{=} 
 f_O(\mathcal{M}(\mu_1))
  \stackrel{(\mathrm{\ref{WSEP}})}{=}
 \widetilde{\mathbb{P}}_{\mathcal{E}}(O\,;\mu_1)= \frac{1}{2}\mathbb{P}_{\mathcal{E}}\left(O\,;\psi\right)+\frac{1}{2}\mathbb{P}_{\mathcal{E}}\left(O\,;\phi\right)\,.
\end{align*}
Multiplying both sides by $2$ gives the identity required by Postulate \ref{PL}.
\end{proof}

\subsection{Postulates \ref*{IP} and \ref*{WSEP} imply Postulate \ref*{PT}}\label{sec:WSEPtoPT}

We showed in Section \ref{sec:WSEPtoPL} that Postulate \ref{WSEP} implies Postulate \ref{PL}, in Section \ref{sec:PLtoQF} that Postulate \ref{PL} implies Postulate \ref{QF}, and in Lemma \ref{densityReduction} that Postulate \ref{QF} implies Postulate \ref{WSEP}. These three postulates are therefore logically equivalent, so we can use Postulate \ref{QF} in this subsection. We will need a preliminary lemma about quadratic forms defined on tensor products of Hilbert spaces.

\begin{lemma}\label{simpleTensors} Let $Q_1$ and $Q_2$ be bounded quadratic forms on $\mathcal{H}_A\otimes \mathcal{H}_B$ such that $Q_1(\psi\otimes\phi)=Q_2(\psi\otimes\phi)$ for all $\psi\in S(\mathcal{H}_A)$ and $\phi\in S(\mathcal{H}_B)$. Then $Q_1=Q_2$.
\end{lemma}

\begin{proof} Let $B_1$ and $B_2$ be the bounded sesquilinear forms associated with $Q_1$ and $Q_2$, and let $B=B_1-B_2$. By homogeneity, $B(x\otimes y,x\otimes y)=0$ for all $x\in\mathcal{H}_A$ and $y\in\mathcal{H}_B$. Recall that, by the complex polarization identity $s(x,x')=\frac{1}{4}\sum_{m=1}^4 i^{-m}s(x+i^mx',x+i^mx')$, a sesquilinear form $s$ that vanishes on the diagonal vanishes identically. First, fix $y\in\mathcal{H}_B$. The form $(x,x')\mapsto B(x\otimes y,x'\otimes y)$ is sesquilinear on $\mathcal{H}_A$ and vanishes on the diagonal, so $B(x\otimes y,x'\otimes y)=0$ for all $x,x'\in\mathcal{H}_A$. Next, fix $x,x'\in\mathcal{H}_A$. The form $(y,y')\mapsto B(x\otimes y,x'\otimes y')$ is sesquilinear on $\mathcal{H}_B$ and, by the previous step, vanishes on the diagonal, so $B(x\otimes y,x'\otimes y')=0$ for all $y,y'\in\mathcal{H}_B$. By sesquilinearity, $B$ vanishes on all pairs of finite sums of product vectors. Since such sums are dense in $\mathcal{H}_A\otimes\mathcal{H}_B$ and $B$ is bounded, and hence continuous, $B=0$. It follows that $Q_1=Q_2$, as required.
\end{proof}

The proof of the following lemma shows that, under Postulates \ref{WSEP} and \ref{IP}, Bob's one-sided outcomes on any post-preparation independent composite are represented by operators of the form $I_A\otimes E$, where $E$ is an operator on Bob's Hilbert space that is independent of Alice's experiment. The partial trace description then follows (cf.\ \cite[Box~2.6]{nielsen2010quantum}, where the representation $M\otimes I$ for measurements on one subsystem is inferred from their behavior on product states, and the partial trace is shown to be the unique map reproducing the resulting statistics).

\begin{lemma}\label{lem:WSEPandIPtoPT} Postulate \ref{WSEP} and Postulate \ref{IP} imply Postulate \ref{PT}.
\end{lemma}

\begin{proof}
Let $\{\mathcal{J}_\alpha~|~\alpha\in\mathcal{I}\}$ be a
preparation-compatible family of available composite HPMMs,
each describing a post-preparation independent composite
experiment, with common left-hand preparation space
$\mathcal{H}_A$ and fixed right-hand component $\mathcal{E}_B$. If $\mathcal{I}$ is empty, Property \ref{PT} holds by taking
each $f_{O_B}$ to be the zero function. Assume $\mathcal{I}\ne\varnothing$,
and fix $O_B\in\Sigma_{\mathcal{E}_B}$.

By Lemmas \ref{DIPL}, \ref{PLtoQF}, and \ref{lem:POVMeqQF}, Postulate \ref{WSEP} implies that each available HPMM is associated with a POVM. In particular, for each $\alpha\in\mathcal{I}$, the operator $M^\alpha_{O_B}\equiv E_{\mathcal{J}_\alpha}(\Omega_{\mathcal{E}_A^\alpha}\times O_B)$ on $\mathcal{H}_A\otimes\mathcal{H}_{\mathcal{E}_B}$ satisfies $0\le M^\alpha_{O_B}\le I$ and $\mathbb{P}_{\mathcal{J}_\alpha}\big(\Omega_{\mathcal{E}_A^\alpha}\times O_B\,;\Psi\big)=\langle\Psi,M^\alpha_{O_B}\Psi\rangle$ for every $\Psi\in S(\mathcal{H}_A\otimes\mathcal{H}_{\mathcal{E}_B})$. Choose $\alpha_0\in\mathcal{I}$ and
$\psi_0\in S(\mathcal{H}_A)$, and let
$V:\mathcal{H}_{\mathcal{E}_B}\to\mathcal{H}_A\otimes\mathcal{H}_{\mathcal{E}_B}$
be the isometry $V(\phi)=\psi_0\otimes\phi$. Define $E_{O_B}=V^\dagger M^{\alpha_0}_{O_B}V$. Then $0\le E_{O_B}\le I_B$, where $I_B$ is the identity on $\mathcal{H}_{\mathcal{E}_B}$, and Postulate \ref{IP} gives $\langle\phi,E_{O_B}\phi\rangle = \mathbb{P}_{\mathcal{J}_{\alpha_0}}
\big(\Omega_{\mathcal{E}_A^{\alpha_0}}\times O_B
\,;\psi_0\otimes\phi\big) = \mathbb{P}_{\mathcal{E}_B}(O_B\,;\phi)$ for every $\phi\in S(\mathcal{H}_{\mathcal{E}_B})$. Since a bounded
self-adjoint operator is determined by the values of its associated quadratic form
restricted to the unit sphere, $E_{O_B}$ depends only on
$\mathcal{E}_B$ and $O_B$.

For any $\alpha\in\mathcal{I}$ and unit vectors
$\psi\in\mathcal{H}_A$ and $\phi\in\mathcal{H}_{\mathcal{E}_B}$,
Postulate \ref{IP} now yields
$$
\langle\psi\otimes\phi,
M^\alpha_{O_B}(\psi\otimes\phi)\rangle
=
\mathbb{P}_{\mathcal{E}_B}(O_B\,;\phi)
=
\langle\psi\otimes\phi,
(I_A\otimes E_{O_B})(\psi\otimes\phi)\rangle.
$$
By Lemma \ref{simpleTensors}, these two bounded quadratic forms therefore agree on all of $\mathcal{H}_A\otimes\mathcal{H}_{\mathcal{E}_B}$, and the polarization identity gives $M^\alpha_{O_B}=I_A\otimes E_{O_B}$.

Define $f_{O_B}(\rho)=\mathrm{Tr}(E_{O_B}\rho)$.
Since $0\le E_{O_B}\le I_B$, $f_{O_B}:\mathcal{D}(\mathcal{H}_{\mathcal{E}_B})\to[0,1]$ is an affine function. For every
$\alpha\in\mathcal{I}$ and
$\Psi\in S(\mathcal{H}_A\otimes\mathcal{H}_{\mathcal{E}_B})$,
\begin{align*}
\mathbb{P}_{\mathcal{J}_\alpha}
\big(\Omega_{\mathcal{E}_A^\alpha}\times O_B\,;\Psi\big)
&=\langle\Psi,(I_A\otimes E_{O_B})\Psi\rangle=\mathrm{Tr}\!\left(
E_{O_B}\,\mathrm{Tr}_A(\ket{\Psi}\!\bra{\Psi})
\right)=f_{O_B}\!\left(
\mathrm{Tr}_A(\ket{\Psi}\!\bra{\Psi})
\right),
\end{align*}
where the second equality follows from the defining
identity of the partial trace. Since $O_B$ was arbitrary
and $f_{O_B}$ is independent of $\alpha$, the family has
Property \ref{PT}.
\end{proof}

\subsection{Postulate \ref*{PT} implies Postulate \ref*{NS}}\label{sec:PTtoNS}

This subsection shows that, under the partial trace formalism, signaling via experimental choice for post-preparation independent experiments is impossible. This argument shares the core features of other partial-trace-based arguments for ``no-signaling'' 
(see \cite{ghirardi1980general, ghirardi1989cp, kennedy1995empirical, genovese2010photonic}).

\begin{lemma}\label{lem:PTtoNS} Postulate \ref{PT} implies Postulate \ref{NS}.
\end{lemma}

\begin{proof} Let $\{(\mathcal{J}_\alpha;\mathcal{E}_A^\alpha,\mathcal{E}_B)~|~\alpha\in\mathcal{I}\}$ be a preparation-compatible family of available composite HPMMs with common left-hand space of preparations $\mathcal{H}_A$, each of which describes a post-preparation independent composite experiment. Let $\alpha_1,\alpha_2\in\mathcal{I}$ and $O_B\in \Sigma_{\mathcal{E}_B}$. By Postulate \ref{PT}, there exists a single affine function $f_{O_B}$ (the same for every $\alpha \in \mathcal{I}$) such that for any $\Psi\in S(\mathcal{H}_A\otimes\mathcal{H}_{\mathcal{E}_B})$,
\[
\mathbb{P}_{\mathcal{J}_{\alpha_1}}(\Omega_{\mathcal{E}_A^{\alpha_1}}\!\times O_B\,;\Psi)=f_{O_B}\big(\mathrm{Tr}_{A}(\ket{\Psi}\!\bra{\Psi})\big)=\mathbb{P}_{\mathcal{J}_{\alpha_2}}(\Omega_{\mathcal{E}_A^{\alpha_2}}\!\times O_B\,;\Psi)\,.\qedhere
\]
\end{proof}

\subsection{Postulates \ref*{PT} and \ref*{IP} imply Postulate \ref*{SMQC}}\label{sec:PTtoSQ}

\begin{lemma}\label{lem:PTtoSQ} Assuming Postulate \ref{IP}, Postulate \ref{PT} implies Postulate \ref{SMQC}.
\end{lemma}

\begin{proof} Let $(\mathcal{J};\mathcal{E}_{\overline{u}},\mathcal{E}_B)$ be an available semiorthodox composite HPMM describing a post-preparation independent composite experiment, where $u\in S(\mathbb{C}^2)$. Let $u^\perp\in S(\mathbb{C}^2)$ with $\langle u,u^\perp\rangle=0$, let $O_B\in \Sigma_{\mathcal{E}_B}$, and let $f_{O_B}$ be the affine function whose existence is guaranteed by Postulate \ref{PT} for the (trivially preparation-compatible) one-element family $\{\mathcal{J}\}$. Let $\Psi\in S(\mathbb{C}^2\otimes\mathcal{H}_{\mathcal{E}_B})$ be expanded as $\Psi = \alpha\,u\otimes\psi + \beta\,u^\perp\otimes\phi$ with $\alpha,\beta\in\mathbb{C}$ and $\psi,\phi\in S(\mathcal{H}_{\mathcal{E}_B})$.

By Postulate \ref{IP}, $\mathbb{P}_{\mathcal{J}}\left(\Omega_{\mathcal{E}_{\overline{u}}}\!\times O_B\,;u\otimes \psi\right)=\mathbb{P}_{\mathcal{E}_B}\left(O_B\,; \psi\right)$, and by Postulate \ref{PT},
\begin{equation*}
\mathbb{P}_{\mathcal{J}}\big(\Omega_{\mathcal{E}_{\overline{u}}}\!\times O_B\,;u\otimes \psi\big)
=
f_{O_B}\big(\mathrm{Tr}_A(|u\otimes\psi\rangle\langle u\otimes \psi|)\big)
=
f_{O_B}(\ket{\psi}\!\bra{\psi})\,.
\end{equation*}
Putting the above two equations together, we have $f_{O_B}(\ket{\psi}\!\bra{\psi})=\mathbb{P}_{\mathcal{E}_B}\big(O_B\,; \psi\big)$. Applying an identical argument to $u^\perp\otimes \phi$ yields $f_{O_B}(\ket{\phi}\!\bra{\phi})=\mathbb{P}_{\mathcal{E}_B}\left( O_B\,; \phi\right)$. Applying Postulate \ref{PT} to the state $\Psi$, we have
\begin{align*}
\mathbb{P}_{\mathcal{J}}\left(\Omega_{\mathcal{E}_{\overline{u}}}\!\times O_B\,;\Psi\right)
&=
f_{O_B}\big(\mathrm{Tr}_A(\ket{\Psi}\!\bra{\Psi})\big)\\
&=f_{O_B}\big(\,\big(\bra{u}\otimes I\big)\ket{\Psi}\!\bra{\Psi}\big(\ket{u}\otimes I\big)+\big(\bra{u^\perp}\otimes I\big)\ket{\Psi}\!\bra{\Psi}\big(\ket{u^\perp}\otimes I\big)\,\big)\\
&=f_{O_B}\left(|\alpha|^2\,\ket{\psi}\!\bra{\psi}+|\beta|^2\,\ket{\phi}\!\bra{\phi}\right)
=|\alpha|^2\,f_{O_B}(\ket{\psi}\!\bra{\psi})+|\beta|^2\,f_{O_B}(\ket{\phi}\!\bra{\phi})\\
&=|\alpha|^2\,\mathbb{P}_{\mathcal{E}_B}(O_B\,;\psi)+|\beta|^2\,\mathbb{P}_{\mathcal{E}_B}(O_B\,;\phi)\,,
\end{align*}
as required, where the application of affinity is valid because $|\alpha|^2 + |\beta|^2 = \|\Psi\|^2 = 1$.
\end{proof}

\subsection{Postulate \ref*{WPL} implies Postulate \ref*{QF} in dimension at least 3}\label{sec:WPLtoQF}

In this subsection, we show that Postulate \ref{QF} can be proved using Postulate \ref{WPL} and techniques from Gleason's celebrated paper \cite{gleason1957}. The central tool we will need is Gleason's definition of a ``frame function'' on a Hilbert space $\mathcal{H}$. A \emph{frame function} is a function $f:S(\mathcal{H})\to \mathbb{R}$ such that for any orthonormal basis $\{e_k\}$ of $\mathcal{H}$, $\sum_k f(e_k) = w$ for some fixed real number $w$ called the \emph{weight} of $f$. A function $f:S(\mathcal{H})\to \mathbb{R}$ is \emph{regular} if there exists a self-adjoint operator $A$ such that $f(v) = \langle v, A v  \rangle$ for all unit vectors $v\in S(\mathcal{H})$. Gleason proved that any non-negative frame function on a Hilbert space of dimension at least 3 is regular \cite[Thm.~3.5]{gleason1957}. This is false in dimension 2, as illustrated by Gleason's quadrant construction \cite[p.~886]{gleason1957}.

In Section \ref{sec:PLtoQF}, we defined the parallelogram law for a function $f: S(\mathcal{H}) \to \mathbb{R}$. We will say that $f:S(\mathcal{H}) \to \mathbb{R}$ satisfies the \emph{weak} parallelogram law if the parallelogram law holds in the special case where $U$ is the Hadamard matrix and $\psi$ and $\phi$ are orthonormal, i.e., if $f\left(\psi \right) +  f\left(\phi \right)= f\left(\tfrac{1}{\sqrt{2}}(\psi + \phi)\right) + f\left(\tfrac{1}{\sqrt{2}}(\psi - \phi)\right)$ for all orthogonal $\psi, \phi\in S(\mathcal{H})$. It is straightforward to check that any regular function satisfies the (strong) parallelogram law. The lemmas below give a partial converse.

\begin{lemma}\label{frame} Let $\mathcal{H}$ be a finite-dimensional Hilbert space. If $f: S(\mathcal{H}) \to \mathbb{R}$ is phase invariant and satisfies the weak parallelogram law, then $f$ is a frame function.
\end{lemma}

\begin{proof} The proof proceeds by showing that phase changes and Hadamard transformations acting on a pair of basis vectors preserve the sum of $f$ over an orthonormal basis. Since every unitary change of orthonormal basis is a product of such transformations, the sum is then the same for every orthonormal basis.

Let $n=\dim(\mathcal{H})$. If $n=1$, then for any fixed unit vector $v$, every other unit vector is of the form $e^{i\theta}v$. It follows that $f$ is constant by phase invariance and hence trivially a frame function. We may therefore assume $n\ge 2$. Let $\{e_1,e_2,\dots, e_n\}$ be an orthonormal basis for $\mathcal{H}$, and let $U = \left(\begin{smallmatrix} a & b \\ c & d \end{smallmatrix}\right)$ be an arbitrary $2\times 2$ unitary matrix. For distinct $k,l\in\{1,\dots,n\}$, let $\widetilde{U}^{k,l}$ denote the linear transformation of $\mathcal{H}$ defined by $e_k\mapsto a e_k + c e_l$, $e_l \mapsto b e_k + d e_l$, and $e_j \mapsto e_j$ for $j\notin \{k,l\}$. Since the operator $\widetilde{U}^{k,l}$ maps the orthonormal basis $\{e_1,\dots,e_n\}$ to another orthonormal basis, $\widetilde{U}^{k,l}$ is unitary. If $U$ is of the form $\left(\begin{smallmatrix} e^{i\alpha} & 0 \\ 0 & e^{i\beta} \end{smallmatrix}\right)$ for $\alpha, \beta\in\mathbb{R}$, then 
$$
\sum_{j=1}^n f(e_j) = \sum_{j=1}^n f\big(\widetilde{U}^{k,l}(e_j)\big)\,,
$$
since for all $j$, there is a $\gamma_j \in\mathbb{R}$ such that $\widetilde{U}^{k,l}(e_j) = e^{i\gamma_j} e_j$, and $f$ is phase invariant. If $U$ is the Hadamard matrix $H=\frac{1}{\sqrt{2}}\left(\begin{smallmatrix} 1 & \phantom{-}1 \\ 1 & -1 \end{smallmatrix}\right)$, then by the weak parallelogram law,
\begin{align*}
\sum_{j=1}^n f(e_j)
&=
f(e_k) + f(e_l) +\!\sum_{j\notin \{ k,l\}} \! f(e_j) =  f\left( \frac{ e_k + e_l}{\sqrt{2}} \right) +
f\left( \frac{ e_k - e_l}{\sqrt{2}} \right)
+\!\sum_{j\notin \{ k,l\}} \!f(e_j)\\
&= f\big( \widetilde{U}^{k,l} e_k \big)
+
f\big( \widetilde{U}^{k,l} e_l \big)
+ \sum_{j\notin \{ k,l\}} f\big( \widetilde{U}^{k,l} e_j\big) =
\sum_{j=1}^n f\big(\widetilde{U}^{k,l} e_j\big).
\end{align*}
These two cases establish that for \emph{any} orthonormal basis $\mathcal{B}=\{e_1,\dots,e_n\}$, if an operator $V$ acts on $\mathcal{B}$ via a two-level transformation drawn from the set
$$
\mathcal{G}=\left\{  H \right\} \cup \left\{\left( \begin{smallmatrix} e^{i\alpha} & 0 \\ 0 & e^{i\beta} \end{smallmatrix}\right)~ \Big|~ \alpha,\beta\in\mathbb{R}\right\}\,,
$$
then $\sum_{j=1}^n f(e_j) = \sum_{j=1}^n f(V e_j)$.

Using the Euler-angle representation (cf.\ \cite[Thm.~4.1]{nielsen2010quantum}), any element of $\mathrm{U}(2)$ can be factored as
$$
e^{i\delta} \left(\begin{smallmatrix} e^{i\alpha} & 0 \\ 0 & e^{-i\alpha} \end{smallmatrix}\right) H  \left(\begin{smallmatrix} e^{i\beta} & 0 \\ 0 & e^{-i\beta} \end{smallmatrix}\right)H \left(\begin{smallmatrix} e^{i\gamma} & 0 \\ 0 & e^{-i\gamma} \end{smallmatrix}\right)\,.
$$
Since the global phase $e^{i\delta}$ acts as the diagonal matrix $\mathrm{diag}(e^{i\delta}, e^{i\delta})$, which belongs to $\mathcal{G}$, it follows that any element of $\mathrm{U}(2)$ can be written as a product of matrices in $\mathcal{G}$. Furthermore, any element of the unitary group $\mathrm{U}(n)$ is a product of two-level unitary transformations (see \cite[Sec.~4.5.1]{nielsen2010quantum}), and each such two-level factor can in turn be written as a product of two-level transformations whose non-trivial blocks lie in $\mathcal{G}$. If we identify each matrix in $\mathrm{U}(n)$ with the operator in $\mathcal{L}(\mathcal{H})$ it defines relative to $\mathcal{B}$, then any $Y \in \mathrm{U}(n)$ can be factored into a finite product $Y = Y_1 Y_2 \cdots Y_m$, where each operator $Y_k$ acts on the original basis $\mathcal{B}$ as a two-level unitary transformation whose non-trivial $2 \times 2$ block belongs to $\mathcal{G}$.

To show that the sum is invariant under the full unitary matrix $Y$, we use a telescoping sequence. For each $k \in \{1,\dots, m\}$, let $W_k = Y_1 \cdots Y_k$, with $W_0 = I$. We define the transformed basis at step $k-1$ as $\mathcal{B}_{k-1} = \{b_1^{(k-1)},\dots,b_n^{(k-1)}\}$ where $b_j^{(k-1)} = W_{k-1}e_j$. Now, define the conjugated operator $T_k = W_{k-1}Y_kW_{k-1}^\dagger$. Then $T_k b_j^{(k-1)}=(W_{k-1}Y_kW_{k-1}^\dagger) W_{k-1} e_j$ is equal to
$$
W_{k-1} Y_k e_j = W_{k-1}\sum_{i=1}^n (Y_k)_{ij}  e_i
 =\sum_{i=1}^n (Y_k)_{ij} W_{k-1} e_i = \sum_{i=1}^n (Y_k)_{ij} b_i^{(k-1)}\,.
$$
It follows that the matrix representing $T_k$ relative to the basis $\mathcal{B}_{k-1}$ is exactly $Y_k$. Since our earlier calculations apply to any orthonormal basis and any operator acting via a two-level unitary transformation whose non-trivial $2 \times 2$ block belongs to $\mathcal{G}$, we conclude that
\[
\sum_{j=1}^n f\big(b_j^{(k-1)}\big) = \sum_{j=1}^n f\big(T_k b_j^{(k-1)}\big)=\sum_{j=1}^n f\big(b_j^{(k)}\big)\,,
\]
since $T_k b_j^{(k-1)} = W_{k-1} Y_k e_j = W_k e_j = b_j^{(k)}$. By induction on $k$ with $\mathcal{B}_0=\mathcal{B}$, it follows that $\sum_{j=1}^n f(e_j) = \sum_{j=1}^n f(W_m e_j) = \sum_{j=1}^n f(Y e_j)$. Since every orthonormal basis of $\mathcal{H}$ is of the form $\{Ye_1,\dots,Ye_n\}$ for some $Y\in\mathrm{U}(n)$, the sum $\sum_j f(e_j)$ is independent of the choice of orthonormal basis, so $f$ is a frame function as required.
\end{proof}

\noindent As we will show in the proof of Lemma \ref{regularity} below, Gleason's arguments allow us to conclude that $f$ is regular, even in the infinite-dimensional case. To streamline the argument, we will use the following consequence of the proof of Gleason's Lemma 3.4 from \cite{gleason1957}.

\begin{lemma}[Strengthened form of Gleason's Lemma 3.4 \cite{gleason1957}]\label{lem:gleason} If a function $f:S(\mathcal H)\to\mathbb R$ is bounded on $S(\mathcal H)$, $\dim(\mathcal{H})\ge 2$, and $f$'s restriction to every two-dimensional subspace of $\mathcal H$ is regular, then $f$ is regular.
\end{lemma}

Although Gleason states Lemma \ref{lem:gleason} under the stronger assumption that $f$ is a global frame function, his proof relies only on global boundedness and regularity on 2-dimensional subspaces, never invoking the full frame-function hypothesis. For a fully self-contained proof of Lemma \ref{regularity} that does not rely on this strengthening of Gleason's lemma, see Appendix \ref{app:regularity}.

\begin{lemma}\label{regularity} Let $\mathcal{H}$ be a Hilbert space of dimension at least 3. If $f: S(\mathcal{H}) \to \mathbb{R}$ is a bounded, phase-invariant, non-negative function satisfying the weak parallelogram law, then $f$ is regular.
\end{lemma}

\begin{proof} Let $W$ be an arbitrary 2-dimensional subspace of $\mathcal{H}$, and let $V\subset \mathcal{H}$ be a finite-dimensional subspace of dimension at least $3$ that contains $W$ (which must exist since $\dim(\mathcal{H}) \ge 3$). Since $f$ satisfies the weak parallelogram law and is phase invariant on $S(\mathcal{H})$, the same holds for its restriction $f|_{S(V)}$. By Lemma \ref{frame}, $f|_{S(V)}$ is a frame function, and since it is non-negative and $\dim(V)\ge3$, it is regular by Gleason's theorem \cite[Thm.~3.5]{gleason1957}. Because $f|_{S(V)}$ is regular, there exists a self-adjoint operator $A_V$ on $V$ such that $f(v) = \langle v, A_V v \rangle$ for all $v \in S(V)$. For $w\in S(W)$, we can rewrite this as $f(w) = \langle w, (P_W A_V P_W) w \rangle$, where $P_W$ is the orthogonal projection onto $W$. The map $P_W A_V P_W$ is a self-adjoint operator on $W$, so it follows that $f|_{S(W)}$ is regular as well.

Thus the restriction of $f$ to any $2$-dimensional subspace is regular. Since $f$ is bounded on $S(\mathcal{H})$ and regular on every 2-dimensional subspace, $f$ is regular by Lemma \ref{lem:gleason}.
\end{proof}

\noindent We are now ready to prove Postulate \ref{QF} from the weak parallelogram law.

\begin{lemma}\label{lem:WPLtoQF} Any HPMM $\mathcal{E}$ with Property \ref{WPL} and $\dim(\mathcal{H}_{\mathcal{E}}) \ge 3$ has Property \ref{QF}. 
\end{lemma}

\begin{proof}Let $O\in \Sigma_{\mathcal{E}}$, let $f(\psi)=\mathbb{P}_{\mathcal{E}}(O\,;\psi)$, and let $Q_O\equiv Q_f$ be its homogeneous extension to $\mathcal{H}_{\mathcal{E}}$ as defined in Section~\ref{sec:PLtoQF}. Since $\mathcal{E}$ has Property \ref{WPL}, $f$ satisfies the weak parallelogram law. By the definition of an HPMM, it is phase invariant, and it is bounded and non-negative, since $0 \le \mathbb{P}_{\mathcal{E}}(O\,;\phi) \le 1$ for all $\phi$. Since $\dim(\mathcal{H}_{\mathcal{E}}) \ge 3$, we can apply Lemma \ref{regularity} to conclude that $f$ is regular, i.e., there exists a self-adjoint operator $A$ such that for all $\psi\in S(\mathcal{H}_{\mathcal{E}})$, $\mathbb{P}_{\mathcal{E}}(O\,;\psi)=\langle \psi, A \psi\rangle$. It follows that for any non-zero $\psi\in \mathcal{H}_{\mathcal{E}}$, 
$$
Q_O(\psi)
=
\|\psi\|^2\, \mathbb{P}_{\mathcal{E}}\left(O\,;\frac{\psi}{\|\psi\|}\right)
=
\|\psi\|^2 \, \left\langle \frac{\psi}{\|\psi\|}, A \frac{\psi}{\|\psi\|}\right\rangle
=
\langle \psi, A \psi \rangle.
$$
Since $Q_O(0) = 0 = \langle 0, A(0)\rangle$ as well, the functions $Q_O$ and $\psi \mapsto \langle \psi, A \psi\rangle$ agree as required. The function $(v,w)\mapsto\langle v,Aw\rangle$ is a sesquilinear form on $\mathcal{H}_{\mathcal{E}}$ whose associated quadratic form is $Q_O$, so the HPMM $\mathcal{E}$ has Property \ref{QF}.
\end{proof}

\subsection{Proof of Theorem \ref*{mainthm}}\label{sec:proofOfMainTheorem}

We have now proved all of the implications shown in Figure \ref{Implications}. From here, Theorem \ref{mainthm} can be read off by following the implication arrows and tracking their dependencies.

\begin{proof}[Proof of Theorem \ref{mainthm}] Tracing through the implications in Figure \ref{Implications}, Postulate \ref{NS} implies \ref{WNS}, and the conjunction of \ref{WNS}, \ref{EQC}, and \ref{SMQC} implies \ref{PL} (Lemma \ref{imp1}). The chain of implications \ref{PL}$\,\Rightarrow\,$\ref{QF}$\,\Rightarrow\,$\ref{POVM} (Lemmas \ref{PLtoQF} and \ref{lem:POVMeqQF}) completes one direction of the theorem. Conversely, \ref{POVM}$\,\Rightarrow\,$\ref{QF}$\,\Rightarrow\,$\ref{WSEP} (Lemmas \ref{lem:POVMeqQF} and \ref{densityReduction}). Since \ref{WSEP} together with \ref{IP} implies \ref{PT} (Lemma \ref{lem:WSEPandIPtoPT}), and \ref{PT} implies \ref{NS} (Lemma \ref{lem:PTtoNS}), the converse holds as required.
\end{proof}

Note that the two directions of Theorem \ref{mainthm} do not rely on the same set of postulates. The derivation of \ref{POVM} from \ref{NS} uses \ref{EQC} and \ref{SMQC}, but not \ref{IP}. The derivation of \ref{NS} from \ref{POVM} uses \ref{IP}, but neither \ref{EQC} nor \ref{SMQC}.

\section{Postulate \ref*{PM} is Independent of the Other Postulates}\label{sec:PMIndep}

Figure \ref{Implications} shows that Postulate \ref{IP}, combined with any single-system postulate other than \ref{WPL}, implies every other postulate in the figure except \ref{PM} and \ref{EQC}.\footnote{Postulates \ref{IP} and \ref{WPL} together fail to imply Postulate \ref{QF} for HPMMs $\mathcal{E}$ with $\dim(\mathcal{H}_{\mathcal{E}}) = 2$ (see Section \ref{sec:WPLtoQF}).} Postulate \ref{EQC} is independent for a simple reason. It is the only postulate that requires composite experiments to exist, so a model whose available HPMMs satisfy the single-system postulates, but which contains no available composites, satisfies every other postulate (the composite-system ones vacuously) while violating \ref{EQC}. Whether \ref{PM} is independent is less obvious. This section shows that it is. Since \ref{PM} holds in textbook NRQM whenever the prepared state is the only source of correlation between the two wings (see Section \ref{sec:PPI}), it follows that the other postulates in Figure \ref{Implications} do not fully characterize textbook NRQM.

I first show, in Section \ref{subsec:tsirelson}, that Postulates \ref{PM} and \ref{POVM} together imply Tsirelson's bound. Section \ref{subsec:PRmodel} then adapts Popescu and Rohrlich's example \cite{popescu1994quantum} to construct a model in which every postulate in Figure \ref{Implications} other than \ref{PM} holds. The model also violates Tsirelson's bound, so the correlations it permits are not merely non-multiplicative but stronger than any that textbook NRQM allows. Finally, Section \ref{subsec:uniqueness} shows that, given Postulate \ref{QF}, Postulate \ref{PM} ensures that available post-preparation independent composites are determined by their components, something the model shows cannot be established without it.

\subsection{Tsirelson's bound under Postulates \ref*{POVM} and \ref*{PM}}\label{subsec:tsirelson}

We begin with some notation. Suppose we have a composite HPMM $(\mathcal{J};\mathcal{E}_A,\mathcal{E}_B)$ and sets of outcomes $O_A\in \Sigma_{\mathcal{E}_A}$ and $O_B\in \Sigma_{\mathcal{E}_B}$, with complements ${O_A}^{c}\equiv\Omega_{\mathcal{E}_A}\setminus O_A$ and ${O_B}^{c}\equiv\Omega_{\mathcal{E}_B}\setminus O_B$. Let $\mathcal{O}_{\mathrm{eq}}(O_A,O_B)\equiv(O_A\times O_B) \cup  ({O_A}^{c}\times{O_B}^{c})$ and $\mathcal{O}_{\mathrm{op}}(O_A,O_B)\equiv (O_A\times {O_B}^{c}) \cup ({O_A}^{c}\times O_B)$. Define the correlation function
$$
\mathrm{C}_{\mathcal{J}}(O_A,O_B\,;\Psi) \equiv \mathbb{P}_{\mathcal{J}}\big(\mathcal{O}_{\mathrm{eq}}(O_A,O_B)\,;\Psi\big)-\mathbb{P}_{\mathcal{J}}\big(\mathcal{O}_{\mathrm{op}}(O_A,O_B)\,;\Psi\big)\,.
$$
Tsirelson's bound constrains sums of correlation functions for different post-preparation independent composite experiments. To state it precisely, let $\mathcal{E}_{A_1}$ and $\mathcal{E}_{A_2}$ be HPMMs such that $\mathcal{H}_{\mathcal{E}_{A_1}}\!=\mathcal{H}_{\mathcal{E}_{A_2}}\!\equiv\mathcal{H}_A$ and $\Sigma_{\mathcal{E}_{A_1}}=\Sigma_{\mathcal{E}_{A_2}}\equiv\Sigma_A$, and let $\mathcal{E}_{B_1}$ and $\mathcal{E}_{B_2}$ be HPMMs such that $\mathcal{H}_{\mathcal{E}_{B_1}}\!=\mathcal{H}_{\mathcal{E}_{B_2}}\!\equiv\mathcal{H}_B$ and $\Sigma_{\mathcal{E}_{B_1}}=\Sigma_{\mathcal{E}_{B_2}}\equiv \Sigma_B$. For $i,j\in\{1,2\}$, let $(\mathcal{J}_{ij};\mathcal{E}_{A_i},\mathcal{E}_{B_j})$ be a composite HPMM describing a post-preparation independent composite experiment. The \emph{Clauser--Horne--Shimony--Holt (CHSH) function} \cite{clauser1969proposed} associated with the tuple $(\mathcal{J}_{11},\mathcal{J}_{12},\mathcal{J}_{21},\mathcal{J}_{22})$ is the function $\mathrm{CHSH}_{\mathcal{J}_{11},\mathcal{J}_{12},\mathcal{J}_{21},\mathcal{J}_{22}}(O_A,O_B\,;\Psi)$ defined by
$$
\Big|\mathrm{C}_{\mathcal{J}_{11}}(O_A,O_B\,;\Psi) + \mathrm{C}_{\mathcal{J}_{12}}(O_A,O_B\,;\Psi)+
\mathrm{C}_{\mathcal{J}_{21}}(O_A,O_B\,;\Psi) -
\mathrm{C}_{\mathcal{J}_{22}}(O_A,O_B\,;\Psi)\Big|\,.
$$
Tsirelson's bound can then be stated as follows.

\begin{lemma}\label{Tsirelson} Let $\mathcal{J}_{11}, \mathcal{J}_{12},\mathcal{J}_{21}$ and $\mathcal{J}_{22}$ be available composite HPMMs, each describing a post-preparation independent composite experiment, as in the above definition of a CHSH function. Let $\Psi\in S(\mathcal{H}_A\otimes\mathcal{H}_B)$, $O_A\in \Sigma_A$, and $O_B\in \Sigma_B$. Then under Postulates \ref{POVM} and \ref{PM}, $\mathrm{CHSH}_{\mathcal{J}_{11},\mathcal{J}_{12},\mathcal{J}_{21},\mathcal{J}_{22}}(O_A,O_B\,;\Psi)\le 2\sqrt{2}$.
\end{lemma}

Before proving Lemma \ref{Tsirelson}, we will need two preliminary lemmas.
\begin{lemma}[Tsirelson \cite{cirel1980quantum}]\label{lem:tsirBoundLemma} Let $A_1,A_2,B_1$, and $B_2$ be self-adjoint operators defined on a Hilbert space $\mathcal{H}$ such that $-I\le A_i\le I$, $-I\le B_j\le I$, and $[A_i,B_j]=0$ for all $i,j\in\{1,2\}$, where $I$ is the identity operator on $\mathcal{H}$. Then $-2\sqrt{2}\,I \le A_1B_1+A_1B_2+A_2B_1-A_2B_2 \le 2\sqrt{2}\,I$.
\end{lemma}

\begin{proof} Using $[A_i,B_j]=0$, a direct expansion shows that
\begin{align*}
A_1B_1&+A_1B_2+A_2B_1-A_2B_2\\
&=\tfrac{1}{\sqrt{2}}\left(A_1^2+A_2^2+B_1^2+B_2^2\right) -\tfrac{1}{\sqrt{2}}\Big[\big(A_1-\tfrac{1}{\sqrt{2}}(B_1+B_2)\big)^2+\big(A_2-\tfrac{1}{\sqrt{2}}(B_1-B_2)\big)^2\Big]\,.
\end{align*}
Since each term in the bracket is the square of a self-adjoint operator and hence positive, and since $-I\le A_i,B_j\le I$ implies $A_i^2\le I$ and $B_j^2\le I$, we have
\begin{equation}\label{eq:ABUpperBound}
A_1B_1+A_1B_2+A_2B_1-A_2B_2\le\tfrac{1}{\sqrt2}\big(A_1^2+A_2^2+B_1^2+B_2^2\big)\le\tfrac{4}{\sqrt2}\,I=2\sqrt2\,I\,.
\end{equation}
The quadruple $(-A_1,-A_2,B_1,B_2)$ satisfies the same hypotheses, so applying Equation \eqref{eq:ABUpperBound} to it gives $-(A_1B_1+A_1B_2+A_2B_1-A_2B_2)\le 2\sqrt2\,I$. Together with Equation \eqref{eq:ABUpperBound}, this gives the two-sided bound in the statement of the lemma.
\end{proof}

The next lemma shows that, under Postulates \ref{PM} and \ref{QF}, the effect that any available post-preparation independent composite assigns to a product outcome $O_A\times O_B$ is the tensor product of the corresponding component effects.

\begin{lemma}\label{lem:opProd} Suppose Postulates \ref{PM} and \ref{QF} hold, and let $(\mathcal{J};\mathcal{E}_A,\mathcal{E}_B)$ be an available composite HPMM that describes a post-preparation independent composite experiment. Let $E_{\mathcal{J}}$ be the POVM associated with $\mathcal{J}$. Then $\mathcal{E}_A$ and $\mathcal{E}_B$ have Property \ref{POVM}, and their POVMs $E_{\mathcal{E}_A}$ and $E_{\mathcal{E}_B}$ satisfy $E_{\mathcal{J}}(O_A\times O_B) = E_{\mathcal{E}_A}(O_A)\otimes E_{\mathcal{E}_B}(O_B)$ for all $O_A\in\Sigma_{\mathcal{E}_A}$ and $O_B\in\Sigma_{\mathcal{E}_B}$.
\end{lemma}

\begin{proof} Let $\mathcal{H}_A\equiv\mathcal{H}_{\mathcal{E}_A}$, $\mathcal{H}_B\equiv\mathcal{H}_{\mathcal{E}_B}$, and $\Omega_B\equiv\Omega_{\mathcal{E}_B}$. Since Postulate \ref{QF} is equivalent to Postulate \ref{POVM} (as established in Section \ref{sec:POVM}), we may assume Postulate \ref{POVM}. Postulate \ref{POVM} applies to $\mathcal{J}$, since it is available. To obtain POVMs for the components, fix a unit vector $\phi_0\in\mathcal{H}_B$, let $V:\mathcal{H}_A\to \mathcal{H}_A\otimes \mathcal{H}_B$ be given by $V(\psi)= \psi \otimes \phi_0$, and define $E_{\mathcal{E}_A}(O_A)\equiv V^\dagger E_{\mathcal{J}}(O_A\times\Omega_B)V$.    Since $V$ is an isometry, $E_{\mathcal{E}_A}$ is a normalized POVM on $\mathcal{H}_A$, and
\begin{align*}
\langle\psi,E_{\mathcal{E}_A}(O_A)&\psi\rangle
=\langle\psi,V^\dagger E_{\mathcal{J}}(O_A\times\Omega_B)V\psi\rangle 
=
\langle \psi\otimes \phi_0,E_{\mathcal{J}}(O_A\times\Omega_B) \psi\otimes \phi_0\rangle
\\
&\stackrel{(\mathrm{\ref{POVM}})}{=}\mathbb{P}_{\mathcal{J}}(O_A\times\Omega_B\,;\psi\otimes\phi_0)\stackrel{(\mathrm{\ref{PM}})}{=}\mathbb{P}_{\mathcal{E}_A}(O_A\,;\psi)\,\mathbb{P}_{\mathcal{E}_B}(\Omega_B\,;\phi_0)=\mathbb{P}_{\mathcal{E}_A}(O_A\,;\psi)
\end{align*}
for all $\psi\in S(\mathcal{H}_A)$, so $\mathcal{E}_A$ has Property \ref{POVM}. The right-hand side of this display depends only on $\mathcal{E}_A$. Since a bounded operator on a complex Hilbert space is determined by its expectation values on unit vectors, $E_{\mathcal{E}_A}$ is the unique POVM representing $\mathcal{E}_A$, and in particular it does not depend on the choice of $\phi_0$ or of $\mathcal{J}$. The same construction with a unit vector $\psi_0\in\mathcal{H}_A$ yields $E_{\mathcal{E}_B}$.

Both $E_{\mathcal{J}}(O_A\times O_B)$ and $E_{\mathcal{E}_A}(O_A)\otimes E_{\mathcal{E}_B}(O_B)$ are self-adjoint operators on $\mathcal{H}_A\otimes \mathcal{H}_B$. We first verify that they yield the same expectation values on any product state $\Psi = \psi\otimes \phi$. This follows from Postulates \ref{PM} and \ref{POVM}, since
\begin{align*}
\langle \psi\otimes \phi&,E_{\mathcal{J}}(O_A\times O_B)\psi\otimes \phi\rangle 
\stackrel{(\mathrm{\ref{POVM}})}{=} 
\mathbb{P}_{\mathcal{J}}\big(O_A\times O_B\,;\psi\otimes \phi\big) \stackrel{(\mathrm{\ref{PM}})}{=}
\mathbb{P}_{\mathcal{E}_A}\big(O_A\,;\psi\big)\,\mathbb{P}_{\mathcal{E}_B}\big( O_B\,;\phi\big)\\
&\stackrel{(\mathrm{\ref{POVM}})}{=}
\left\langle \psi,E_{\mathcal{E}_A}(O_A)\psi\right\rangle \left\langle \phi,E_{\mathcal{E}_B}(O_B) \phi\right\rangle =\left\langle \psi\otimes \phi,E_{\mathcal{E}_A}(O_A)\!\otimes\! E_{\mathcal{E}_B}(O_B)\,(\psi\otimes\phi)\right\rangle \,.
\end{align*}
Both $\Psi\mapsto\langle \Psi, E_{\mathcal{J}}(O_A\times O_B) \Psi \rangle$ and $\Psi\mapsto\langle \Psi, E_{\mathcal{E}_A}(O_A)\otimes E_{\mathcal{E}_B}(O_B) \Psi \rangle$ are bounded quadratic forms on $\mathcal{H}_A\otimes\mathcal{H}_B$ that agree on all product vectors by the computation above. By Lemma \ref{simpleTensors}, they agree everywhere. Applying the polarization identity, we have the operator equality $E_{\mathcal{J}}(O_A\times O_B) = E_{\mathcal{E}_A}(O_A)\otimes E_{\mathcal{E}_B}(O_B)$.
\end{proof}

\begin{proof}[Proof of Lemma \ref{Tsirelson}] Let $E_{\mathcal{E}_{A_i}}$ and $E_{\mathcal{E}_{B_j}}$ be the normalized POVMs given by Lemma \ref{lem:opProd}, and let $I_A$ and $I_B$ be the identity operators on $\mathcal{H}_A$ and $\mathcal{H}_B$ respectively. Let $A_i$ equal $\big(E_{\mathcal{E}_{A_i}}(O_A)-E_{\mathcal{E}_{A_i}}({O_A}^{c})\big)\otimes I_B$ and $B_j$ equal $I_A\otimes \big(E_{\mathcal{E}_{B_j}}(O_B)-E_{\mathcal{E}_{B_j}}({O_B}^{c})\big)$. Since $E_{\mathcal{E}_{A_i}}({O_A}^{c})$ equals $I_A-E_{\mathcal{E}_{A_i}}(O_A)$, we have $A_i=\big(2E_{\mathcal{E}_{A_i}}(O_A)-I_A\big)\otimes I_B$, and $0\le E_{\mathcal{E}_{A_i}}(O_A)\le I_A$ gives $-I\le A_i\le I$. Likewise $-I\le B_j\le I$. Since $A_i$ and $B_j$ commute for any $i$ and $j$, the hypotheses of Lemma \ref{lem:tsirBoundLemma} hold.    Taking expectation values in the two-sided operator inequality of Lemma \ref{lem:tsirBoundLemma} gives, for any $\Psi\in S(\mathcal{H}_A\otimes\mathcal{H}_B)$,
\begin{equation}\label{eq:tsirBound} 
\big|\langle \Psi,A_1B_1 \Psi \rangle + \langle \Psi,A_1B_2 \Psi \rangle + \langle \Psi,A_2B_1 \Psi \rangle - \langle \Psi,A_2B_2 \Psi \rangle\big|\le 2\sqrt{2}\,.
\end{equation}
The result then follows from the identity
\begin{align*}
\mathrm{C}_{\mathcal{J}_{ij}}(O_A,O_B\,;\Psi)&~~~=~~~ \mathbb{P}_{\mathcal{J}_{ij}}\big(\mathcal{O}_{\textrm{eq}}(O_A,O_B)\,;\Psi\big)-\mathbb{P}_{\mathcal{J}_{ij}}\big(\mathcal{O}_{\textrm{op}}(O_A,O_B)\,;\Psi\big)\\
&~~~=~~~\mathbb{P}_{\mathcal{J}_{ij}}\big(O_A\times O_B\,;\Psi\big)+\mathbb{P}_{\mathcal{J}_{ij}}\big({O_A}^{c}\times {O_B
}^{c}\,;\Psi\big)\\
&\qquad\qquad-\mathbb{P}_{\mathcal{J}_{ij}}\big(O_A\times {O_B}^{c}\,;\Psi\big)-\mathbb{P}_{\mathcal{J}_{ij}}\big({O_A}^{c}\times O_B
\,;\Psi\big)\\
&\stackrel{(\mathrm{\ref{POVM}})}{=} 
\langle \Psi,E_{\mathcal{J}_{ij}}(O_A\times O_B)\Psi\rangle + \langle \Psi,E_{\mathcal{J}_{ij}}({O_A}^{c}\times {O_B}^{c})\Psi\rangle\\
&\qquad\qquad - \langle \Psi,E_{\mathcal{J}_{ij}}(O_A\times {O_B}^{c})\Psi\rangle - \langle \Psi,E_{\mathcal{J}_{ij}}({O_A}^{c}\times O_B)\Psi\rangle\\
&\,\stackrel{\mathrm{Lem.~\ref{lem:opProd}}}{=}\langle \Psi,E_{\mathcal{E}_{A_i}}(O_A)\otimes E_{\mathcal{E}_{B_j}}( O_B)\Psi\rangle + \langle \Psi,E_{\mathcal{E}_{A_i}}({O_A}^{c})\otimes E_{\mathcal{E}_{B_j}}({O_B}^{c})\Psi\rangle\\
&\qquad - \langle \Psi,E_{\mathcal{E}_{A_i}}(O_A)\otimes E_{\mathcal{E}_{B_j}}({O_B}^{c})\Psi\rangle - \langle \Psi,E_{\mathcal{E}_{A_i}}({O_A}^{c})\otimes E_{\mathcal{E}_{B_j}}( O_B)\Psi\rangle\\
&~~~\,=~~~ \left\langle \Psi,\big(E_{\mathcal{E}_{A_i}}(O_A)-E_{\mathcal{E}_{A_i}}({O_A}^{c})\big)\otimes \big(E_{\mathcal{E}_{B_j}}(O_B)-E_{\mathcal{E}_{B_j}}({O_B}^{c})\big)\,\Psi\right\rangle \\
&~~~\,=~~~ \left\langle \Psi,A_iB_j\,\Psi\right\rangle\,,
\end{align*}
which shows $\mathrm{CHSH}_{\mathcal{J}_{11},\mathcal{J}_{12},\mathcal{J}_{21},\mathcal{J}_{22}}(O_A,O_B\,;\Psi)$ is equal to the left-hand side of Equation \eqref{eq:tsirBound}.
\end{proof}

\subsection{A model in which Postulate \ref*{PM} fails}\label{subsec:PRmodel}

To show that Postulate \ref{PM} cannot be derived from the other postulates, it suffices to exhibit a model in which the other postulates hold and \ref{PM} fails. Since the postulates concern which HPMMs are available and which available composites describe post-preparation independent experiments, a model consists of a set $\mathcal{A}$ of available HPMMs, together with a specification of which available composite HPMMs $(\mathcal{J};\mathcal{E}_A,\mathcal{E}_B)$ describe post-preparation independent composite experiments.
 
The following model is built from two kinds of single-qubit experiment. The first are the orthodox HPMMs $\mathcal{E}_{\overline{u}}$ of Section \ref{sec:OrthodoxOPFs}, whose POVMs are defined by $E_{\overline{u}}(\{\mathrm{Y}\})=\ket{u}\!\bra{u}$ and $E_{\overline{u}}(\{\mathrm{N}\})=I-\ket{u}\!\bra{u}$. The second is the HPMM $\mathcal{E}_{\mathrm{PR}}\equiv\left(\{0,1\},2^{\{0,1\}},\mathbb{C}^2,\mathbb{P}_{E}\right)$, where $E(\{0\})=E(\{1\})=\frac{1}{2}I_2$ and $\mathbb{P}_E(O\,;\psi)=\langle \psi,E(O)\psi\rangle$. The outcome of $\mathcal{E}_{\mathrm{PR}}$ can be thought of as being determined by a fair coin flip that does not depend on the prepared state.
 
Next, for $i,j\in\{1,2\}$, let $E_{ij}$ be the POVM on $\{0,1\}^2$ acting on $\mathbb{C}^2\otimes\mathbb{C}^2$ defined by
\begin{align*}
E_{ij}(\{(a,b)\})=
\begin{cases}
\frac{1}{2} I_4 & a+b \equiv (i-1)(j-1) \pmod{2}\,, \\
0 & \mathrm{otherwise,}
\end{cases}
\end{align*}
extended to all subsets of $\{0,1\}^2$ by additivity. For each $(i,j)$, exactly two of the four outcomes satisfy the parity condition, so $E_{ij}(\{0,1\}^2)=I_4$ and $E_{ij}$ is a normalized POVM. These are Popescu and Rohrlich's correlations \cite{popescu1994quantum}, with inputs $x=i-1$ and $y=j-1$. The outputs satisfy $a\oplus b=xy$, and each output on its own is uniformly random. Let
$$
\mathcal{J}_{ij}\equiv\left(\{0,1\}^2,2^{\{0,1\}^2},\mathbb{C}^2\otimes\mathbb{C}^2,\mathbb{P}_{E_{ij}}\right)\,,
$$
which is a composite of $\mathcal{E}_{\mathrm{PR}}$ with itself. Note that each $\mathcal{J}_{ij}$ has the same components.

Finally, given $\overline{u}\in S^2$ and an HPMM $\mathcal{X}$ with finite outcome space, discrete $\sigma$-algebra, and POVM $E_{\mathcal{X}}$, let $\mathcal{E}_{\overline{u}}\otimes\mathcal{X}$ denote the composite of $\mathcal{E}_{\overline{u}}$ and $\mathcal{X}$ whose POVM $E_{\overline{u}\otimes\mathcal{X}}$ is given on singletons by $E_{\overline{u}\otimes\mathcal{X}}(\{(s,t)\})=E_{\overline{u}}(\{s\})\otimes E_{\mathcal{X}}(\{t\})$ and extended by additivity. To satisfy Postulate \ref{EQC}, the model includes the composite $\mathcal{E}_{\overline{u}}\otimes\mathcal{X}$ for every $\overline{u}\in S^2$ and every available HPMM $\mathcal{X}$. Since these composites are themselves available, the construction must be iterated. Let $\mathcal{A}_0$ consist of the HPMMs $\mathcal{E}_{\overline{u}}$ for all $\overline{u}\in S^2$, the HPMM $\mathcal{E}_{\mathrm{PR}}$, and the HPMMs $\mathcal{J}_{ij}$. For $n\ge 0$, let $\mathcal{A}_{n+1}=\mathcal{A}_n\cup\left\{\mathcal{E}_{\overline{u}}\otimes\mathcal{X}~\big|~\overline{u}\in S^2,~\mathcal{X}\in\mathcal{A}_n\right\}$, and let $\mathcal{A}=\bigcup_{n}\mathcal{A}_n$. By induction on $n$, every element of $\mathcal{A}$ has a finite outcome space with a discrete $\sigma$-algebra and is defined by a normalized POVM, so each $\mathcal{E}_{\overline{u}}\otimes\mathcal{X}$ is well defined. We take the available HPMMs to be the elements of $\mathcal{A}$, and we take the composite HPMMs describing post-preparation independent composite experiments to be exactly $(\mathcal{J}_{ij};\mathcal{E}_{\mathrm{PR}},\mathcal{E}_{\mathrm{PR}})$ for $i,j\in\{1,2\}$ and $(\mathcal{E}_{\overline{u}}\otimes\mathcal{X};\mathcal{E}_{\overline{u}},\mathcal{X})$ for $\overline{u}\in S^2$ and $\mathcal{X}\in\mathcal{A}$.
 
\begin{theorem}\label{thm:PMindep}
In the model just described, every postulate in Figure \ref{Implications} other than Postulate \ref{PM} holds, and Postulate \ref{PM} fails.
\end{theorem}
 
\begin{proof}
Since every postulate other than \ref{PM} appearing in Figure \ref{Implications} follows from the conjunction of Postulates \ref{POVM}, \ref{IP}, and \ref{EQC}, it suffices to verify Postulates \ref{POVM}, \ref{IP}, and \ref{EQC} for the model described above, and to show that Postulate \ref{PM} fails. Postulate \ref{POVM} follows from the model's definition, since every element of $\mathcal{A}$ is defined by a normalized POVM.
 
To show Postulate \ref{EQC} holds, let $\overline{u}\in S^2$ and let $\mathcal{X}\in\mathcal{A}$. Since $\mathcal{A}$ is the union of the sets $\mathcal{A}_n$, $\mathcal{X}\in\mathcal{A}_n$ for some $n$. From the definition of $\mathcal{A}_{n+1}$, $\mathcal{E}_{\overline{u}}\otimes\mathcal{X}\in\mathcal{A}_{n+1}\subseteq\mathcal{A}$. It follows that $(\mathcal{E}_{\overline{u}}\otimes\mathcal{X};\mathcal{E}_{\overline{u}},\mathcal{X})$ is an available composite HPMM, and by construction it describes a post-preparation independent composite experiment.

To see that Postulate \ref{IP} holds, first consider the composites $\mathcal{E}_{\overline{u}}\otimes\mathcal{X}$. Since the identity $E_{\overline{u}}(\{\mathrm{Y}\})+E_{\overline{u}}(\{\mathrm{N}\})=I$ holds, we have, for all $\psi\in S(\mathbb{C}^2)$, $\phi\in S(\mathcal{H}_{\mathcal{X}})$, and $O_B\in\Sigma_{\mathcal{X}}$,
\begin{align*}
\mathbb{P}_{\mathcal{E}_{\overline{u}}\otimes\mathcal{X}}\big(\{\mathrm{Y},\mathrm{N}\}\times O_B\,;\psi\otimes\phi\big)
&=\big\langle\psi\otimes\phi,\big(I\otimes E_{\mathcal{X}}(O_B)\big)\,\psi\otimes\phi\big\rangle\\
&=\langle \psi, \psi\rangle \big\langle \phi, E_{\mathcal{X}}(O_B) \phi\big\rangle=\mathbb{P}_{\mathcal{X}}(O_B\,;\phi)\,.
\end{align*}
For the composites $\mathcal{J}_{ij}$, we must show that $\mathbb{P}_{\mathcal{J}_{ij}}(\{0,1\}\times O_B\,;\psi\otimes\phi)=\mathbb{P}_{E}(O_B\,;\phi)$ for all $\psi,\phi\in S(\mathbb{C}^2)$ and $O_B\subseteq\{0,1\}$. For each $b\in\{0,1\}$, exactly one $a\in\{0,1\}$ satisfies $a+b\equiv(i-1)(j-1)\pmod{2}$, so $E_{ij}(\{0,1\}\times\{b\})=\frac{1}{2}I_4$. Since also $E(\{b\})=\frac{1}{2}I_2$, both sides equal $\frac{1}{2}$ when $O_B=\{b\}$. The cases $O_B=\varnothing$ and $O_B=\{0,1\}$ are immediate.

Finally, Postulate \ref{PM} fails for the composite HPMM $(\mathcal{J}_{11};\mathcal{E}_{\mathrm{PR}},\mathcal{E}_{\mathrm{PR}})$, which describes a post-preparation independent composite experiment. This is because for any $\psi,\phi\in S(\mathbb{C}^2)$, $\mathbb{P}_{\mathcal{J}_{11}}(\{(0,1)\}\,;\psi\otimes\phi)=0$, whereas ${\mathbb{P}_{E}(\{0\}\,;\psi)\,\mathbb{P}_{E}(\{1\}\,;\phi)=\tfrac{1}{4}}$.
\end{proof}

The above model also violates Tsirelson's bound. To see this, take $\mathcal{E}_{A_1} = \mathcal{E}_{A_2} = \mathcal{E}_{B_1} = \mathcal{E}_{B_2} = \mathcal{E}_{\mathrm{PR}}$ in the definition of the CHSH function, and let $O_A = O_B = \{0\}$. Since each $E_{ij}(\{(a,b)\})$ is either $\frac{1}{2}I_4$ or $0$, every elementary outcome has probability $\frac{1}{2}$ or $0$ in $\mathcal{J}_{ij}$ for any $\Psi$. For $(i,j) \ne (2,2)$, the outcomes $(0,0)$ and $(1,1)$ have probability $\frac{1}{2}$, while $(0,1)$ and $(1,0)$ have probability $0$, so $\mathrm{C}_{\mathcal{J}_{ij}}(\{0\},\{0\}\,;\Psi)=1$. For $(i,j)=(2,2)$, the roles are reversed, and $\mathrm{C}_{\mathcal{J}_{22}}(\{0\},\{0\}\,;\Psi)=-1$. Hence, for every $\Psi\in S(\mathbb{C}^2\otimes\mathbb{C}^2)$, $\mathrm{CHSH}_{\mathcal{J}_{11},\mathcal{J}_{12},\mathcal{J}_{21},\mathcal{J}_{22}}\big(\{0\},\{0\}\,;\Psi\big)=|1+1+1-(-1)|=4>2\sqrt{2}$. By Lemma \ref{Tsirelson}, this gives a second proof that Postulate \ref{PM} fails in the model.

\subsection{Postulate \ref*{PM} and the uniqueness of composites}\label{subsec:uniqueness}

The model of Section \ref{subsec:PRmodel} also shows that, without Postulate \ref{PM}, a composite HPMM need not be determined by its components. Indeed, the composites $\mathcal{J}_{11}$ and $\mathcal{J}_{22}$ have the same components but different OPFs. Under Postulates \ref{PM} and \ref{QF}, however, an available composite describing a post-preparation independent composite experiment is determined by its components, as the following lemma shows.

\begin{lemma}\label{lem:decompUniqueness} Under Postulates \ref{PM} and \ref{QF}, given HPMMs $\mathcal{E}_A$ and $\mathcal{E}_B$, there is at most one HPMM $\mathcal{J}$ such that $(\mathcal{J};\mathcal{E}_A,\mathcal{E}_B)$ is an available composite HPMM describing a post-preparation independent composite experiment.
\end{lemma}

\begin{proof}
Let $(\mathcal{J}_1;\mathcal{E}_A,\mathcal{E}_B)$ and $(\mathcal{J}_2;\mathcal{E}_A,\mathcal{E}_B)$ be available composite HPMMs, each describing a post-preparation independent composite experiment. Since the components fix the outcome set, outcome $\sigma$-algebra, and Hilbert space of a composite, it suffices to show that $\mathbb{P}_{\mathcal{J}_1}=\mathbb{P}_{\mathcal{J}_2}$.

By Postulate \ref{QF} and the results of Section \ref{sec:POVM}, $\mathcal{J}_1$ and $\mathcal{J}_2$ have associated POVMs, each uniquely determined by its OPF. Applying Lemma \ref{lem:opProd} to $\mathcal{J}_1$ and to $\mathcal{J}_2$ shows that $\mathcal{E}_A$ and $\mathcal{E}_B$ have POVMs $E_{\mathcal{E}_A}$ and $E_{\mathcal{E}_B}$, which do not depend on the composite, and that $E_{\mathcal{J}_1}(O_A\times O_B) = E_{\mathcal{E}_A}(O_A)\otimes E_{\mathcal{E}_B}(O_B) = E_{\mathcal{J}_2}(O_A\times O_B)$ for every $O_A\in\Sigma_{\mathcal{E}_A}$ and $O_B\in\Sigma_{\mathcal{E}_B}$. Hence for each $\Psi\in S(\mathcal{H}_{\mathcal{E}_A}\otimes\mathcal{H}_{\mathcal{E}_B})$,
$$
\mathbb{P}_{\mathcal{J}_1}(O_A\times O_B\,;\Psi)
=\langle\Psi,E_{\mathcal{J}_1}(O_A\times O_B)\Psi\rangle
=\langle\Psi,E_{\mathcal{J}_2}(O_A\times O_B)\Psi\rangle
=\mathbb{P}_{\mathcal{J}_2}(O_A\times O_B\,;\Psi),
$$
so the probability measures $O\mapsto\mathbb{P}_{\mathcal{J}_i}(O\,;\Psi)$ agree on all measurable rectangles. These form a $\pi$-system generating $\Sigma_{\mathcal{E}_A}\otimes\Sigma_{\mathcal{E}_B}$, so by Dynkin's $\pi$-$\lambda$ lemma \cite[Thm.~3.2]{billingsley2012probability} the two measures agree on the entire product $\sigma$-algebra. Since $\Psi$ was arbitrary, $\mathbb{P}_{\mathcal{J}_1}=\mathbb{P}_{\mathcal{J}_2}$, and hence $\mathcal{J}_1=\mathcal{J}_2$.
\end{proof}

\appendix
\section{A Stand-Alone Proof of Lemma \ref*{regularity}}\label{app:regularity}

In this appendix, we give a complete proof of Lemma \ref{regularity} without relying on Gleason's proof of Lemma 3.4, which assumes that $f$ is a global frame function. At least a priori, this is stronger than $f$ being bounded and regular on every two-dimensional subspace, though Gleason's proof relies only on these weaker hypotheses.\footnote{Gleason's relations (1)--(4) in the proof of Lemma 3.4 are derived solely from the restrictions of $f$ to two-dimensional subspaces, and the remainder of the proof uses only the boundedness of $f$, which follows easily from the fact that $f$ is a non-negative frame function.} The following proof invokes only Gleason's Theorem 3.5, applied to finite-dimensional subspaces.

\begin{proof}[Proof of Lemma \ref{regularity}] Let $\mathcal{H}$ be a complex Hilbert space of dimension at least $3$, and let
$f : S(\mathcal{H}) \to \mathbb{R}$ be bounded, phase invariant, and non-negative, and suppose $f$ satisfies the weak parallelogram law. Since $f$ satisfies the weak parallelogram law and is phase invariant on $S(\mathcal{H})$, for any finite-dimensional subspace $V\subseteq \mathcal{H}$, the restriction $f|_{S(V)}$ has the same properties. From Lemma \ref{frame}, we can conclude that $f|_{S(V)}$ is a frame function. If $V$ has dimension at least 3, then since $f|_{S(V)}$ is non-negative, we can also conclude from Gleason's theorem \cite[Thm.~3.5]{gleason1957} that $f|_{S(V)}$ is regular, i.e., there exists a self-adjoint operator $A_V:V\to V$ such that $f(v) = \langle v, A_V v\rangle$ for all $v\in S(V)$. As in the proof of Lemma \ref{lem:POVMeqQF}, since the quadratic form $v \mapsto \langle v, A_V v\rangle$ is homogeneous of degree $2$, its values on the unit sphere determine it on all of $V$. The polarization identity then uniquely determines the operator $A_V$.

For any $W \subseteq V$, the identity $(f|_{S(V)})|_{S(W)} = f|_{S(W)}$ then gives $f(v)=\big\langle v, (P_W A_V P_W) v\big\rangle$ for all $v\in S(W)$, where $P_W$ is the orthogonal projection onto $W$. Since the compression $P_W A_V P_W$ (viewed as an operator on $W$) is self-adjoint, the uniqueness argument above implies it is the unique self-adjoint operator on $W$ whose quadratic form evaluates to $f|_{S(W)}$.

We now need to show that $P_W A_V P_W$ depends only on $W$ and not on the choice of the superspace $V$. To see this, let $V_1$ and $V_2$ be two such superspaces, and let $U\supseteq V_1\cup V_2$ be any finite-dimensional subspace. Applying Gleason's theorem to $f|_{S(U)}$ gives a unique self-adjoint operator $A_U$ such that for all $u\in S(U)$, $f(u)=\langle u, A_U u\rangle$. Since for $v\in S(V_i)\subset S(U)$, we have $\langle v, A_{V_i} v\rangle = f(v)=\langle v, (P_{V_i} A_U P_{V_i})v\rangle$, uniqueness implies that $A_{V_i} = P_{V_i} A_U P_{V_i}$. Because $P_W P_{V_i} = P_{V_i} P_W = P_W$ for any $W \subseteq V_i$, we have the composition property:
$$
P_W A_{V_1} P_W 
=
P_W (P_{V_1} A_U P_{V_1}) P_W
=
P_W A_U P_W
=
P_W (P_{V_2} A_U P_{V_2}) P_W
=
P_W A_{V_2} P_W\,,
$$
as required. This allows us to unambiguously define $A_W \equiv P_W A_V P_W$ for any finite-dimensional subspace $W$ of $\mathcal{H}$, by choosing any finite-dimensional superspace $V \supseteq W$ with $\dim(V) \ge 3$. Given any set of elements $\{v_1,\dots,v_n\}\subset\mathcal{H}$, let $\mathrm{sp}(v_1,\dots,v_n)$ denote the subspace spanned by $\{v_1,\dots,v_n\}$. Let $B(v,w) = \langle v, A_{\mathrm{sp}(v,w)} w\rangle$. We claim that $B$ is sesquilinear. It follows from the self-adjointness of $A_{\mathrm{sp}(v,w)}=A_{\mathrm{sp}(w,v)}$ that
$$
B(v,w)=\langle v, A_{\mathrm{sp}(v,w)} w\rangle
=
\langle A_{\mathrm{sp}(v,w)}v,  w\rangle
=
\langle A_{\mathrm{sp}(w,v)}v,  w\rangle
=
\langle w,A_{\mathrm{sp}(w,v)}v\rangle^\ast
=
B(w,v)^\ast\,,
$$
so conjugate symmetry is satisfied. Homogeneity in the second argument follows from the identity $\mathrm{sp}(v,w)=\mathrm{sp}(v,\lambda w)$ for $\lambda \ne 0$, which implies 
$$
B(v,\lambda w)
=
\langle v, A_{\mathrm{sp}(v,\lambda w)} \lambda w \rangle
=
\langle v, A_{\mathrm{sp}(v, w)} \lambda w \rangle
=
\lambda \langle v, A_{\mathrm{sp}(v, w)} w \rangle
=
\lambda B(v,w)\,,
$$
and the identity holds trivially when $\lambda=0$, since then both sides vanish. Homogeneity in the first argument follows immediately from conjugate symmetry, since $B(\lambda v,w)=B(w,\lambda v)^\ast=\big(\lambda B(w,v)\big)^\ast=\lambda^\ast B(v,w)$. For additivity in the first argument, note that if $W\subseteq U$ are finite-dimensional subspaces, then choosing $V\supseteq U$ with $\dim(V)\ge 3$ gives $A_W=P_WA_VP_W=P_W(P_UA_VP_U)P_W=P_WA_UP_W$. Hence for any $u,w\in W$, we can pass freely between operators in inner products:
\begin{equation}\label{eq:compressionTrick}
\langle u, A_U w\rangle = \langle P_W u, A_U P_W w\rangle = \langle u, (P_W A_U P_W)w\rangle = \langle u, A_W w\rangle\,.
\end{equation}
Applying Equation \eqref{eq:compressionTrick} to the superspace $U = \mathrm{sp}(u,v,w)$ and the three relevant subspaces ${\mathrm{sp}(u+v,w)}$, $\mathrm{sp}(u,w)$, and $\mathrm{sp}(v,w)$, we obtain:
\begin{align*}
B(u + v,w)
&~=
\langle u + v, A_{\mathrm{sp}(u+v, w)} w\rangle \\
&\stackrel{\eqref{eq:compressionTrick}}{=}
\langle u + v, A_{\mathrm{sp}(u,v, w)} w\rangle =
\langle u, A_{\mathrm{sp}(u,v, w)} w\rangle
+
\langle v, A_{\mathrm{sp}(u,v, w)} w\rangle\\ 
&\stackrel{\eqref{eq:compressionTrick}}{=}
\langle u, A_{\mathrm{sp}(u, w)} w\rangle
+
\langle v, A_{\mathrm{sp}(v, w)} w\rangle =
B(u,w) + B(v,w).
\end{align*}
Additivity in the second argument follows easily from conjugate symmetry:
$$
B(u,v+w)=B(v+w,u)^\ast=\big(B(v,u)+B(w,u)\big)^\ast=B(u,v)+B(u,w)\,.
$$

Next, we show that $B$ is bounded. Let $M = \sup_{u\in S(\mathcal{H})} f(u)$, which is finite by assumption. By the definition of $B$ and $A_{\mathrm{sp}(x)}$, for any non-zero $x \in \mathcal{H}$ we have 
\begin{equation}\label{eq:bdiagbound}
B(x,x) = \langle x, A_{\mathrm{sp}(x)} x \rangle = \|x\|^2 f(x/\|x\|) \le M \|x\|^2\,,
\end{equation}
and $B(0,0)=0$. For any non-zero $v,w \in \mathcal{H}$, choose a phase $\eta \in \mathbb{C}$ such that $|\eta|=1$ and $B(\eta\, v, w) = \eta^* B(v,w) = |B(v,w)|$. Using the sesquilinearity of $B$, its non-negativity on the diagonal, and the diagonal bound \eqref{eq:bdiagbound},
\begin{align*}
4|B(v,w)| &= 4\,\mathrm{Re}\,B(\eta\, v,w) = B(\eta\, v+w,\eta \,v+w)-B(\eta\, v-w,\eta\, v-w)\\
&\le B(\eta\, v+w,\eta\, v+w) \le M\|\eta \,v+w\|^2 \le 2M\big(\|v\|^2+\|w\|^2\big)\,.
\end{align*}
Replacing $v$ and $w$ by $v'=(\|w\|/\|v\|)^{1/2}\,v$ and $w'=(\|v\|/\|w\|)^{1/2}\,w$, which leaves $B(v,w)$ unchanged, gives $|B(v,w)|\le\tfrac12M\big(\|v'\|^2+\|w'\|^2\big)=M\|v\|\|w\|$. Since the same bound holds trivially if either $v=0$ or $w=0$, boundedness is established. 

By the Riesz representation theorem for bounded sesquilinear forms, there exists a unique bounded operator $A$ such that $B(v,w) = \langle v, A w \rangle$. Since $B(v,w)=B(w,v)^\ast$, $A$ is self-adjoint, and since $f(v)= B(v,v) = \langle v, A v \rangle$ for $v\in S(\mathcal{H})$, $f$ is regular as required.
\end{proof}

\makeatletter
\immediate\write\@auxout{\string\citation{apsrev42Control}}
\makeatother

\bibliographystyle{apsrev4-2}
\phantomsection
\addcontentsline{toc}{section}{References}
\bibliography{PLBibliographyAPS}

\end{document}